\RequirePackage{fix-cm}
\documentclass[opre,sglanonrev]{informs4}
\OneAndAHalfSpacedXI
\MANUSCRIPTNO{}
\RequirePackage{xurl}

\makeatletter
\let\OPREsmall\small
\renewcommand{\small}{\OPREsmall\let\@currsize\small}
\let\OPREnormalsize\normalsize
\renewcommand{\normalsize}{\OPREnormalsize\let\@currsize\normalsize}
\let\OPREfootnotesize\footnotesize
\renewcommand{\footnotesize}{\OPREfootnotesize\let\@currsize\footnotesize}
\let\theoremstyle\@undefined
\let\proof\@undefined
\let\endproof\@undefined
\newcommand{\appendix}{\par\setcounter{section}{0}\setcounter{subsection}{0}%
  \gdef\thesection{\@Alph\c@section}%
  \@for\OPREthmenv:=theorem,lemma,proposition,corollary,assumption,definition,remark,example\do{%
    \setcounter{\OPREthmenv}{0}%
    \expandafter\xdef\csname the\OPREthmenv\endcsname{A.\noexpand\arabic{\OPREthmenv}}%
    \expandafter\xdef\csname theH\OPREthmenv\endcsname{appendix.\noexpand\arabic{\OPREthmenv}}}}
\makeatother

\usepackage[margin=1in]{geometry}
\usepackage{amsmath,amssymb,amsthm,mathtools}
\usepackage{thmtools}
\usepackage[english]{babel}
\usepackage{bbm}

\usepackage{bm}
\usepackage{enumitem}
\usepackage{xcolor}
\usepackage{hyperref}
\usepackage{graphicx}
\usepackage{mathrsfs}

\graphicspath{{images/}}
\usepackage{subcaption}
\usepackage{placeins}
\usepackage[authoryear]{natbib}
\bibpunct[, ]{(}{)}{,}{a}{}{,}
\def\bibfont{\small}

\usepackage[ruled,vlined,linesnumbered]{algorithm2e}
\SetKwComment{Comment}{$\triangleright$\ }{}
\usepackage{dsfont}

\hypersetup{colorlinks=true,linkcolor=blue!55!black,citecolor=blue!55!black,filecolor=blue!55!black,urlcolor=blue!55!black}

\newtheorem{theorem}{Theorem}
\newtheorem{lemma}{Lemma}
\newtheorem{proposition}{Proposition}
\newtheorem{corollary}{Corollary}

\theoremstyle{definition}
\newtheorem{definition}{Definition}
\newtheorem{remark}{Remark}
\newtheorem{example}{Example}

\usepackage[capitalize]{cleveref}
\crefname{theorem}{Theorem}{Theorems}
\Crefname{theorem}{Theorem}{Theorems}

\crefname{lemma}{Lemma}{Lemmas}
\Crefname{lemma}{Lemma}{Lemmas}

\crefname{proposition}{Proposition}{Propositions}
\Crefname{proposition}{Proposition}{Propositions}

\crefname{corollary}{Corollary}{Corollaries}
\Crefname{corollary}{Corollary}{Corollaries}

\crefname{assumption}{Assumption}{Assumptions}
\Crefname{assumption}{Assumption}{Assumptions}

\crefname{definition}{Definition}{Definitions}
\Crefname{definition}{Definition}{Definitions}

\crefname{remark}{Remark}{Remarks}
\Crefname{remark}{Remark}{Remarks}

\crefname{example}{Example}{Examples}
\Crefname{example}{Example}{Examples}

\newcommand{\E}{\mathbb E}
\newcommand{\N}{\mathbb N}
\newcommand{\pro}{\mathbb P}

\newcommand{\init}{\mathrm{init}}
\newcommand{\newb}{\mathrm{new}}
\newcommand{\LRU}{\mathrm{LRU}}
\newcommand{\one}{\mathds{1}}

\newcommand{\Exp}{\text{Exp}}

\makeatletter
\renewcommand{\theARTICLEABSTRACT}{%
  \HOOKb
  \vspace*{18pt}%
  \begin{minipage}[t]{\textwidth}\parindent1em
    \ABSfont
    \noindent\theABSTRACT\endgraf
    \vskip5pt
    \theFUNDING
    \theKEYWORDS
    \theSUBJECTCLASS
    \theAREAOFREVIEW
    \theMSCCLASS
    \theORMSCLASS
    \if@BLINDREV\else\theHISTORY\fi
    \noindent\hrulefill
  \end{minipage}%
  \vspace*{0pt}%
}
\AtBeginDocument{\captionsetup{font=small}\captionsetup[sub]{font=small}%
  \captionsetup[subfigure]{font=small}%
  \hypersetup{pdfauthor={Chutong Gao, David Simchi-Levi, Heyuan Yao},pdftitle={\theTITLE}}}
\AtBeginDocument{%
  \let\OPRElabel\label
  \renewcommand{\label}[1]{\@ifundefined{OPREseen@#1}{%
    \expandafter\gdef\csname OPREseen@#1\endcsname{1}\OPRElabel{#1}}{}}%
  \let\cref@label\label}
\makeatother

\TITLE{Improving Hit Ratios in LLM Serving:\texorpdfstring{\\}{ }From Descriptive Characterization to Prescriptive Optimization}
\RUNTITLE{Improving Hit Ratios in LLM Serving}
\renewcommand{\footnoterule}{%
  \ifnum\value{page}=1\relax
    \kern-3pt\hrule width .4\columnwidth height .4pt\kern 2.6pt
  \fi}

\renewcommand{\theARTICLETOP}{}
\renewcommand{\theARTICLETITLE}{%
  \begingroup
  \centering\fontsize{15.5}{20}\sffamily\bfseries\selectfont
  \theTITLE\HD{0}{15}\par
  \endgroup}
\LRHSecondLine{\fs.7.9.\rm arXiv preprint}
\RRHSecondLine{\fs.7.9.\rm arXiv preprint}

\ARTICLEAUTHORS{%
  \AUTHOR{Chutong Gao\textsuperscript{1}\qquad David Simchi-Levi\textsuperscript{1}\qquad Heyuan Yao\textsuperscript{2}}
  \AFF{\textsuperscript{1}Institute for Data, Systems, and Society, Massachusetts Institute of Technology;
    \mbox{Mitch Daniels School of Business, Purdue University}}
  \AFF{\textsuperscript{2}Department of Industrial Engineering and Management Sciences, Northwestern University}
}
\RUNAUTHOR{Gao, Simchi-Levi, and Yao}
\ABSTRACT{
Workflows in modern large language model (LLM) serving systems exhibit prolonged conversations consisting of recursive prompts and responses. Due to the autoregressive nature of LLMs, a new response must be generated based on KV states of the entire conversation history. Therefore, the hit ratio, i.e., the fraction of KV caches reused directly from the memory among all KV caches in one conversation, is a crucial metric that governs the performance and economics of LLM serving. Describing the hit ratio as a function of system primitives is highly nontrivial due to the complex system dynamics; hence prescriptively optimizing the hit ratio is even more challenging. We approach the two problems by formulating the system as a stochastic multi-turn conversation model (MCM). We first characterize the hit ratio of the MCM under the widely used least-recently-used (LRU) policy through a mean-field framework. We prove that as the workflow arrival rate and the memory capacity scale proportionally, the hit ratio converges to an explicit limit. Our descriptive characterization reveals that LRU implicitly leads to a possibly suboptimal workflow quota allocation when the workflow is a mixture of heterogeneous classes. Therefore, we further propose a family of workflow-aware quota LRU (WAQ-LRU) policies that make quota allocation an explicit control decision, among which we identify the optimal quota vector, yielding a policy that dominates LRU in the mean-field limit. Experiments with Qwen3 models on Ascend NPUs fed by real-world conversation datasets validate our resulting hit ratio empirical estimator, and show substantial gains of WAQ-LRU over LRU in both hit ratio and time-to-first-token. Our results provide both descriptive and prescriptive insights into the understanding and improvement of multi-turn LLM serving.

}

\begin{document}
\maketitle
\begingroup
\renewcommand{\thefootnote}{}
\footnotetext[0]{Emails: \href{mailto:chutong@mit.edu}{\texttt{chutong@mit.edu}},
  \href{mailto:dslevi@mit.edu}{\texttt{dslevi@mit.edu}},
  \href{mailto:heyuanyao2029@u.northwestern.edu}{\texttt{heyuanyao2029@u.northwestern.edu}}.}
\endgroup

\section{Introduction}
\label{sec: intro}

On April 24, 2026, DeepSeek launched its V4-series large language model (LLM). Two days after the release, following an already huge 75\% discount on the first day, the company further slashed the API price for cached (``hit'') input tokens by 90\%, gaining a significant edge over its rivals \citep{aju2026deepseekpricecut, reuters2026deepseekpricecut}. One core reason behind the aggressive price reduction is tied to the high \textit{\textbf{hit ratio}} that DeepSeek's service infrastructure managed to maintain when facing real customer workflows, which was reported to achieve approximately 95\% for V4-Pro and 91\% for V4-Flash \citep{funda2026deepseekcache}. In this paper, we study the characterization and optimization problems of the hit ratio in multi-turn LLM serving systems through the lens of queueing theory. 

\subsection{The Importance of Achieving High Hit Ratios}
As LLMs enter large-scale commercial deployment, the major expenses of AI companies have shifted from model training to model serving (inference), with OpenAI's serving spending already surpassing training (\$4 billion versus \$3 billion) in 2024 \citep{thebatch2024openai}, and LLM serving projected to account for roughly two-thirds of all AI compute in 2026 amid \$400--\$450 billion in global AI data center capital 
 \citep{deloitte2025tmtpredictions}. Optimizing the operational efficiency of LLM serving has therefore become a critical business priority. 

Within modern LLM serving, \textit{multi-turn conversations} have become the dominant workload because complex real-world tasks, such as interactive problem solving and autonomous coding agents that iteratively edit and test code through repeated tool calls, rarely conclude in a single prompt and instead require sequential back-and-forth exchanges \citep{github2025codingagent}. Within one \textit{conversation}, the LLM must use all conversation history (i.e., all previous prompt-response pairs plus the incumbent new prompt) for new response generation, and such multi-turn LLM calls create prolonged and growing conversation history. A key strategy to hedge against the correspondingly growing compute is to store and reuse the generated key-value (KV) caches in the high-bandwidth memory (HBM), such that a new arriving prompt within the same conversation can avoid repeated prefill (i.e., the process of calculating the KV states of input tokens) and thus massively reduce time-to-first-token (TTFT) and inference costs\footnote{For example, DeepSeek reports that context caching reduces the TTFT of repetitive 128K-token inputs from 13s to 500ms, and consequently, it charges a prompt that successfully hits the cached session $1/10$ of the original price \citep{deepseek2024contextcache}.}. 
 
 However, the total memory capacity is finite, so the KV caches of some conversations must be evicted when the total memory usage of all active conversations reaches the capacity. Therefore, not all new prompts can reuse the KV caches of their conversation history from the memory.
  Following such system dynamics, we define the \textit{KV cache hit ratio} (hereafter abbreviated as hit ratio) as the average fraction of KV caches reused directly from the memory among all KV caches of a conversation in a stationary system. A new prompt that finds its previous content's KV caches in the memory (``hits'') incurs much lower inference costs than one that ``misses''. Therefore, the hit ratio serves as a crucial first-order metric that governs the performance and the economics of LLM serving platforms, as we have demonstrated in the example of DeepSeek V4's release. \Cref{tab: API cache prices} further summarizes the API price differences of mainstream LLMs for hit inputs vs. missed inputs. These tenfold or greater price gaps render the hit ratio optimization a primary business concern in LLM serving operations.

\begin{table}[htbp]
    \centering
    \caption{Official API input-token prices as of November 16, 2026, in USD per million tokens.}
    \label{tab: API cache prices}
    \small
    \setlength{\tabcolsep}{6pt}
    \renewcommand{\arraystretch}{1.15}
    \begin{tabular}{@{}lrrr@{}}
        \hline
        Model & Cache miss & Cache hit & Price reduction \\
        \hline
        DeepSeek V4.1-Flash-1017 \citep{deepseek2026apipricing} & 0.30 & 0.006 & 98\% \\
        Claude Fable 5.2 \citep{anthropic2026apipricing} & 8.00 & 0.20 & 97.5\% \\
        GPT-6.1 Sol \citep{openai2026apipricing} & 2.00 & 0.10 & 95\% \\
        Qwen4-Max \citep{alibaba2026apipricing} & 1.50 & 0.075 & 95\% \\
        Gemini 4 Hydrogen \citep{google2026apipricing} & 1.50 & 0.075 & 95\% \\
        Kimi K3 \citep{moonshot2026contextcache} & 3.00 & 0.30 & 90\% \\
        \hline
    \end{tabular}
    % \par\smallskip
\end{table}

\subsection{Descriptive Hit Ratio Characterization via a Mean-Field Framework}

The first step towards hit ratio optimization is to describe it analytically.
Currently, the hit ratio in multi-turn LLM serving is often obtained empirically by running simulation or evaluated on the fly when serving the real customer workloads.
Either way evaluates a particular combination of workflow parameters, and the manager needs to rerun the experiments when there is a distribution shift of the workflows. Moreover, numerical experiments alone do not provide an analytical description explaining how the different moving parts of the workflow distribution (e.g., the workflow arrival rate, the interarrival time between two consecutive turns, and the prompt and response lengths) impact the hit ratio. This motivates our first research question:

\begin{quote}
\textit{Research question 1: How to analytically characterize the hit ratio as a function of the workload distribution and the memory capacity?}
\end{quote}

Characterizing the hit ratio as a function of the system primitives is highly nontrivial due to the complex system dynamics. % A conversation's KV-caches grow across turns, its next access occurs after a random interarrival time, and its history competes for memory with both newly arriving and revisiting conversations. Consequently, the probability of a cache hit depends on the interaction between the workload distribution and the eviction policy. 
To this end, we model multi-turn LLM serving systems as a stochastic multi-turn conversation model (MCM), and we first focus on the widely adopted least-recently-used (LRU) eviction policy. Under LRU, a new prompt arrival  refreshes the recency of the conversation it belongs to, and the conversation with the oldest access time is evicted first when the memory usage is full. LRU is widely used in practical LLM serving engines such as vLLM and SGLang, both choosing LRU as their default eviction policy \citep{vllm2026prefixcaching, sglang2026eviction}. Moreover, \cite{liu2026fancy} find that LRU remains competitive with more sophisticated eviction policies on production LLM traces. 

We model the workflow distribution as parametrized by four system primitives: the Poisson workflow arrival rate, the random number of total turns, the random interarrival time of consecutive turns, and the stochastic increments of KV blocks in each turn. We analyze a mean-field asymptotic regime in which the conversation arrival rate and the memory capacity scale proportionally. For MCM operating under LRU, we prove that the hit ratio converges to an explicit mean-field limit determined by an equation for the mean-field time-to-live (TTL; see \cref{thm: Hit Ratio Convergence Full}). This characterization provides an analytical description of how the four moving parts of the workflow distribution and the memory capacity jointly determine hit ratio performance. Based on this, we propose an empirical hit ratio estimator using the empirical workload distribution as the input, whose accuracy is validated by our real LLM serving experiments (see details in \cref{ssec: LRU estimator real practice}).

\subsection{Prescriptive Hit Ratio Optimization by Workflow-Aware Quota LRU}

\paragraph{\textbf{Our Hit Ratio Characterization Reveals Hidden Room for Improvement.}}
LLM workflows in reality are often composed of heterogeneous classes, e.g., coding agent workflows that involve recursive tool calling, single-shot workflows from daily inquiries, and research workflows with long interarrival times. The differences in the workflow distribution are evident in production LLM conversation traces \citep{liu2026fancy}. Under such heterogeneous workflows, our analysis of the hit ratio under LRU shows that LRU induces an implicit quota allocation across workflows (see \cref{thm: implicit workflow allocation under LRU}), which can be suboptimal. This observation motivates our second research question:

\begin{quote}
\textit{Research question 2: Under heterogeneous workflows, can we propose a new eviction policy that improves the hit ratio compared to LRU?}
\end{quote}

We answer the above question by proposing a new family of \textit{workflow-aware quota LRU} (WAQ-LRU) policies. Each policy partitions the memory capacity into workflow-specific quotas and maintains an independent LRU list within each quota for each workflow. An arriving prompt can therefore evict cached content only within its own workflow. This policy makes the allocation of memory across workflows an explicit control decision. We use our previous mean-field characterization to express the limiting hit ratio as a function of the quota vector that parametrizes one WAQ-LRU policy (\cref{thm: mixed workflow WAQ LRU mean field}). Then we identify the optimal quota vector that maximizes the limiting hit ratio. The corresponding WAQ-LRU policy achieves a limiting hit ratio no less than that of LRU, with strict improvement whenever the implicit LRU allocation is suboptimal. Numerical experiments in \cref{sec: LLM-Inference Experiment and Empirical Results} showcase significant hit ratio and TTFT improvements of the optimal WAQ-LRU over LRU across different Qwen models and production LLM datasets.

\subsection{Main Contributions}

To summarize, we make three main contributions.
\begin{enumerate}
    \item \textit{\textbf{Modeling and Hit Ratio Characterization via a Mean-Field Framework.}} We formulate multi-turn LLM serving as a stochastic queueing model (the MCM) and establish the first analytical characterization of the hit ratio in the mean-field limit, yielding a tractable approximation for finite-capacity systems.

    \item \textit{\textbf{Hit Ratio Optimization.}} Our characterization exposes a hidden structural inefficiency induced by LRU under heterogeneous workflows. Specifically, when the workflow is a mixture of heterogeneous classes, we identify LRU's implicit quota allocation for each workflow. We then propose a new family of WAQ-LRU policies that turns the quota allocation into an explicit control decision. Within this family, we find the quota allocation that maximizes the limiting hit ratio, and thus obtain a policy that dominates LRU in the mean-field limit.

    \item \textit{\textbf{Real LLM Serving Experiments.}} Building on our mean-field analysis, we construct a practical empirical hit ratio estimator. We conduct LLM serving experiments with Qwen3-4B, Qwen3-8B, and Qwen3-32B models on Ascend 910B2 NPUs, using the ShareGPT and WildChat conversation datasets, to validate the accuracy of the estimator and demonstrate the improvements of WAQ-LRU over LRU in both hit ratio and TTFT.
\end{enumerate}

\paragraph{\textbf{Organization.}}
\Cref{sec: prior work} reviews prior work. \Cref{sec: Model Settings} introduces the model, the LRU policy, and the mean-field regime. \Cref{sec: main results} establishes the mean-field limit for the hit ratio of MCM under a single workflow and under LRU. \Cref{sec: mixed workflow optimization} characterizes LRU's implicit workflow quota allocation, develops WAQ-LRU and its quota optimization. \Cref{sec: LLM-Inference Experiment and Empirical Results} presents the practical hit ratio estimators and LLM serving experiments. \Cref{sec: conclusion} concludes. Proofs and additional experimental results are provided in the appendices.

\subsection[Notation]{Notation}
We write $\mathbb N=\{0,1,2,\ldots\}$ and $\mathbb N_+=\{1,2,\ldots\}$ for the sets of nonnegative and positive integers, respectively. We denote the set of real numbers by $\mathbb R$ and its nonnegative subset by $\mathbb R_+ = \mathbb R \cap [0,+\infty)$. We use $\overset{\mathbb P}{\to}$ for convergence in probability. For a sequence of random variables $X_N$ and positive constants $a_N$, $N=1, 2, \ldots$, we write $X_N=O_{\mathbb P}(a_N)$ if $X_N/a_N$ is bounded in probability, i.e., for every $\varepsilon>0$ there exist $K > 0$ and $N_0 \in\mathbb N_+$ such that $\mathbb P(|X_N|>K a_N)<\varepsilon$ for  every $N\ge N_0$.

We write $\mathbb D=\mathbb D([0,\infty),\mathbb R)$ for the space of real-valued, right-continuous functions with left limits, endowed with the Skorokhod $J_1$ topology; see \citet[Section~11.5]{whitt2002limits} for details. For a continuous deterministic function $d$ and a sequence of stochastic processes $\{D_N: N\ge 1\}$ in $\mathbb D$, $D_N\overset{\mathbb P}{\to} d$ in $\mathbb D$ is equivalent to $\sup_{0\le t\le T}|D_N(t)-d(t)|\overset{\mathbb P}{\to}0$ for every $T<\infty$.\par

\section{Prior Work}
\label{sec: prior work}

\subsection{Caching Systems and KV-Cache Management for Multi-Turn LLM Serving}

Heterogeneous resource requirements have been frequently observed in capacity-constrained computing systems. In cloud computing resource-allocation systems, jobs may require different amounts of resources \citep{harchol2022multiserver,grosof2023reset,yao2026throughput}. Particularly, in conventional caching systems, a related form of heterogeneity arises from highly variable object sizes, which have been observed in production workloads and incorporated into cache admission, miss-ratio analysis, and resource allocation~\citep{berg2020cachelib,berger2017adaptsize,berger2018practical,berger2018robinhood}. Recent caching systems reveal several workload characteristics that are particularly relevant to multi-turn LLM serving. In LLM serving, KV states are typically managed at the block level \citep{kwon2023pagedattention}, while block reuse may span multiple conversation turns and even multiple storage tiers \citep{deepseek2024contextcache,gao2024cachedattention,qin2025mooncake}. Production traces further show substantial heterogeneity in KV-cache reuse intervals, reuse probabilities, and reusable prefix lengths across requests \citep{wang2025kvcache}. Practical serving platforms therefore employ mechanisms such as LRU- or priority-based eviction, prefix matching, and KV-aware routing under limited cache capacity \citep{nvidia2025kvcachereuse}. Moreover, a distinctive feature of multi-turn LLM workloads is that the cacheable blocks associated with the same conversation grow across successive accesses, rather than remaining a fixed-size object. 

\cite{liu2026fancy} analyze production traces of conversational and agentic LLM workloads and compare 14 eviction algorithms across HBM and larger memory pools. They find that LRU remains competitive on the evaluated traces and relate this performance to the relatively regular timing of requests within active sessions, which makes recency informative about reuse. They also identify substantial variation in session sizes and miss costs, showing that a larger block hit ratio need not yield greater prefill-computation savings. These findings motivate our focus on LRU and on the dependence of block hit ratios on workload distributions and cache capacity.

\subsection{Analytical Models for LLM Serving}
\label{ssec: Analytical Models for LLM Inference Serving}

In the past two years, a growing stream of analytical research studies LLM inference through scheduling, queueing, and resource-allocation models. 
These studies mostly model each arrival as an individual single-turn inference request, consisting of one prefill phase followed by one autoregressive decode phase. 
The corresponding KV-cache footprint evolves only during the active lifetime of that request and is treated primarily as a resource constraint governing admission, batching, scheduling, stability, or serving cost
\citep{ao2025optimizing,ao2026service,li2026asymptotic,nie2026queueing,wang2025llm,yao2026efficiency}.
However, these models do not capture the multi-turn setting in which a conversation becomes inactive after one response, retains reusable KV states in the cache, and may return with a new prompt after a random interarrival time. Some related analytical work \citep{bu2026tackling,chen2026load,song2026analytical} also focuses on the decode phase in large-scale or disaggregated serving architectures, where growing KV states and uncertain completion times complicate load balancing across workers. These works develop routing and provisioning methods to improve decode throughput in large-scale or disaggregated serving systems.

In contrast, relatively few analytical studies characterize cache-hit behavior during the prefill phase of multi-turn conversations. 
A cache miss requires the historical context to be recomputed during prefill, which is substantially more expensive than directly reusing resident KV states from memory.
The hit ratio is therefore a key performance metric that affects TTFT and reflects how effectively the limited memory capacity preserves reusable KV states.

\subsection{Previous Analyses of LRU and Hit Ratio}
\label{ssec: Previous Analyses of LRU and Hit Ratio}

We now review prior analytical work on the LRU cache eviction policy. Existing analyses have mainly focused on web and session caching systems, where a fixed population of cacheable objects is repeatedly requested according to independent, stationary, or temporally correlated request processes. A common objective is to characterize the miss or hit probability of LRU in large-cache regimes, often through asymptotic approximations that replace the complicated LRU ordering dynamics with a more tractable working-set or TTL representation. One of the earliest asymptotic analyses of LRU can be found in \cite{fagin1977asymptotic}, which showed that the expected LRU miss ratio can be asymptotically approximated by the miss ratio of a corresponding working-set model. This working-set model characterizes cache residency through a fixed window of recent page references. A closely related line of work analyzes LRU through a deterministic TTL that approximates the random eviction time of a cached object. \cite{che2002hierarchical} introduced this TTL approximation for hierarchical web caches, which was later widely referred to as the Che approximation. A mathematical explanation for its accuracy was provided in \cite{fricker2012versatile} by relating the eviction of a tagged object to the cumulative number of distinct objects requested after its last access. This TTL decoupling principle was further extended to general renewal traffic, multiple cache management policies, and interconnected caches \citep{martina2014unified}. Subsequent work established stronger asymptotic foundations for TTL and TTL approximations under more general request processes, including shot-noise request models with temporal locality \citep{leonardi2017modeling} and independent stationary and ergodic request processes \citep{jiang2018convergence}. In addition, \cite{gast2017ttl} connected TTL approximations for LRU variants to deterministic transient mean-field ordinary differential equation (ODE) limits under the independent reference model (IRM), whose fixed points recover the proposed TTL approximations.

\cite{zhang2026tail} study cache eviction for multi-turn LLM conversations using a stochastic conversation model (SCM) with Poisson conversation arrivals, Poisson prompt arrivals during each conversation's active lifetime, and unobserved exponential termination times. They propose Tail-Optimized LRU and establish optimality of its generalized form for a tail-excess-latency objective. Under homogeneous turn-arrival rates, their result also recovers LRU optimality for average latency. Our MCM permits more general distributions of conversation lengths and prompt interarrival times, and we derive a mean-field characterization of the block-weighted hit ratio under LRU.

Our analysis builds on the TTL representation used in the classical LRU literature. We represent LRU eviction through the cumulative workload that moves ahead of a tagged conversation session and identify a deterministic TTL in the mean-field limit. However, multi-turn LLM conversations introduce several challenges that are absent from the classical object-based caching models reviewed above. First, cacheable KV blocks have finite reuse lifetimes. In our model, a block can only be reused by subsequent turns of its corresponding conversation. Once the conversation terminates, its blocks may still occupy memory but can never generate another cache hit. Among the classical caching models reviewed above, the shot-noise model of \cite{leonardi2017modeling} partially captures such a feature. Moreover, the server does not know in advance when a conversation will terminate, which contrasts with classical IRM settings with a fixed population of objects and time-invariant request probabilities. Second, the reusable history associated with a conversation grows over time, so a later access can touch a nondecreasing number of KV blocks. As a result, the displacement process must track the amount of KV block workload moved ahead of the tagged session, while the size of the tagged session itself is also random and increases with the turn state. Finally, request-level hit probability is insufficient for quantifying cache effectiveness in LLM serving because different prompts may reuse substantially different numbers of historical KV blocks. Inspired by the byte hit ratio used in a content delivery network \citep{berger2017adaptsize}, we therefore define our hit ratio as the fraction of requested KV-block workload served directly from memory.

Beyond characterizing the hit ratio, the heterogeneity of LLM workflows raises the question of how a shared cache should be allocated across request classes.
In chip multiprocessors, where multiple applications share a last-level cache, \cite{qureshi2006utility} observe that LRU allocates capacity by demand rather than by marginal benefit. 
Hence, they propose utility-based cache partitioning as a runtime heuristic and validate it through simulation. 
For multiple content providers with disjoint files and Poisson requests, \cite{chu2016allocating} show under the Che approximation that a shared LRU cache is equivalent to a particular partition. 
As a result, optimizing the partition sizes can weakly improve the aggregate utility. 
Motivated by these insights, we prove an analogous implicit quota allocation property of LRU in the MCM, and propose our WAQ-LRU policy, which we further prove to weakly dominate LRU in the mean-field limit. 
Unlike the fixed-size and permanently reusable objects in these classical caching works, the KV caches in our setting grow with conversation sessions, have finite reuse lifetimes, and are reused through prefix matching. 
Moreover, our quota optimization is built on a proven mean-field limit, rather than a direct application of the Che approximation.

\section{The Model}
\label{sec: Model Settings}

In this section, we introduce the model in detail. 
In \cref{ssec: Generalized SCM}, we define the mixed-workflow MCM. 
We then define caching policies under our consideration and introduce workflow-specific and system-wide hit ratios in \cref{ssec: The  Hit Ratio}. 
We define the whole-session LRU policy in \cref{ssec: The LRU Policy}, and specify the mean-field regime in \cref{ssec: Our Mean-Field Model}.
Finally, in \cref{ssec: tagged-conversation Displacement Process and Characteristic Time}, we introduce some supporting concepts which will be extensively used to analyze the mean-field behavior of our system.

\subsection{Mixed-Workflow Multi-Turn Conversation Model}
\label{ssec: Generalized SCM}

A \textit{workflow} represents a statistically homogeneous class of conversations, whereas the aggregate workflow offered to a server can be a mixture of finitely many heterogeneous workflows. We first define a generic workflow and its primitives, and then introduce the conversation arrival process from $J$ different workflows. The model conditions in this subsection apply throughout the paper. 

\paragraph{\textbf{A Generic Workflow.}} In LLM practice, KV states are managed in fixed-size blocks. For instance, 128 tokens can be arranged in one block, which has a constant memory usage. A workflow $c$ is characterized by 
\begin{equation*}
    \mathcal W_c:=\left(\lambda_c, W_c,M_c,\{Z_{c,\ell}\}_{\ell\geq1}\right),
    \qquad c\in\mathcal C:=\{1,\ldots,J\};
\end{equation*}
the meaning of each term is specified as follows. Conversations belonging to workflow $c$ arrive according to a Poisson process of rate $\lambda_c>0$. The $\mathbb N_{+}$-valued random variable $M_c$ represents the number of turns in a workflow-$c$ conversation.
Within a conversation, the interarrival times (IATs) between consecutive turns form an independent and identically distributed (i.i.d.) sequence with generic term $W_c$ and cumulative distribution function (CDF) $F_c$, where we assume $\E[W_c]<\infty$ and $F_c$ is Lipschitz-continuous with constant $L_F$ for any workflow $c\in \mathcal C$; mathematically, we write
\begin{equation} \label{eq: assumption Lipschitz}
	F_c \in Lip(L_F).
\end{equation}
Finally, the $\mathbb N_+$-valued random variable $Z_{c,\ell}$ denotes the turn-$\ell$ block increment, $\ell = 1, 2, \ldots$. Note that the increments in the sequence $\{ Z_{c, \ell} \}_{\ell \ge 1}$ need not be independent nor identically distributed across turns. The Poisson conversation arrival process, the entire IAT sequence, $M_c$, and $\{Z_{c,\ell}\}_{\ell \ge 1}$ are mutually independent. Across workflows and conversations, primitive conversation trajectories are independent conditional on their workflow labels. 

The Poisson conversation arrival process to the system is a mixture of $J$ mutually independent Poisson conversation arrival processes from the $J$ different workflows. Each arriving conversation carries an observable workflow label $c \in \mathcal C$.
%, \textcolor{blue}{which may be obtained from the serving endpoint, application identifier, routing rule, or another workload classifier.
%The label is used only by the workflow-aware policy introduced in \cref{sec: mixed workflow optimization}; LRU does not use it when making eviction decisions.}

We consider a server with HBM KV-block capacity $N\in\mathbb N_+$.
For a workflow-$c$ conversation, the \textit{session size} immediately after its turn-$\ell$ access is distributed as $R_{c,\ell}:=\sum\limits_{r=1}^{\ell}Z_{c,r}.$
The \textit{terminal session size} of a workflow-$c$ conversation is distributed as $B_c:=R_{c,M_c}$. 
When $J=1$, we suppress the workflow index $c$ and write $\lambda_0,W,F,M, \{Z_\ell\}_{\ell\ge 1}, \{ R_\ell \}_{\ell\ge 1},B$, with $\mu_W:=\E[W]$. In particular,
\begin{equation}
    \label{equL The history session size R_j}
    R_j:=\sum\limits_{\ell=1}^{j}Z_\ell,
    \qquad B:=R_M.
\end{equation}
For every workflow $c\in\mathcal C$, the terminal session size $B_c$ has a finite moment generating function in a positive neighborhood of $0$. Specifically, there exists $\theta_0>0$ such that 
\begin{equation} \label{eq:assumption B}
	\E[e^{\theta_0 B_c}]<\infty, \qquad c\in\mathcal C.
\end{equation}

The abstract increment $Z_j$ covers a wide range of LLM serving architectures in practice.
We now use the following prefill-decode disaggregation architecture as an example.

\begin{example}[Prefill-Decode Disaggregation as One Realization of the Increments]

Consider a prefill-decode disaggregated system such as Mooncake \citep{qin2025mooncake}, and we focus on a server on the prefill side. 
Let $X_j$ be the number of KV blocks produced from the turn-$j$ prompt on the prefill side, and let $Y_j$ be the number of KV blocks corresponding to the turn-$j$ response. The current Mooncake architecture supports transferring KV blocks from the prefill side to the decode side, but not in the opposite direction. 
So a turn-$(j-1)$ response's KV blocks are immediately deleted after decoding, and this response must be reprefilled together with the next arriving prompt at turn $j$. Then we have $Z_1=X_1,$ and $Z_j=X_j+Y_{j-1}, j\geq2$. Therefore
$R_j=\sum\limits_{\ell=1}^{j}X_\ell+\sum\limits_{\ell=1}^{j-1}Y_\ell.$
\Cref{fig: KV session increments example} illustrates this realization. 

Other LLM serving architectures have different realizations of $Z_j$. For example, response KV blocks may be written back to the prefill-side worker and retained for later reuse, or prefill and decode may run on the same processing unit. 
These cases are accommodated by redefining the block increment $Z_j$ at turn $j$.
% Our theoretical analysis below depends only on the sequence $\{Z_j\}_{j\geq1}$ and not on its decomposition into prompt-side and response-side blocks.

\begin{figure}[ht]
    \centering
    \includegraphics[width=0.75\linewidth]{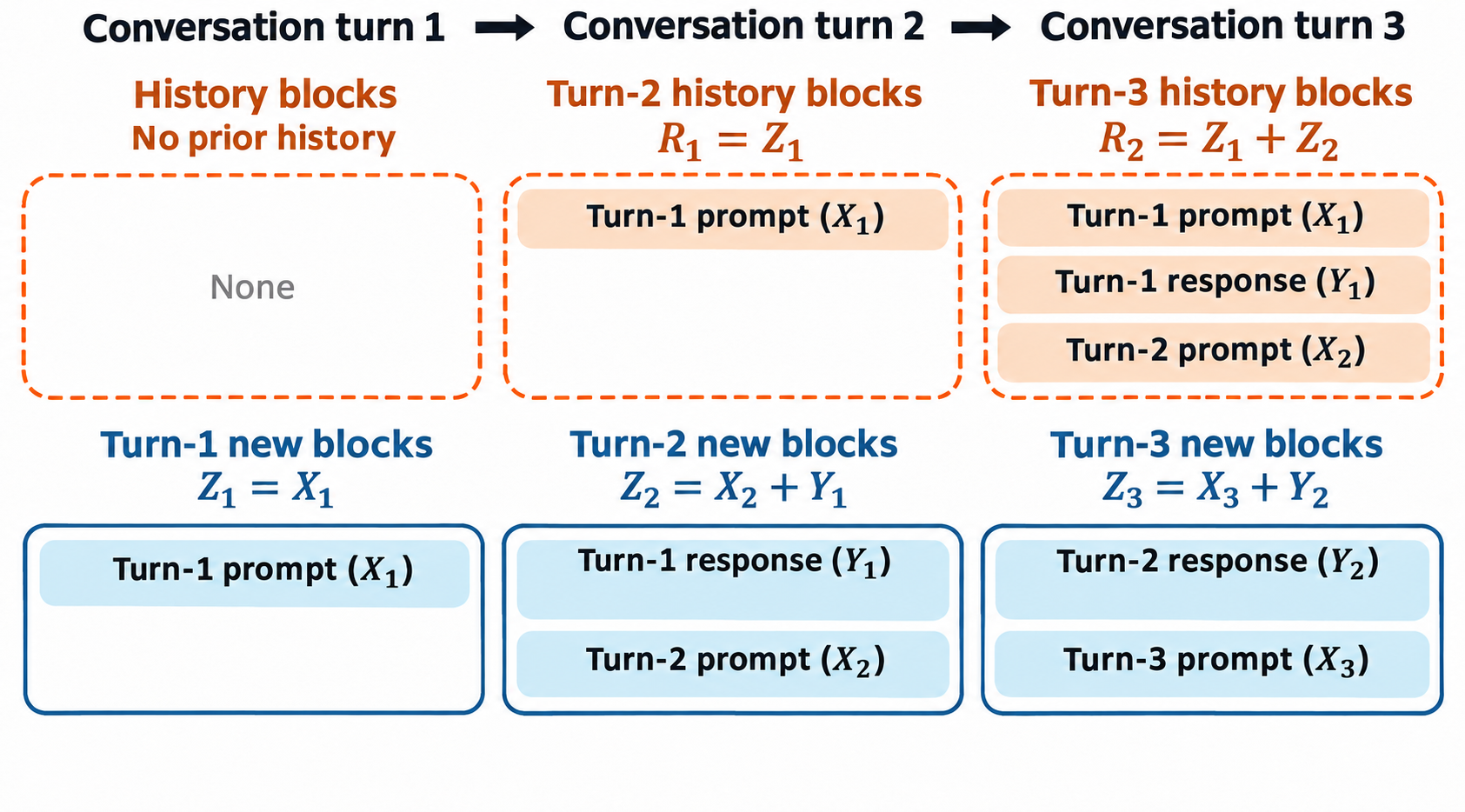}
    \caption{Illustration of the KV block composition at successive turns, under the prefill-decode disaggregation setting. Example adapted from \cite{deepseek2024contextcache}.}
    \label{fig: KV session increments example}
\end{figure}

\end{example}

\begin{remark}
Our single-workflow MCM reduces to the SCM introduced in \cite{zhang2026tail} under the following additional assumptions:
\begin{enumerate}
    \item The block increments have a common mean $\E[Z_j]=z>0$ for every $j\geq1$;
    \item The prompt IATs are i.i.d. exponentially distributed with rate ${\mu}>0$;
    \item The total number of turns $M$ is geometrically distributed. Specifically, after each completed turn, a conversation continues with probability $p\in(0,1)$ and terminates with probability $q:=1-p$. Equivalently, an active conversation can be viewed as having an independent new-prompt Poisson clock with rate $p{\mu}$ and an independent termination Poisson clock with rate $q{\mu}$.
\end{enumerate}
Under these additional assumptions, the expected session size after turn $j$ is given by ${r}_j:=\E[R_j]=jz$.
\end{remark}

\subsection{Caching Policies and the Hit Ratio}
\label{ssec: The  Hit Ratio}

We now define the caching policies under our consideration, and next define the hit ratio under a specified policy.

% \begin{definition}[Caching Policy]
%     \label{def: admissible caching policy}
%     Consider a system with HBM block capacity $N$. 
%     Each KV block in the system has a unique identifier $I\in \mathcal I$ and a label $L\in \mathcal L$. The label stores its information that can be accessed by the system. 
%     The state space is given by $\mathcal S = \mathcal I \times \mathscr L$, where $\mathscr L = \{f: \mathcal I\to \mathcal L\}$ denotes the space of mappings that assign each block a label. At any time $t$, let $V(t) = (I(t),f_t) \in \mathcal S$ denote the state of the system. 

%    In this paper we consider \textit{caching policies} represented by a mapping $\pi:\mathcal{S} \rightarrow \mathcal I$, which satisfies the following properties:
%     \begin{enumerate}
%         \item $\pi$ is Markovian. It uses the current state $V(t) = (I(t),f_t)$ to decide whether to evict blocks and which blocks to evict.
%         Specifically, at any time $t$, $\pi$ evicts the blocks with identifiers in $\pi(V(t)) \subseteq V(t)$ immediately.
%         \item It ensures that the total number of blocks remains below the capacity $N$, i.e., for any $V = (I,f)\in \mathcal S$, $|I\setminus\pi(V)| \leq N.$
%     \end{enumerate}
% \end{definition}

\begin{definition}[Caching Policy]
\label{def: admissible caching policy}
Consider a system with HBM block capacity $N$. Each KV block has a unique identifier $i\in\mathcal I$ and a label $L\in\mathcal L$. The label stores the information about the block that can be accessed by the system. Let $\mathcal P(\mathcal I)$ denote the collection of all finite subsets of $\mathcal I$. The state space is given by $\mathcal S:=\{(\mathcal B,f):\mathcal B\in\mathcal P(\mathcal I),\ f:\mathcal B\to\mathcal L\}.$

At each cache-decision epoch $n$, let $V_n^-=(\mathcal B_n^-,f_n)\in\mathcal S$ denote the state immediately before eviction, where $\mathcal B_n^-$ contains all blocks currently resident in or newly admitted to HBM.

A mapping $\pi:\mathcal S\to\mathcal P(\mathcal I)$ is called a \emph{caching policy} if it satisfies the following properties:
\begin{enumerate}
    \item It is Markovian and uses only the current state $V_n^-$ to decide whether to evict blocks and which blocks to evict. Specifically, $\pi(V_n^-)\subseteq\mathcal B_n^-$ is the set of blocks evicted at epoch $n$.
    \item It maintains the number of resident blocks below the HBM capacity, i.e., for every $V=(\mathcal B,f)\in\mathcal S$, $|\mathcal B\setminus\pi(V)|\leq N$.
\end{enumerate}
After applying the policy, the resident block set becomes $\mathcal B_n=\mathcal B_n^-\setminus\pi(V_n^-)$, and the corresponding post-eviction state is $V_n=(\mathcal B_n,f_n|_{\mathcal B_n})$.
\end{definition}

In \cref{def: admissible caching policy}, the labels are policy-dependent and contain only information available up to the current cache-decision epoch. Different blocks may have the same label because they remain distinguishable through their unique identifiers. In our following analysis for whole-session LRU, the only block label we need is the most recent access time.

% In \cref{def: admissible caching policy}, a label of a block can be any metadata required by a caching policy, such as block's initial admission time, most recent access time, access count, conversation identifier, or workflow type. 

To define the hit ratio, we first specify when a historical KV block can be reused upon the arrival of a new prompt. 
For each conversation, its historical KV blocks are indexed sequentially according to their positions in the token sequence.
Cache reuse follows the prefix-hit rule: a historical block can be reused only if that block and all preceding blocks of the same conversation are available for reuse.
Equivalently, among all historical blocks of a conversation, only the longest consecutively available prefix can be reused.
Once the cache lookup encounters the first missing block, neither that block nor any subsequent block can be reused, even if some of the subsequent blocks remain resident.

For example, suppose a conversation has blocks indexed from $0$ to $10$. Suppose blocks $0,1,2$ and $3$ are resident in HBM but block $4$ is missing. 
Then even though block 5 is resident, only blocks $0,1,2,3$ can be reused, and blocks $4$ and $5$, as well as any subsequent unavailable portion required by the new prompt, must be recomputed.

In practice, we additionally require a block to be fully populated before it is eligible for reuse.
Incorporating this requirement would introduce no substantive technical difficulty. To keep the main model and analysis concise, we omit this condition from the theoretical development. 
Instead, in \cref{sec: LLM-Inference Experiment and Empirical Results}, we will introduce a practical version of the MCM model that explicitly accounts for this condition, and present the corresponding, slightly modified results.

We then introduce workflow-specific and system-wide hit ratios, which measure the cache-reuse effectiveness of an admissible policy. 
Unlike a request-level metric that assigns equal weight to each access, our hit ratio quantifies the fraction of the requested KV-block workload that is served through direct HBM reuse. 
This distinction is important because workflows, conversations, and turns may have substantially different block lengths.

For workflow $c$, let $s_{c,\ell}:=\mathbb P(M_c\geq\ell)$ denote the tail probability of the turn count $M_c$, and let $r_{c,\ell}:=\mathbb E[R_{c,\ell}]$ denote the expected session size after turn $\ell$.
We define the \emph{expected cumulative requested workload} and the \emph{expected reusable history workload}, respectively, by
\begin{equation}
    a_c^{\mathrm{cum}}
    :=\mathbb E\!\left[\sum\limits_{\ell=1}^{M_c}R_{c,\ell}\right],
    \quad
    a_c^{\mathrm{reu}}
    :=\mathbb E\!\left[\sum\limits_{\ell=1}^{M_c-1}R_{c,\ell}\right].
    \label{equ: workflow cumulative and reusable workloads}
\end{equation}
When $J=1$, we suppress the workflow index and write $s_\ell,r_\ell,a^{\mathrm{cum}},$ and $a^{\mathrm{reu}}$.
For each workflow $c$, \cref{eq:assumption B} and $M_c\leq B_c$ imply that $a_c^{\mathrm{cum}}$ and $a_c^{\mathrm{reu}}$ are finite.

For a workflow-$c$ conversation, let $R_{c,\ell}^H(\pi)$ denote the stationary number of history-prefix blocks reused directly from HBM at turn $\ell$ under policy $\pi$. 
The workflow-$c$ \emph{workflow-specific hit ratio} is defined as
\begin{equation*}
    h_c^\pi:=
    \frac{\mathbb E\big[\sum\limits_{\ell=2}^{M_c}R_{c,\ell}^H(\pi)\big]}{a_c^{\mathrm{cum}}}.
\end{equation*}
The \emph{system-wide hit ratio} $h^\pi$ measures the stationary block reuse fraction in the entire system,
which aggregates reusable and requested blocks across workflows.
Because workflow-$c$ conversations arrive at rate $\lambda_c$, their expected hit and requested workloads are weighted by $\lambda_c$ when aggregated into system-wide block rates.
Specifically, we define $h^\pi$ as
\begin{equation*}
    h^\pi:=
    \frac{\sum\limits_{c=1}^J\lambda_c\mathbb E\left[\sum\limits_{\ell=2}^{M_c}R_{c,\ell}^H(\pi)\right]}
    {\sum\limits_{c=1}^J\lambda_c a_c^{\mathrm{cum}}}.
\end{equation*}
% {\color{blue}
% Because the HBM capacity $N$ is fixed throughout this subsection, we suppress the capacity index $(N)$ in $h_c^\pi$ and $h^\pi$.
% We restore this index when introducing the mean-field sequence of systems.
% }
% Consequently, $h^\pi$ is a requested-block-weighted average, rather than an arrival-rate-weighted average, of the workflow-specific hit ratios:
% \begin{equation}
%     \label{equ: server wide weighted average}
%     h^\pi=\sum\limits_{c=1}^J\omega_c h_c^\pi,
%     \qquad
%     \omega_c:=\frac{\lambda_c a_c^{\mathrm{cum}}}{\sum\limits_{r=1}^J\lambda_r a_r^{\mathrm{cum}}}.
% \end{equation}
Throughout this paper, unless otherwise specified, the term \emph{hit ratio} refers to the system-wide hit ratio.
When $J=1$ and $\pi$ is whole-session LRU, which we will define in \cref{ssec: The LRU Policy}, the hit ratio reduces to the single-workflow definition below.
\begin{definition}[Single-Workflow Hit Ratio under LRU]
    \label{def:  hit ratio}
    For each turn $j\geq 2$, let the random variable $R^H_j \in \N$  denote the stationary number of history-prefix blocks reused directly from HBM under whole-session LRU by a turn-$j$ access, with $0\leq R^H_j \leq R_{j-1}$. The \emph{hit ratio} under LRU is defined as
    \begin{equation}
    \label{equ: hr}
        h^{\LRU} :=\frac{\E \big[\sum\limits_{j=2}^{M}R^H_{j} \big]}{\E \big[\sum\limits_{j=1}^{M} R_j\big]} =\frac{\E \big[\sum\limits_{j=2}^{M}R^H_{j} \big]}{a^{\mathrm{cum}}} .
    \end{equation}
\end{definition}

\paragraph{\textbf{A Running Example.}}
Suppose a conversation contains three turns with increments $Z_1=2$, $Z_2=6$, and $Z_3=8$. Then the session sizes are $R_1=2$, $R_2=8$, and $R_3=16$. Suppose that, under LRU, all history blocks required at turns 2 and 3 are successfully reused from HBM. Then the hit ratio is
\begin{equation*}
    h^{\LRU}=\frac{0+2+8}{2+8+16}=\frac{10}{26}.
\end{equation*}
% In \cref{equ: hr}, the denominator $\E \big[\sum\limits_{j=1}^{M} R_j\big]$ represents the expected cumulative history session workload over the lifetime of a conversation, whereas the numerator represents the portion of this workload that is successfully touched through direct HBM reuse.

% Consider a generic conversation with $M$ turns and history session sizes $\{R_j\}_{j=1}^M$. If the conversation continues from turn $j$ to turn $j+1$, then the $R_j$ history blocks become reuse-eligible at the next prompt. Ideally, these blocks can be retrieved directly from HBM rather than reconstructed through reprefilling or transferred from lower-level storage such as CPU's Dynamic Random-Access Memory (DRAM) or an external Solid-State Drive (SSD). Before decoding the turn-$(j+1)$ response, the previously generated history and the newly arrived turn-$(j+1)$ prompt must be organized as the complete prefix required by the model. We note that, on the one hand, recovering missing history from lower-level storage introduces additional communication and storage overhead. On the other hand, reprefilling the corresponding tokens incurs additional computation and increases the workload on the prefiller. In contrast, directly reusing resident KV blocks from HBM can reduce both data transfer overhead and recomputation cost. This motivates a workload-level metric that captures the amount of reusable KV history served from HBM.

A key feature of the hit ratio is that its intrinsic upper bound is usually strictly smaller than $1$. This is because only history sessions associated with conversations that continue to the next turn are eligible for reuse.
We therefore characterize the upper bound on the single-workflow hit ratio.

For a single workflow, we define its intrinsic reuse potential as
\begin{equation*}
    h_{\max}:=\frac{a^{\mathrm{reu}}}{a^{\mathrm{cum}}}.
\end{equation*}
Then, for any admissible caching policy $\pi$ and any HBM capacity, $h^\pi \leq h_{\max}.$
The quantity $h_{\max}$ depends only on the statistical structure of the multi-turn conversation workload and is independent of the caching policy or the conversation-prompt arrival dynamics.
%For instance, in the SCM, $p_j\equiv p$, $q_j\equiv q:=1-p$, and $r_j=jz$. Hence, $a^{\mathrm{cum}}=\frac{z}{q^2},$ and $a^{\mathrm{reu}}=\frac{pz}{q^2}.$The intrinsic reuse potential is therefore $h_{\max}=p$, which depends only on the continuation probability, or equivalently the mean number of turns.

For the mixed-workflow MCM, the corresponding reuse potentials for the workflow-specific hit ratios are given by $h_{c,\max}= {a_c^{\mathrm{reu}}} / {a_c^{\mathrm{cum}}}$, while the potential for the (system-wide) hit ratio is $ h_{\max}^{\mathrm{mix}}
    := (\sum\limits_{c=1}^J\lambda_c a_c^{\mathrm{reu}}) /
            (\sum\limits_{c=1}^J\lambda_c a_c^{\mathrm{cum}})$.
We have, for any caching policy $\pi$, that
\begin{equation*}
   h_c^\pi \le  h_{c,\max} \quad \text{and} \quad h^\pi \le  h^{\text{mix}}_{\max}.
\end{equation*}

\subsection{The LRU Policy}
\label{ssec: The LRU Policy}

We now define the whole-session LRU policy used in our theoretical analysis. 
Throughout this paper, unless otherwise specified, we use LRU to refer to the whole-session LRU for brevity.

LRU orders cached blocks by their most recent access times and, whenever the cache exceeds its capacity, evicts objects from the LRU end. 
In the MCM, we take a cached object to be the entire history session of a conversation. 
Thus, all blocks belonging to the same conversation are marked the same most-recent-access label and are evicted together. 

More precisely, at each cache-decision epoch $n$, whole-session LRU groups the candidate blocks in $\mathcal B_n^-$ by conversation, and then orders the resulting conversation sessions by their most recent access times. 
It evicts all block identifiers belonging to the least recently accessed conversation, one conversation at a time, until the capacity constraint is restored.
Therefore, $\pi_{\LRU}(V_n^-)$ is the union of the block identifiers in the evicted conversations and satisfies $|\mathcal B_n^-\setminus\pi_{\LRU}(V_n^-)|\leq N$.

\Cref{alg: Whole-session LRU} gives an operational description. 
The system maintains conversation identifiers in decreasing order of recency. 
A turn first observes whether its history session is resident. 
After the turn is processed, the enlarged session is placed at the most-recently-used (MRU) end, 
and whole sessions are removed from the LRU end until the capacity constraint is restored.
\begin{algorithm}[ht]
\caption{Whole-Session LRU}
\label{alg: Whole-session LRU}
    \KwIn{HBM capacity $N$; conversation identifiers $\Theta=(\theta_1,\ldots,\theta_K)$ ordered from MRU to LRU; resident session block sizes $\{S_\theta:\theta\in\Theta\}$; an arriving turn $j$ of conversation $\theta$ and its post-access session size $R_{\theta,j}$.}

    \KwOut{Updated $\Theta$, resident session sizes, and hit indicator $H_{\theta,j}$.}

    $H_{\theta,j}\gets \one\{\theta\in\Theta\}$\;

    Remove $\theta$ from $\Theta$ if $H_{\theta,j} =1$\;
    $S_\theta\gets R_{\theta,j}$\;
    Prepend $\theta$ to $\Theta$
    \tcp*[r]{$\theta$ is now at MRU end}

    \While{$\displaystyle\sum\limits_{i\in\Theta}S_i>N$}{
        $i^*\gets$ the last conversation in $\Theta$\;
        Remove $i^*$ from $\Theta$ and delete $S_{i^*}$
        \tcp*[r]{Evict the entire session in the LRU end}
    }
    \Return{$(\Theta,\{S_i:i\in\Theta\})$}\;
\end{algorithm}

\subsection{The Mean-Field Asymptotic Regime}
\label{ssec: Our Mean-Field Model}

We next introduce the mean-field scaling used in our theoretical analysis.
We note that, if the HBM capacity $N$ grows while the conversation arrival process remains unchanged, an arriving prompt can retain almost all reusable history blocks in HBM.
As a result, the hit ratio approaches the upper bound $h_{\max}$. 
Similarly, when $N$ is very small, only a negligible fraction of the history workload can remain resident in HBM, and the hit ratio approaches $0$. 

To obtain a nontrivial asymptotic regime, we scale every workflow's conversation arrival rate proportionally with HBM capacity. 
Specifically, in the system indexed by $N$, the HBM capacity is $N$, 
and the workflow-$c$ conversations arrive according to a Poisson process with rate
\begin{equation*}
    \lambda_c^{(N)}:=N\lambda_c,
    \qquad c=1,\ldots,J,
\end{equation*}
while $F_c$, $M_c$, and the block-increment sequence $\{Z_{c,\ell}\}$ remain unchanged. 
Because $J$ is fixed, the aggregate conversation arrival rate is $N\sum\limits_{c=1}^J\lambda_c$. 
For $J=1$, we write $\lambda_1=\lambda_0$, recovering the single-workflow scaling $\lambda_0^{(N)}=N\lambda_0$ used in \cref{sec: main results}.

\subsection{Background Displacement Process and Time-to-Live (TTL)}
\label{ssec: tagged-conversation Displacement Process and Characteristic Time}

Our mean-field analysis uses a classical tagged-job approach in queueing theory. 
In our setting, we tag a conversation immediately after its turn-$j$ access and study its subsequent evolution in stationarity. 
We introduce two key notions: the background displacement process, and the corresponding TTL of a tagged conversation.
These notions provide the basis for our mean-field analysis of eviction and the LRU hit ratio.

In the mixed-workflow MCM under LRU, the background displacement process records the cumulative KV-block mass that moves ahead of a tagged session after its most recent access. Fix a tagged-conversation access at time $0$ and impose the virtual condition that the tagged conversation receives no further access before eviction. At time $0$, the tagged session is moved to the MRU end, so no resident blocks lie ahead of it. As other conversations are accessed, their KV blocks move to the MRU side and accumulate in front of the tagged session, eventually pushing it to the LRU end.

A background conversation that has already moved ahead of the tagged session does not contribute its entire session again at later accesses. 
Once its pre-existing history has crossed the tagged session, only newly appended KV blocks produce additional displacement. For each background conversation $i$, let $C_i(t)$ denote its cumulative displacement contribution during $(0,t]$.
To be specific, $C_i(t)$ is zero before the conversation's first access after time $0$. 
At that first access, its pre-existing history and newly generated blocks contribute.
At every subsequent access, only newly appended blocks contribute. 
We will prove in \cref{lem: Palm invariance of background displacement} that, the law of the displacement process does not depend on the tagged workflow index or turn index of the tagged conversation. 
Therefore, for notational simplicity, we let $D^{(N)}$ denote this process.

\begin{definition}[Background Displacement Process]
    \label{def: tagged-conversation displacement process}
    In the MCM with capacity $N$, the background displacement process $D^{(N)}$ is defined by
    \begin{equation*}
        D^{(N)}(t):=\sum\limits_{i\neq \mathrm{tag}} C_i(t), \qquad t\geq0.
    \end{equation*}
    The normalized displacement process is $D^{(N)}(\cdot)/N$.
\end{definition}

% \textcolor{blue}{The marked Poisson construction in \cref{sssec: Representation of the Displacement Process} shows that, under the Palm distribution associated with a workflow-$c$, turn-$\ell$ access, removing the tagged conversation leaves the background trajectories with their original stationary law. Consequently, the law of $D^{(N)}$ does not depend on the tagged workflow index or turn index.}

\begin{definition}[TTL of a Tagged Conversation]
    \label{def: TTL}
    Suppose that a workflow-$c$ conversation with session size $R_{c,\ell}$ is tagged immediately after its turn-$\ell$ access. Its TTL is
    \begin{equation*}
        \tau_{c,\ell}^{(N)}
        :=\begin{cases}
           0, & R_{c,\ell}>N,\\
           \inf\{t\geq0:D^{(N)}(t)>N-R_{c,\ell}\}, & R_{c,\ell}\leq N.
        \end{cases}
    \end{equation*}
\end{definition}

The process $D^{(N)}$ is nondecreasing and is generated by the independent superposition of all $J$ workflow streams. For a single workflow, we suppress $c$ and write $R_\ell$ and $\tau_\ell^{(N)}$.
In \cref{sec: main results}, we establish convergence in probability of $D^{(N)}(\cdot)/N$ in $\mathbb D$ and of $\tau_\ell^{(N)}$, which are used to derive the mean-field convergence of the LRU hit ratio.

\section{Mean-Field Asymptotics of a Single-Workflow MCM under LRU}
\label{sec: main results}

% In this section, we consider a single workflow, i.e., $J=1$.
% Therefore we suppress the workflow index $c$ and use the notation $\lambda_0,F,M,\{Z_\ell\}_{\ell\geq1},B,b,a^{\mathrm{cum}},$ and $a^{\mathrm{reu}}$ introduced in \cref{sec: Model Settings}.
% We establish convergence of the displacement process, TTL, and LRU hit ratio.
% All proofs of the main results in this section are deferred to \cref{sec: Proofs}. For the closed-form mean-field result to the SCM, we refer interested readers to \cref{ssec: Mean-Field Results of the SCM under LRU}.

In this section, we consider the single-workflow MCM, i.e., $J=1$. 
We drop the workflow index $c$, and use the notation $\lambda_0,F,M,\{Z_\ell\}_{\ell\geq1},B,b,a^{\mathrm{cum}},$ and $a^{\mathrm{reu}}$ introduced in \cref{sec: Model Settings}.
We first state our main result in this section, \cref{thm: Hit Ratio Convergence Full}, which characterizes the convergence of the hit ratios under LRU policy to the closed-form mean-field limit.
We then present the two important underlying ingredients: sample-path convergence of the displacement process and convergence of the TTL.
All proofs of results in this section are deferred to \cref{sec: Proofs}. For the closed-form mean-field result for the simplified model SCM introduced in \cite{zhang2026tail}, we refer interested readers to \cref{ssec: Mean-Field Results of the SCM under LRU}.

Let $h^{\LRU,(N)}$ denote the hit ratio under LRU for the system with capacity $N$. 
Recall from \cref{equ: workflow cumulative and reusable workloads} that $a^{\mathrm{cum}}$ and $a^{\mathrm{reu}}$ denote the expected cumulative requested workload and expected reusable history workload, respectively.
Before stating the limiting result, we first introduce the mean displacement curve and the mean-field TTL, which we will use to define the mean-field limiting hit ratio.

\begin{definition}[Mean Displacement Curve and Mean-Field TTL]
    \label{def: Mean displacement Curve}
    We write the \emph{mean displacement curve} $d:[0,\infty)\to\mathbb R$ as
    \begin{equation*}
    d(t):=\lambda_0bt+\lambda_0a^{\mathrm{reu}}\int_0^t\overline F(u)\,du,
    \end{equation*}
    \label{def: deterministic mean-field TTL}
    where $b := \mathbb E[B] = \mathbb E[R_M]$ is the expected terminal session size, and $F$ is the CDF of IAT $W$.
    
    The \emph{mean-field TTL} $t^{\mathrm{MF}}$ is the unique solution to
    \begin{equation*}
        d(t^{\mathrm{MF}})=1.
    \end{equation*}
\end{definition}

The two terms in $d(t)$ represent two distinct sources of displacement. The first term $\lambda_0bt$ is the normalized mean number of newly generated KV blocks during $(0,t]$, whereas $\lambda_0a^{\mathrm{reu}}\int_0^t\overline F(u)\,du$ counts pre-existing history blocks that move ahead of the tagged session when their conversations are revisited. Here, $\int_0^t\overline F(u)\,du=\E[\min\{t,W\}]$ captures the stationary mass of such revisits; equivalently, $d(t)$ can be written as $d(t) = \lambda_0 b t + \lambda_0 a^{\mathrm{reu}} \mathbb E[\min \{t, W\}]$.

We now define the mean-field hit ratio as
\[
    h^{\LRU,\infty}:=\frac{a^{\mathrm{reu}}}{a^{\mathrm{cum}}}F(t^{\mathrm{MF}}).
\]

\begin{restatable}[Hit Ratio Convergence in the MCM-LRU Model]{theorem}{hitconvergencetheorem}
\label{thm: Hit Ratio Convergence Full}
    % Consider the single-workflow MCM under LRU. Under \cref{ass:independent block incrememts,ass: Regularity of one cumulative conversation blocks,ass: Exponential Prompt Inter-arrival Time}, the hit ratio under the whole-session LRU policy converges when $N\rightarrow\infty$:

    Consider the single-workflow MCMs under LRU. We have
     \begin{equation*}
        h^{\LRU,(N)} \to  h^{\LRU,\infty} ~\text{as} ~ N \to \infty.
    \end{equation*}
\end{restatable}

The limiting hit ratio in \cref{thm: Hit Ratio Convergence Full} therefore admits a meaningful decomposition. 
Specifically, $ a^{\mathrm{reu}} / a^{\mathrm{cum} } = h_{\max}$ is the fraction of requested blocks eligible for reuse, and $F(t^{\mathrm{MF}})$ denotes the asymptotic probability of reuse before eviction.
In particular, the workload's statistical structure affects the limiting hit ratio through the reusable-workload fraction $a^{\mathrm{reu}}/a^{\mathrm{cum}}$ and through its effect on the mean-field TTL.

The proof of \cref{thm: Hit Ratio Convergence Full} rests on the following two supporting results.
First, the normalized displacement process converges to function $d$ in probability in $\mathbb D$. 
Second, the TTL of a tagged history session converges in probability to $t^{\mathrm{MF}}$. 
We will additionally use these results in \cref{sec: mixed workflow optimization} to characterize LRU's implicit workflow allocation and the mean-field performance of WAQ-LRU.

Our first supporting result formalizes the sample-path convergence of the normalized displacement process.

\begin{restatable}[Sample Path Convergence of the Normalized Displacement Process]{proposition}{displacementconvergencetheorem}
\label{thm: Displacement Process Convergence}
    Consider the single-workflow MCM under LRU. As $N\to\infty$,
    \begin{equation}
    \label{equ: displacement sample path convergence}
        \begin{gathered}
            \frac{D^{(N)}(\cdot)}{N}\overset{\mathbb P}{\to} d \quad\text{in}~~\mathbb D.
        \end{gathered}
    \end{equation}
\end{restatable}

Our second supporting result applies this sample-path convergence to the TTL of a tagged turn-$j$ history session.

\begin{restatable}[TTL Convergence]{proposition}{ttlconvergencetheorem}
\label{thm: Characteristic Time Convergence- Full}
    Consider the single-workflow MCM under LRU. For every fixed turn state $j$ with $s_j=\mathbb P(M\ge j)>0$,
    \begin{equation*}
        \tau_j^{(N)}\xrightarrow{\mathbb P}t^{\mathrm{MF}}.
    \end{equation*}
\end{restatable}

We conclude this section with a proof sketch of \cref{thm: Hit Ratio Convergence Full}. 
We build our argument on the preceding convergence results. 
For each fixed turn state, we combine the TTL convergence with the Lipschitz continuity of $F$ to establish convergence of the corresponding hit probability. 
We then use a truncation argument to extend the fixed-turn result to the overall hit ratio, by showing that sufficiently large turn indices contribute negligibly.

In the following section, we use our characterization of the limiting LRU hit ratio to reveal hidden room for improvement when heterogeneous workflows share memory.

\section{Hit Ratio Optimization for Mixed-Workflow MCM}
\label{sec: mixed workflow optimization}

We now consider the mixed-workflow MCM. 
For workflow $c$, let $b_c:=\mathbb E[B_c]$ denote the expected terminal session size. Recall that $a_c^{\mathrm{cum}}$ and $a_c^{\mathrm{reu}}$ are its expected requested and reusable workloads. 
We first in \cref{ssec: Implicit Workflow Allocation of LRU} characterize LRU and demonstrate its important property in the mixed-workflow MCM, that it implicitly allocates memory among workflows.
In \cref{ssec: Workflow-Aware Quota LRU}, we then introduce WAQ-LRU, which makes this allocation an explicit control decision. 
Finally, in \cref{ssec: Optimal Quota Allocation under WAQ-LRU}, we demonstrate that our policy can always (weakly) improve the hit ratio upon LRU. 
All proofs of the main results in this section are deferred to \cref{app: section5 proofs}.

\subsection{Implicit Workflow Allocation of LRU}
\label{ssec: Implicit Workflow Allocation of LRU}

For workflow $c$, similar to \cref{def: Mean displacement Curve}, we define its mean displacement curve by
\begin{equation*}
    d_c(t):=\lambda_c b_c t
    +\lambda_c a_c^{\mathrm{reu}}\int_0^t\overline F_c(u)\,du,
    \qquad t\geq0.
\end{equation*}
Likewise, the first term is contributed by newly arriving workflow-$c$ KV blocks, while the second is generated by reusing the historical workflow-$c$ blocks, which are
already present at time zero.
Because the $J$ workflow streams are mutually independent and disjoint, the mean displacement curve is therefore the sum of these $J$ curves
\begin{equation*}
    d(t):=\sum\limits_{c=1}^Jd_c(t).
\end{equation*}
Since $d$ is monotonically increasing, $d(0)=0$, and $d(t) \to \infty$ as $t\to \infty$,
there is a unique solution $t^{\mathrm{MF}}_{\LRU}$ for $d(t^{\mathrm{MF}}_{\LRU})=1$, which we call the \emph{mixed-workflow mean-field TTL under LRU}.

\begin{restatable}[Mixed-Workflow Mean-Field Results under LRU]{proposition}{mixedlrutheorem}
    \label{thm: mixed workflow LRU mean field}
    Consider the mixed-workflow MCM under LRU, with $\mathcal C = \{1,\ldots,J\}$.
    The following statements hold:
    \begin{enumerate}
        \item\label{item: mixed workflow TTL convergence} For any workflow $c\in \mathcal C$, and for any turn state $j$ such that $s_{c,j}>0$, the TTL $\tau_{c,j}^{(N)} \overset{\mathbb P}{\to} t^{\mathrm{MF}}_{\LRU}$ as $N\to \infty$.
        \item For any workflow $c\in \mathcal C$, the workflow-specific hit ratio $h_c^{\LRU,(N)}$ satisfies
        \begin{equation}
            \label{equ: type specific LRU hit ratio mean field}
             h_c^{\LRU,(N)} \to h_c^{\LRU,\infty} =\frac{a_c^{\mathrm{reu}}}{a_c^{\mathrm{cum}}}F_c(t^{\mathrm{MF}}_{\LRU})\quad \text{as } N\to\infty.
        \end{equation}
        \item The system-wide hit ratio $h^{\LRU,(N)}$ satisfies
        \begin{equation}
            \label{equ: server wide LRU hit ratio mean field}
            h^{\LRU,(N)} \to h^{\LRU,\infty}=\frac{\sum\limits_{c=1}^J\lambda_c a_c^{\mathrm{reu}} F_c(t^{\mathrm{MF}}_{\LRU})}{\sum\limits_{c=1}^J\lambda_c a_c^{\mathrm{cum}}}\quad \text{as } N\to\infty.
        \end{equation}
    \end{enumerate}
\end{restatable}

The above results show that in the mean-field limit, LRU implicitly allocates cache capacity across workflows. 
We make this implicit allocation precise in \cref{thm: implicit workflow allocation under LRU}, where $D_c^{(N)}(t)$ denotes the displacement contributed by workflow $c$ to the background displacement process during $(0,t]$. Define the \textit{implicit LRU quota of workflow $c$} by 
\begin{equation}\label{eq: implicitLRUquota}
	\alpha_c^{\LRU}:=d_c(t^{\mathrm{MF}}_{\LRU}).
\end{equation}
 Note that $\alpha_c^{\LRU} \ge 0$ and $\sum\limits_{c\in\mathcal C}\alpha_c^{\LRU}=1$.

\begin{restatable}[Implicit Workflow Allocation under LRU]{theorem}{implicitallocationtheorem}
\label{thm: implicit workflow allocation under LRU}
    Consider the mixed-workflow MCM under LRU, with $\mathcal C = \{1,\ldots,J\}$.
    For every fixed turn state $j$ such that $s_{c,j}>0$,
    \begin{equation}
        \label{equ: random-time implicit quota convergence}
        \frac{D_c^{(N)}(\tau_{c,j}^{(N)})}{N}
        \xrightarrow{\mathbb P} \alpha_c^{\LRU}\quad \text{as} ~N \to \infty.
    \end{equation}
\end{restatable}

Thus, in our mean-field regime, LRU allocates cache capacity for each workflow indirectly through workflow statistical structure, and enforces identical TTLs.

% The formulas in this section use the fully hashable MCM notation. For the practical MCM with one unhashable tail block per completed turn, every statement remains unchanged after replacing $a_c^{\mathrm{reu}}$ by the workflow-specific expected reusable hashable workload $a_c^{\mathrm{reu,hash}}$. In particular, this replacement must be made both in $d_c$ and in the hit ratio numerator; $a_c^{\mathrm{cum}}$ and $b_c$ continue to count all physical blocks that consume HBM or require computation.

\subsection{Workflow-Aware Quota LRU (WAQ-LRU)}
\label{ssec: Workflow-Aware Quota LRU}

Instead of letting LRU implicitly allocate capacity across workflows (as governed by \cref{thm: implicit workflow allocation under LRU}), we propose the family of WAQ-LRU policies, which makes these allocations explicit and optimizable for the hit ratio.

We now define our (whole-session) WAQ-LRU policies. 
Consider an MCM with $J$ workflows. 
We let $\Delta^J:=\{\vec\alpha\in\mathbb R_+^J: \sum\limits_{c=1}^J\alpha_c=1\}$ denote the unit simplex. 
Given the quota parameter $\vec\alpha\in\Delta^J$, the $\vec\alpha$-WAQ-LRU policy partitions memory $N$ into $J$ workflow-specific quotas. 
To be specific, workflow $c$ receives capacity $\lfloor\alpha_cN\rfloor$ and maintains its individual LRU list. 
An access by workflow $c$ can only evict conversation sessions within the workflow-$c$ list.  
\cref{alg: WAQ LRU} gives a detailed description of the whole-session $\vec\alpha$-WAQ-LRU policy for a given $\vec\alpha \in \Delta^J$.

\begin{algorithm}[ht]
\caption{Whole-Session $\vec\alpha$-WAQ-LRU}
\label{alg: WAQ LRU}
\KwIn{Memory capacity $N$; quota vector $\vec\alpha\in\Delta^J$; workflow-specific MRU-to-LRU lists $\{\Theta_c\}_{c=1}^J$ and resident sizes $\{S_\theta\}$; an arriving workflow-$c$ turn of conversation $\theta$ with post-access size $R$.}
\KwOut{Updated workflow lists and resident session sizes.}
$N_c\gets\lfloor\alpha_cN\rfloor$\;
Remove $\theta$ from $\Theta_c$ if it is present\;
$S_\theta\gets R$ and prepend $\theta$ to $\Theta_c$\;
\While{$\displaystyle\sum\limits_{i\in\Theta_c}S_i>N_c$}{
    $i^*\gets$ the last conversation in $\Theta_c$\;
    Remove $i^*$ from $\Theta_c$ and delete $S_{i^*}$\;
}
\Return{$\{\Theta_c\}_{c=1}^J$ and the resident sizes}\;
\end{algorithm}

Consider the mixed-workflow MCM under policy $\vec \alpha$-WAQ-LRU. For workflow $c$, 
we define the mean-field TTL as $t_{c}^{\mathrm{MF}}(\alpha_c):=d_c^{-1}(\alpha_c)$, i.e., the unique solution to 
    $d_c(t_{c}^{\mathrm{MF}}(\alpha_c))=\alpha_c$.
Let $h_c^{\mathrm{WAQ},(N)}(\alpha_c)$ and $h^{\mathrm{WAQ},(N)}(\vec\alpha)$ denote the workflow-specific hit ratio and the system-wide hit ratio in the system with capacity $N$, respectively.
We now show the mean-field results of the mixed-workflow MCM under $\vec \alpha$-WAQ-LRU. 

\begin{restatable}[Mean-Field Hit Ratios under WAQ-LRU]{theorem}{waqmeanfieldtheorem}
    \label{thm: mixed workflow WAQ LRU mean field}
     Consider the mixed-workflow MCM under $\vec\alpha$-WAQ-LRU for a given quota vector $\vec\alpha\in\Delta^J$. As $N\to \infty$, we have that
    \begin{align}
 		\tau_{c,j}^{(N)}(\alpha_c) \xrightarrow{\mathbb P} t_{c}^{\mathrm{MF}}(\alpha_c) \\
        h_c^{\mathrm{WAQ},(N)}(\alpha_c) \to h_c^{\mathrm{WAQ},\infty}(\alpha_c)
        &=\frac{a_c^{\mathrm{reu}}}{a_c^{\mathrm{cum}}}F_c(t_{c}^{\mathrm{MF}}(\alpha_c)),
        \label{equ: type specific WAQ hit ratio}\\
        h^{\mathrm{WAQ},(N)}(\vec\alpha) \to h^{\mathrm{WAQ},\infty}(\vec\alpha)
        &=\sum_{c=1}^J
          \frac{\lambda_c a_c^{\mathrm{cum}}}{\sum_{r=1}^J\lambda_r a_r^{\mathrm{cum}}}\,
          h_c^{\mathrm{WAQ},\infty}(\alpha_c)\notag\\
        &=\frac{\sum\limits_{c=1}^J\lambda_c a_c^{\mathrm{reu}}
        F_c(t_{c}^{\mathrm{MF}}(\alpha_c))}
        {\sum\limits_{c=1}^J\lambda_c a_c^{\mathrm{cum}}}.
        \label{equ: server wide WAQ hit ratio}
    \end{align}
\end{restatable}

We note that, for the implicit LRU quota vector $\vec\alpha^{\LRU}$ in \cref{thm: implicit workflow allocation under LRU}, it holds that 
\begin{equation*}
	h^{\mathrm{WAQ},\infty}(\vec\alpha^{\LRU}) =h^{\LRU,\infty},
\end{equation*}
i.e., the limiting hit ratio under policy $\vec \alpha^{\LRU}$-WAQ-LRU coincides with that under policy LRU. Therefore, the limiting hit ratio obtained by LRU is achievable by the family of WAQ-LRU policies. 
Given the characterization of the limiting hit ratios, we can now optimize the limiting hit ratio over the quota vectors.

\subsection{Optimal Quota Allocation under WAQ-LRU}
\label{ssec: Optimal Quota Allocation under WAQ-LRU}

The quota vector $\vec\alpha$ provides WAQ-LRU with a degree of freedom that is absent from LRU.
To be specific, the server can redistribute cache capacity across workflows, instead of being restricted by the implicit allocation $\vec\alpha^{\LRU}$.
We note that the implicit LRU allocation $\vec\alpha^{\LRU}$ is feasible, and corresponds to the same limiting system-wide hit ratio as LRU.
As a result, optimizing over $\vec\alpha$ can strictly improve LRU's hit ratio when $\vec\alpha^{\LRU}$ is suboptimal. 
Consequently, a higher hit ratio serves a larger fraction of requested KV blocks directly from memory, which improves the effectiveness of cache utilization.
In \cref{sec: LLM-Inference Experiment and Empirical Results}, we further demonstrate empirically that this improvement can also reduce TTFT.

Inspired by \cref{thm: mixed workflow WAQ LRU mean field}, we introduce the following mean-field quota-allocation optimization.
\begin{equation}
\label{equ: WAQ quota optimization}
\begin{aligned}
    \underset{\vec\alpha \in \Delta^J}{ \max}\quad
    & h^{\mathrm{WAQ},\infty}(\vec\alpha)
    =\frac{\sum\limits_{c=1}^J\lambda_c a_c^{\mathrm{reu}}
    F_c\bigl(t_{c}^{\mathrm{MF}}(\alpha_c)\bigr)}
    {\sum\limits_{c=1}^J\lambda_c a_c^{\mathrm{cum}}},
\end{aligned}
\end{equation}
where we recall that $t_c^{\mathrm{MF}} (\alpha_c)$ is defined by the equation $d_c(t_c^{\mathrm{MF}}(\alpha_c))=\alpha_c$.
% Equivalently, an optimal quota parameter $\vec\alpha^*$ belongs to the argmax of the numerator in \cref{equ: WAQ quota optimization} over $\Delta^J$.

Because $d_c^{-1}$ and $F_c$ are continuous and the feasible set $\Delta^J$ is compact, the optimization in \cref{equ: WAQ quota optimization} admits an optimal solution; let $\vec \alpha^*$ be such a solution.
We provide some sufficient conditions for the uniqueness of the optimal quota vector in \cref{prop: optimal WAQ quota marginal condition,prop: WAQ optimizer existence and uniqueness} in \cref{app: section5 proofs}. 
Although developing efficient algorithms for this problem is beyond the focus of this paper, we note that, even without these conditions, a global optimum can be computed efficiently by discretizing the quotas into integer numbers of blocks and applying dynamic programming, because the objective is separable across workflows.

Because the family of WAQ-LRUs can reproduce LRU's implicit quota and additionally optimize over all feasible quotas, its optimized limiting hit ratio, obtained by the $\vec \alpha^*$-WAQ-LRU policy, cannot be lower than that of LRU. We summarize the above reasoning in the following theorem.

\begin{theorem}[Optimal WAQ-LRU Dominates LRU]
\label{cor: optimal WAQ LRU dominates LRU}
The $\vec\alpha^*$-WAQ-LRU limiting hit ratio satisfies
\begin{equation*}
    h^{\mathrm{WAQ},\infty}(\vec\alpha^*)
    \geq h^{\LRU,\infty}=h^{\mathrm{WAQ},\infty}(\vec\alpha^{\LRU} ),
\end{equation*}
with $\vec\alpha^{\LRU}$ defined in \cref{eq: implicitLRUquota}. 
\end{theorem}

Hence, the optimized WAQ-LRU is guaranteed to dominate LRU in the limiting hit ratio. 
In \cref{sec: LLM-Inference Experiment and Empirical Results}, we empirically validate this guarantee and show that the optimized WAQ-LRU improves both the hit ratio and TTFT.

\section{Theoretically-Motivated Hit Ratio Estimator and LLM Serving Experiments}
\label{sec: LLM-Inference Experiment and Empirical Results}

In this section, we apply our theoretical results to practical hit ratio prediction and cache-management guidance for real LLM serving systems. 
We first introduce the practical MCM in \cref{ssec: Practical MCM Consequence and Theoretically-Motivated Hit Ratio Estimator}. 
This model further considers partially populated tail blocks that occupy memory but cannot be reused through prefix hashing,
and the block-level implementations of LRU and WAQ-LRU used in practice.
Building on the corresponding mean-field results, we develop our theoretically-motivated hit ratio estimators for both policies. 
In \cref{ssec: LRU estimator real practice}, we evaluate the LRU estimator under three model--dataset configurations and show that its predictions closely track the observed HBM hit ratios. 
We then consider heterogeneous workloads in \cref{ssec: WAQ-LRU estimator real practice}. 
For each of three workflow mixtures, we use the WAQ-LRU estimator to select a workflow-aware quota allocation and deploy the resulting allocation in the serving system. 
Compared with LRU, the selected quotas achieve higher HBM hit ratios and lower TTFT.
Our experiments validate the predictive accuracy of our mean-field estimators and demonstrate the practical advantage of WAQ-LRU over LRU.

\subsection{Practical MCM and Theoretically-Motivated Hit Ratio Estimator}
\label{ssec: Practical MCM Consequence and Theoretically-Motivated Hit Ratio Estimator}

The MCM in \cref{sec: Model Settings} makes two idealizations that should be relaxed before applying our mean-field results to an LLM serving system.
First, prefix caching can reuse only fully populated and hashed blocks. 
For example, with the 128-token blocks used in our experiments, a 255-token history occupies two physical blocks, but only its first block is fully populated and eligible for a prefix-cache hit. 
The remaining 127 tokens occupy an \emph{unhashable tail block}. When the next prompt arrives, these tokens must be reprefilled together with the new suffix tokens. 
As a result, the tail block consumes memory and contributes to cache displacement even though it cannot produce a hit.
Additionally, production systems evict KV states block by block rather than necessarily evicting an entire conversation session at once. 
To be specific, the cache boundary may cut through a session, leaving a resident prefix and evicting its suffix. 
These two features motivate the \emph{practical MCM} and the block-level LRU policy introduced below.

In our practical MCM model, we now conservatively assume each block increment contains one unhashable tail block.
For a workflow-$c$ conversation, let $R_{c,j}$ denote the number of physical blocks required by its prompt at turn $j$, including the turn-$j$ unhashable tail block. 
Hence, the reusable and hashable history after turn $j$ contains $R_{c,j}-1$ blocks. 
As a result,  $Z_{c,j}$ denotes the number of blocks newly prefilled and admitted at turn $j$, rather than the net increase in the conversation session.
To be specific,
\begin{equation}
\label{equ: practical block increments}
    Z_{c,1}:=R_{c,1},\qquad
    Z_{c,j}:=R_{c,j}-(R_{c,j-1}-1)=R_{c,j}-R_{c,j-1}+1,\quad j\geq2.
\end{equation}
The additional one in \cref{equ: practical block increments} is the preceding turn's tail block, which must be reprefilled again. 
Similarly, in the practical MCM,
\begin{equation*}
    B_c:=\sum_{j=1}^{M_c}Z_{c,j}=R_{c,M_c}+M_c-1
\end{equation*}
is the total lifetime new-block workload, not the terminal session block count. 
We likewise define
\begin{equation*}
    a_c^{\mathrm{cum}}:=\E\!\Big[\sum_{j=1}^{M_c}R_{c,j}\Big],\quad
    a_c^{\mathrm{reu,hash}}:=\E\!\Big[\sum_{j=1}^{M_c-1}(R_{c,j}-1)\Big],\quad
    b_c:=\E[B_c]=a_c^{\mathrm{cum}}-a_c^{\mathrm{reu,hash}},
\end{equation*}
Specifically, $a_c^{\mathrm{cum}}$ denotes the mean number of blocks requested in a workflow-$c$ conversation,
$a_c^{\mathrm{reu,hash}}$ denotes the mean number of reusable blocks,
and $b_c$ denotes the mean number of distinct blocks.

We then specify the \emph{block-level} LRU and WAQ-LRU. Under both policies, resident KV blocks are ordered by their most recent access.
For LRU, an access refreshes the resident hashable prefix and admits the newly prefilled blocks.
When the capacity $N$ is exceeded, blocks are evicted one at a time from the LRU end, until there are at most $N$ total resident blocks. 
Within a conversation session, eviction proceeds from the suffix toward the prefix, so a boundary session may be only partially resident. 
Under block-level $\vec\alpha$-WAQ-LRU, the same rule is applied independently within each workflow-$c$ quota $\lfloor\alpha_cN\rfloor$.
We defer the block-level analysis to \cref{app: direct block-level practical MCM}.
% Formal algorithms and the block-level sandwich argument are given in \cref{sec: Block-Level LRU Extension}.

Define the practical workflow displacement curves by
\begin{equation*}
    d_c^{\mathrm{prac}}(t):=\lambda_cb_ct+
    \lambda_ca_c^{\mathrm{reu,hash}}\int_0^t\overline F_c(u)\,du,
    \qquad
    d_{\LRU}^{\mathrm{prac}}(t):=\sum_{c=1}^Jd_c^{\mathrm{prac}}(t).
\end{equation*}
For LRU, let $t_{\LRU}^{\mathrm{MF,prac}}$ be the unique solution of $d_{\LRU}^{\mathrm{prac}}(t)=1$. 
For WAQ-LRU with quota vector $\vec\alpha\in\Delta^J$, let $t_c^{\mathrm{MF,prac}}(\alpha_c)$ be the unique solution of $d_c^{\mathrm{prac}}(t)=\alpha_c$. 
We now state our main mean-field result in the practical MCM.

\begin{restatable}[Practical Hit Ratio Limits]{theorem}{practicalhitconvergencetheorem}
\label{thm: practical mixed workflow mean field}
Consider the MCM with $J$ workflows. The block-level LRU admits the following hit ratio convergences as $N\to\infty$:
\begin{align}
    h_c^{\LRU,(N)}
    &\to\frac{a_c^{\mathrm{reu,hash}}}{a_c^{\mathrm{cum}}}
    F_c(t_{\LRU}^{\mathrm{MF,prac}}),
    \notag\\
    h^{\LRU,(N)}
    &\to\frac{\sum_{c=1}^J\lambda_ca_c^{\mathrm{reu,hash}}
    F_c(t_{\LRU}^{\mathrm{MF,prac}})}
    {\sum_{c=1}^J\lambda_ca_c^{\mathrm{cum}}}.
    \label{equ: practical mixed LRU server limit main}
\end{align}
For any fixed $\vec\alpha\in\Delta^J$, the block-level $\vec\alpha$-WAQ-LRU admits the following hit ratio convergences as $N\to\infty$:
\begin{align}
    h_c^{\mathrm{WAQ},(N)}(\alpha_c)
    &\to\frac{a_c^{\mathrm{reu,hash}}}{a_c^{\mathrm{cum}}}
    F_c(t_c^{\mathrm{MF,prac}}(\alpha_c)),
    \label{equ: practical mixed WAQ workflow limit main}\\
    h^{\mathrm{WAQ},(N)}(\vec\alpha)
    &\to\frac{\sum_{c=1}^J\lambda_ca_c^{\mathrm{reu,hash}}
    F_c(t_c^{\mathrm{MF,prac}}(\alpha_c))}
    {\sum_{c=1}^J\lambda_ca_c^{\mathrm{cum}}}.
    \label{equ: practical mixed WAQ server limit main}
\end{align}
\end{restatable}
The proof of \cref{thm: practical mixed workflow mean field} is deferred to \cref{app: direct block-level practical MCM}.
% The single-workflow practical MCM is the special case $J=1$.
% The proof is deferred to \cref{app: practical MCM}; the extension from whole-session to block-level eviction follows from \cref{sec: Block-Level LRU Extension}.

Based on \cref{thm: practical mixed workflow mean field}, we introduce our theoretically-motivated estimators under LRU and WAQ-LRU. 
For workflow $c$, suppose that we have the estimate $\widehat a_c^{\mathrm{cum}}, \widehat a_c^{\mathrm{reu,hash}}, \widehat b_c,\widehat F_c(\cdot)$ and $\widehat\lambda_c$ of $a_c^{\mathrm{cum}}, a_c^{\mathrm{reu,hash}}, b_c,F_c(\cdot)$ and $\lambda_c$ , respectively.
We use these estimated primitives to construct practical hit ratio estimators for LRU and WAQ-LRU, and then use the WAQ-LRU estimator to select the quota vector deployed in our implementation.

\paragraph{\textbf{Practical Hit Ratio Estimators under LRU.}}

For LRU, define $\widehat t_{\LRU}^{\mathrm{prac}}$ as the unique solution of
\begin{equation*}
    \sum\limits_{c=1}^J\widehat\lambda_c\Big[\widehat b_c t+\widehat a_c^{\mathrm{reu,hash}}\int_0^t\widehat{\overline F}_c(u)\,du\Big]=N.
\end{equation*}
The workflow-specific and system-wide practical hit ratio estimators are
\begin{equation*}
    \widehat h_c^{\LRU}=\frac{\widehat a_c^{\mathrm{reu,hash}}}{\widehat a_c^{\mathrm{cum}}}\widehat F_c(\widehat t_{\LRU}^{\mathrm{prac}}),
    \qquad
    \widehat h^{\LRU}=\frac{\sum\limits_c\widehat\lambda_c\widehat a_c^{\mathrm{reu,hash}}\widehat F_c(\widehat t_{\LRU}^{\mathrm{prac}})}{\sum\limits_c\widehat\lambda_c\widehat a_c^{\mathrm{cum}}}.
\end{equation*}

\paragraph{\textbf{Practical Hit Ratio Estimators under $\vec\alpha$-WAQ-LRU.}}

For $\vec\alpha$-WAQ-LRU, define $\widehat t_c^{\mathrm{prac}}(\alpha_c)$ by
\begin{equation*}
\widehat\lambda_c\Big[\widehat b_c\widehat t_c^{\mathrm{prac}}(\alpha_c)+\widehat a_c^{\mathrm{reu,hash}}\int_0^{\widehat t_c^{\mathrm{prac}}(\alpha_c)}\widehat{\overline F}_c(u)\,du\Big]=\alpha_cN.
\end{equation*}
Then our hit ratio estimator is given by 
\begin{equation}
\label{equ: practical WAQ hit ratio estimator}
\widehat h^{\mathrm{WAQ}}(\vec\alpha)
=\frac{\sum\limits_c\widehat\lambda_c\widehat a_c^{\mathrm{reu,hash}}\widehat F_c(\widehat t_c^{\mathrm{prac}}(\alpha_c))}{\sum\limits_c\widehat\lambda_c\widehat a_c^{\mathrm{cum}}}.
\end{equation}

\paragraph{\textbf{Optimal Practical Quota Allocation Vector.}}
Guided by the estimator in \cref{equ: practical WAQ hit ratio estimator}, we select the optimal quota vector
\begin{equation*}
\widehat{\vec\alpha}^{\,*}\in\arg\max_{\vec\alpha\in\Delta^J}\widehat h^{\mathrm{WAQ}}(\vec\alpha),
\end{equation*}
and deploy WAQ-LRU with the optimal quota allocation $\widehat{\vec\alpha}^{\,*}$.

{\color{blue}

}

\subsection{LRU Hit Ratio Estimator Validation}
\label{ssec: LRU estimator real practice}

In our first group of experiments,
we validate our practical LRU hit ratio estimator under three model--dataset configurations: Qwen3-32B with ShareGPT, Qwen3-8B with ShareGPT, and Qwen3-4B with WildChat \citep{sharegpt_vicuna_unfiltered,wildchat48m,zhao2024wildchat,qwen3technicalreport}.
All experiments use prefill--decode disaggregation to reduce interference between prefill and decode computation, with vLLM \citep{kwon2023pagedattention} v0.22.1, vLLM-Ascend v0.22.1rc1, and Mooncake v0.3.9 \citep{qin2025mooncake}.
We present the Qwen3-32B--ShareGPT configuration in the main text and defer the other two configurations to \cref{app: additional experiments}.

\subsubsection{Workload Description and Experimental Setup.}
\label{ssec: LRU estimate Workload Description and Experimental Setup}

In this experiment, we deploy Qwen3-32B on eight Ascend 910B2 NPUs, where four NPUs form a TP=4 prefill instance and the other four form a TP=4 decode instance. 
We allocate $36$ GiB of HBM on each prefill NPU exclusively to the KV cache, and obtain a total KV-cache capacity of $144$ GiB.
In the Qwen3-32B model, one token occupies $256$ KiB of KV-cache memory, and a 128-token block occupies $32$ MiB. 
As a result, the prefiller can consequently store at most $N=\lfloor144\,\mathrm{GiB}/(32\,\mathrm{MiB})\rfloor=4608$ KV blocks.

We construct the workload from the public ShareGPT dataset \citep{sharegpt_vicuna_unfiltered} and sample 6,000 conversations containing at least two turns. The first 1,000 conversations, containing 4,690 turns, are used for HBM cache warm-up and are sufficient to fill the HBM before measurement. The remaining 5,000 conversations form the measurement set. Conversation submissions in both phases follow Poisson processes.
% The measurement window starts at the submission time $t_1$ of the first measurement conversation and ends at the submission time $t_{5000}$ of the last.
% Although this finite horizon may censor late turns of active conversations, we avoid increasing the hit ratio by continuing measurement after new-conversation arrivals have stopped.
The measurement window starts at the submission time $t_1$ of the first measurement conversation and ends at the submission time $t_{5000}$ of the last. We stop measurement at $t_{5000}$, because continuing to drain active conversations after new conversation arrivals cease would create an artificially favorable cache environment. 
This convention nevertheless right-censors late turns. The censoring is relatively mild in our estimator-validation experiment: across the reported sweeps, between $87.4\%$ and $96.9\%$ of the $22,347$ scheduled measurement turns are executed within the measurement window.

All block-level workload statistics are computed using the Qwen3 tokenizer and chat template. 
\Cref{fig: ShareGPT Dataset info} summarizes the empirical numbers of turns per conversation and new KV blocks per turn.
Note that \cref{fig: E1_new_kv_block_histogram} describes only the marginal increment distribution, whereas our estimator does not require increments to be i.i.d. across turns. 
The resulting estimated workload statistics are
\[\widehat a^{\mathrm{cum}} = 26.9746,\quad\widehat a^{\mathrm{reu,hash}} = 14.1140, \quad \widehat b = 12.8606,\quad \text{and} \quad\widehat\mu_M = 4.4694.\]

\begin{figure}[ht]
    \centering
    \begin{subfigure}[t]{0.48\linewidth}
        \centering
        \includegraphics[width=\linewidth]{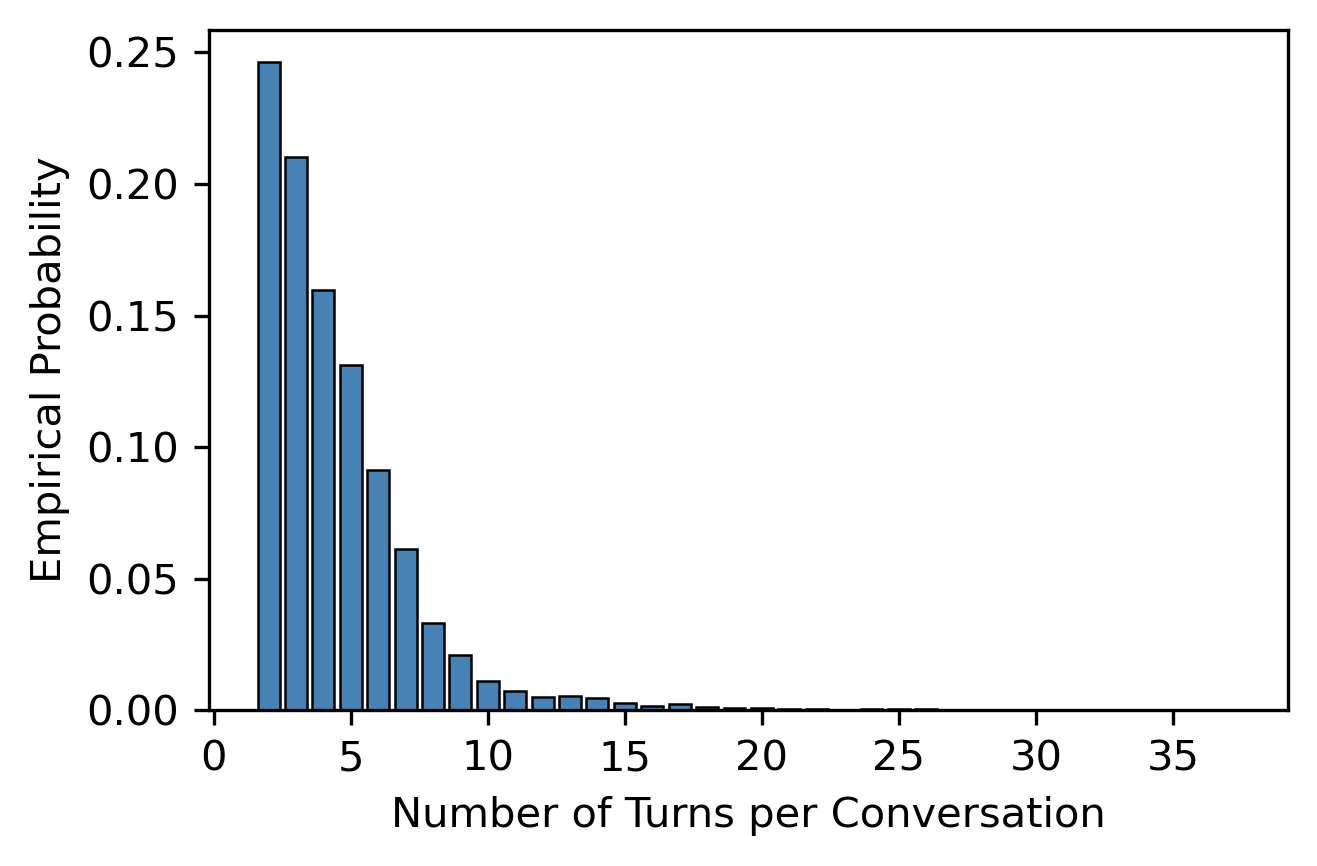}
        \caption{Number of turns per conversation.}
        \label{fig: E1_empirical_conversation_turn_distribution}
    \end{subfigure}
    \hfill
    \begin{subfigure}[t]{0.48\linewidth}
        \centering
        \includegraphics[width=\linewidth]{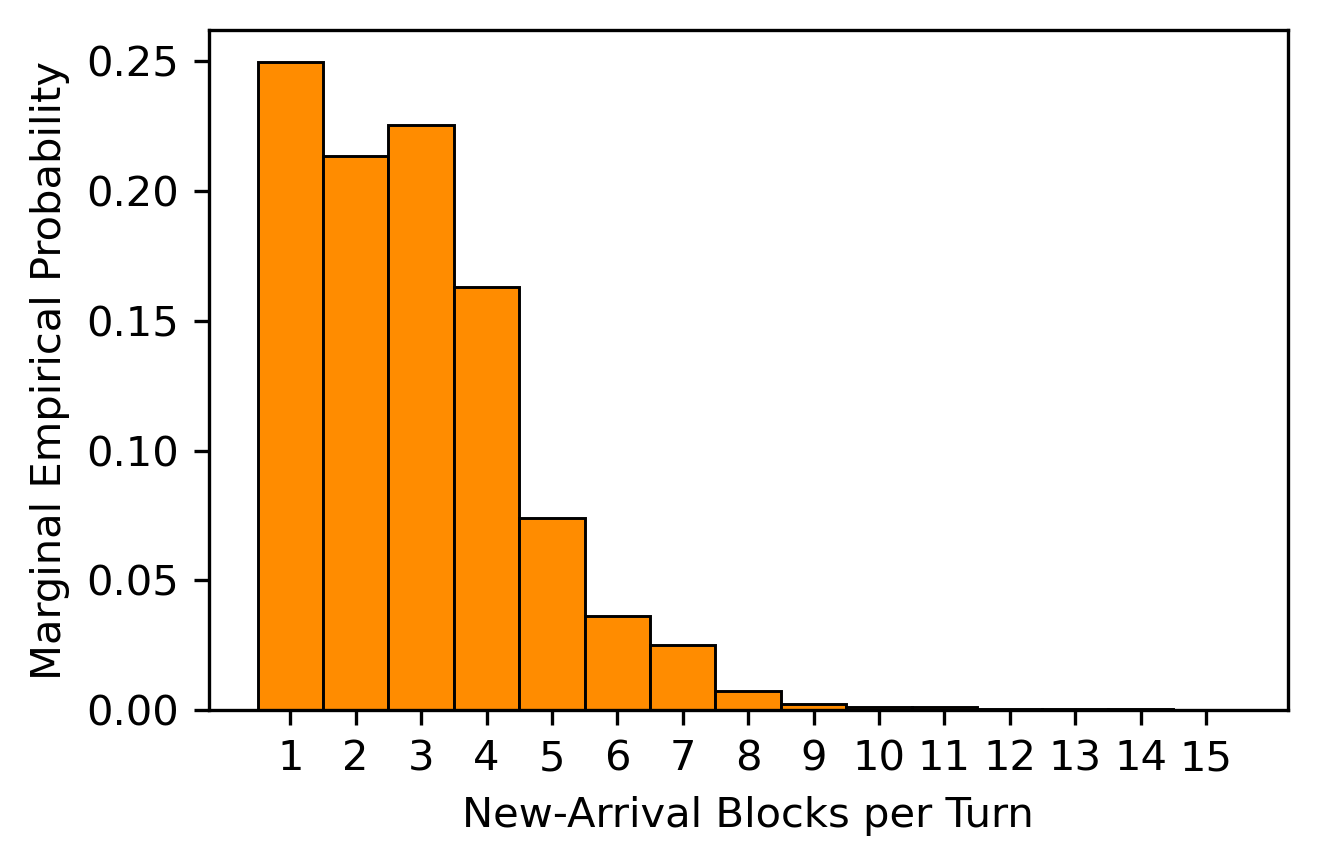}
        \caption{New KV blocks per turn.}
        \label{fig: E1_new_kv_block_histogram}
    \end{subfigure}
    \caption{Empirical characteristics of the ShareGPT measurement set, which contains 5,000 multi-turn conversations and 22,347 turns.
    (a) Empirical probability mass function of the number of turns per conversation, with mean 4.4694, variance 7.9659, and maximum 37.
    (b) Marginal empirical probability mass function of the number of new KV blocks per turn. For turn $j$, this quantity combines the tail tokens and response tokens from turn $j-1$ with the full prompt tokens from turn $j$.}
    \label{fig: ShareGPT Dataset info}
\end{figure}

Because ShareGPT provides both prompts and recorded responses, we replay the original conversation history rather than feeding newly generated text into subsequent turns. 
To reduce decode interference without changing the prefilled conversation history in the trace, we force each decode procedure to generate one token, and its generated text is discarded. 
Each later prompt instead contains the corresponding recorded response.

For turn $j$ of conversation $i$, let $W_{ij}^{\mathrm{E2E}}$ denote the end-to-end (E2E) time required to complete the preceding request, and independently sample $W_{ij}^{\mathrm{sample}}\sim\Exp(1/\mu_F)$.
Because the next prompt cannot be submitted before the preceding request completes, its realized IAT is
\begin{equation*}
    W_{ij}=\max\{W_{ij}^{\mathrm{E2E}},W_{ij}^{\mathrm{sample}}\}.
\end{equation*}
The estimator uses the target exponential CDF $\widehat F(t)=1-\exp(-t/\mu_F)$.

\subsubsection{Experimental Results.}
\label{ssec: Qwen3-32B Model and ShareGPT Workload}

To verify the accuracy of our proposed hit ratio estimator,
we report the absolute error $\mathrm{AE}=|\text{estimated hit ratio}-\text{measured hit ratio}|$ and the relative error $\mathrm{RE}=\mathrm{AE}/\text{measured hit ratio}\times100\%$.
We evaluate the estimator in two complementary sweeps. 
In \cref{fig: Qwen32 fixed arrival analysis}, we fix the conversation-arrival rate at $\lambda_0=1$ per second and vary the target mean IAT $\mu_F$.
In \cref{fig: Qwen32 fixed IAT analysis}, we instead fix $\mu_F=160$ s and vary $\lambda_0$. 
Across both sweeps, the absolute error remains below $0.02$ and the relative error remains below $5\%$. Thus, the estimator accurately captures the effects of both revisit speed and conversation-arrival load on the LRU hit ratio.

\begin{figure}[ht]
    \centering
    \begin{subfigure}[t]{0.45\linewidth}
        \centering
        \includegraphics[width=\linewidth]{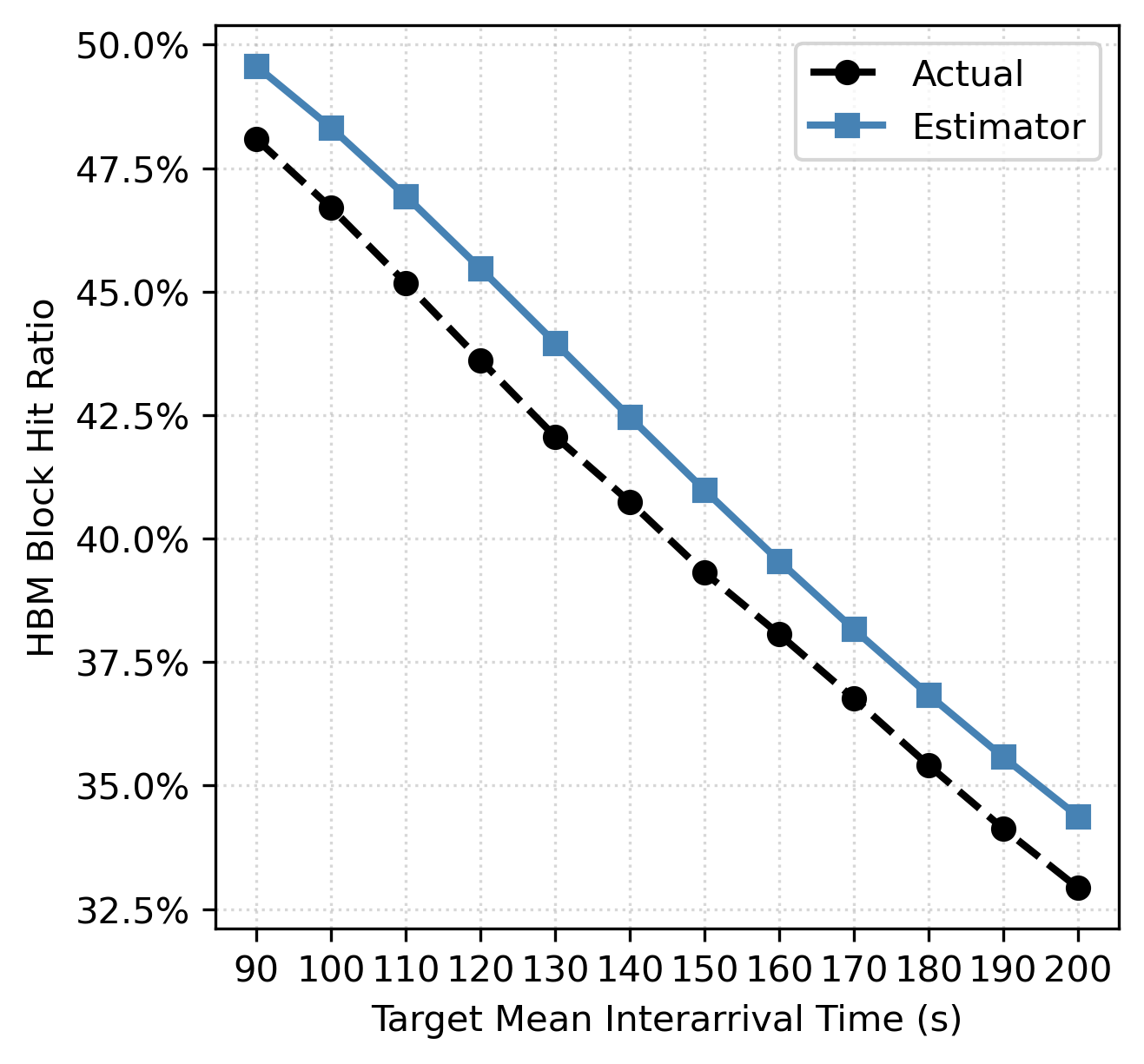}
        \caption{Measured and estimated hit ratios.}
        \label{fig: E3_fixed_ar1_hit_ratio}
    \end{subfigure}
    \hfill
    \begin{subfigure}[t]{0.45\linewidth}
        \centering
        \includegraphics[width=\linewidth]{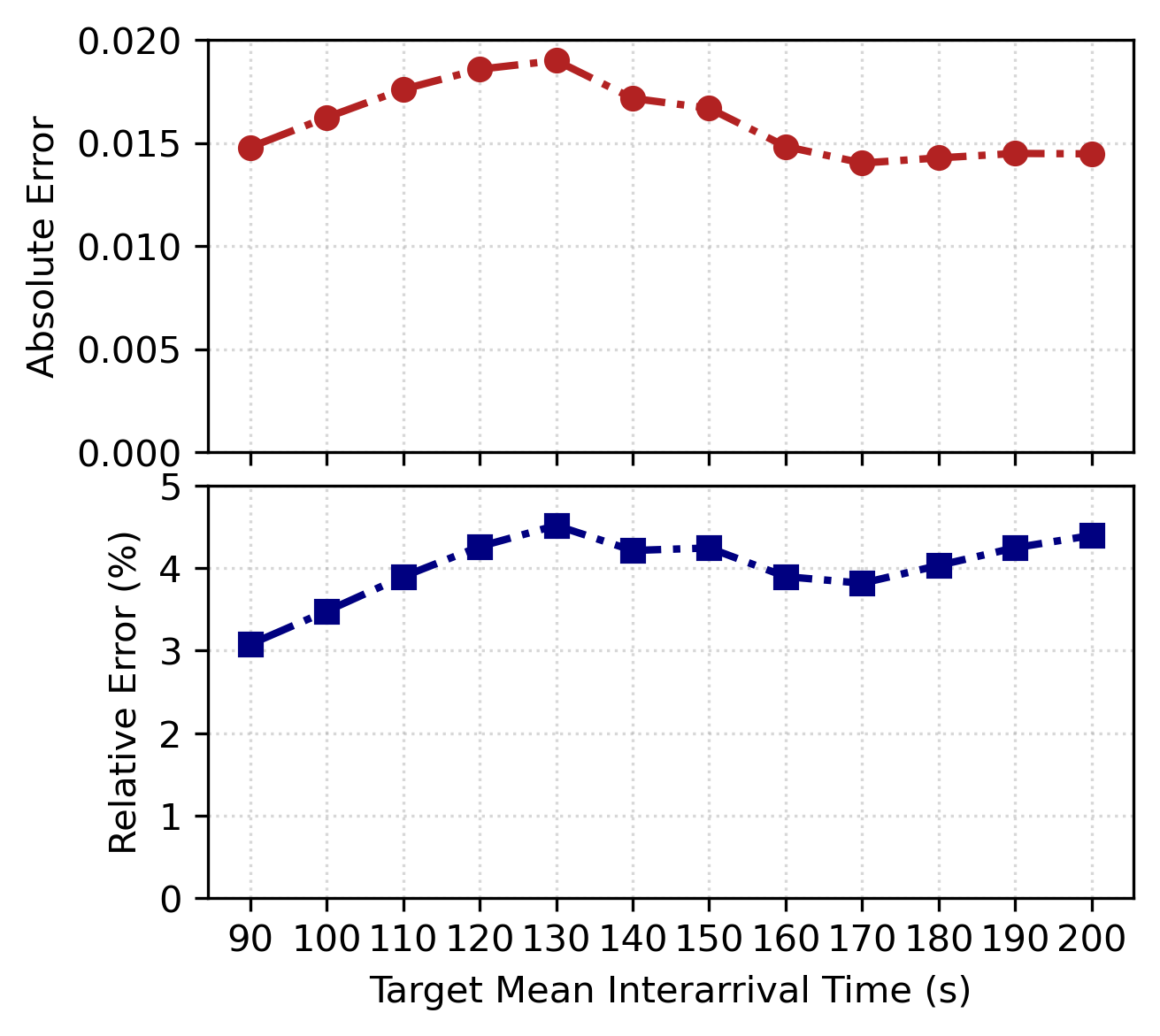}
        \caption{Absolute and relative errors.}
        \label{fig: E3_fixed_ar1_estimator_errors}
    \end{subfigure}
    \caption{Qwen3-32B hit ratio estimation with fixed conversation-arrival rate $\lambda_0=1$ per second and varying target mean IAT $\mu_F$.}
    \label{fig: Qwen32 fixed arrival analysis}
\end{figure}

\begin{figure}[ht]
    \centering
    \begin{subfigure}[t]{0.45\linewidth}
        \centering
        \includegraphics[width=\linewidth]{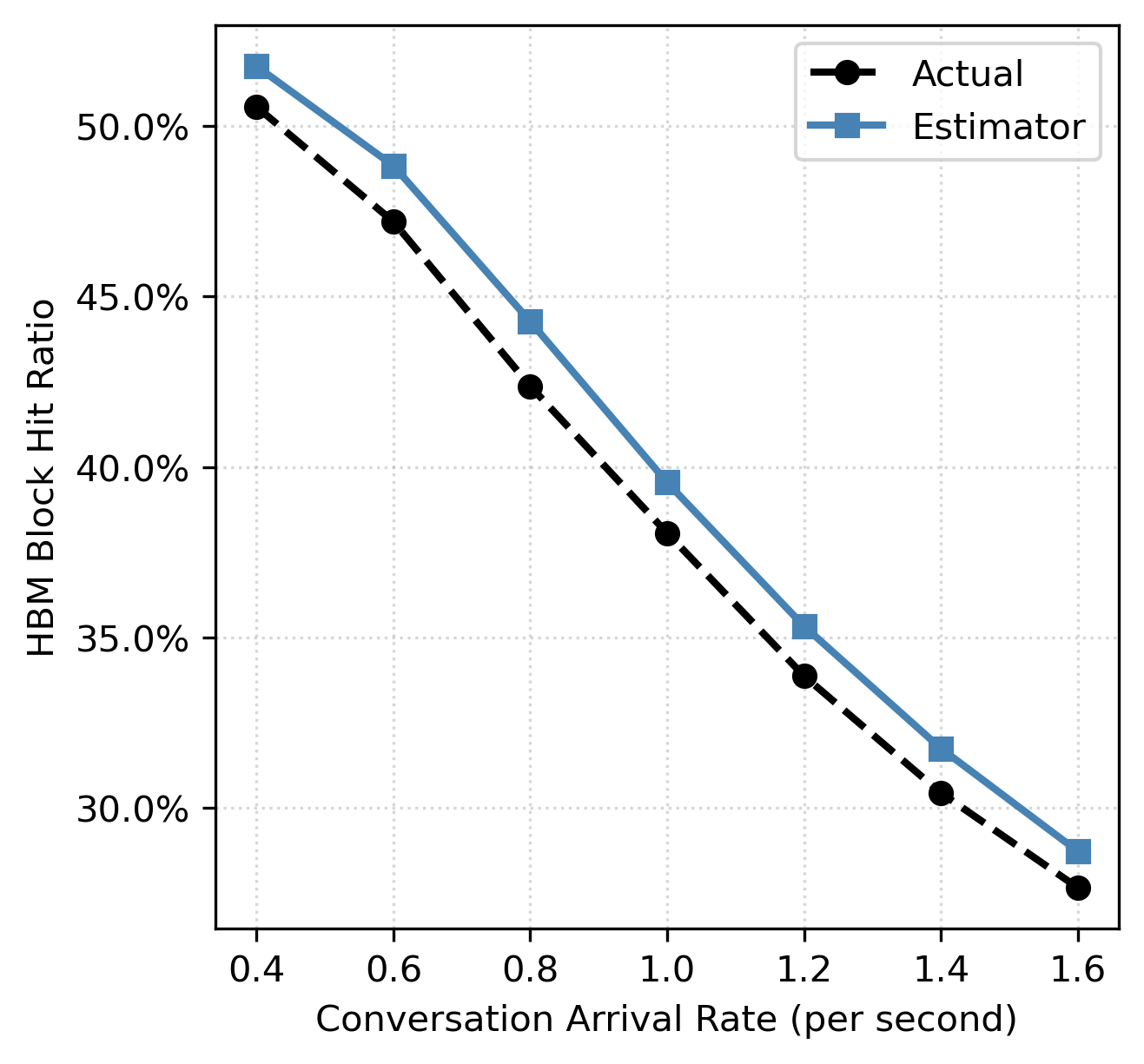}
        \caption{Measured and estimated hit ratios.}
        \label{fig: E3_fixed_mit160_hit_ratio}
    \end{subfigure}
    \hfill
    \begin{subfigure}[t]{0.45\linewidth}
        \centering
        \includegraphics[width=\linewidth]{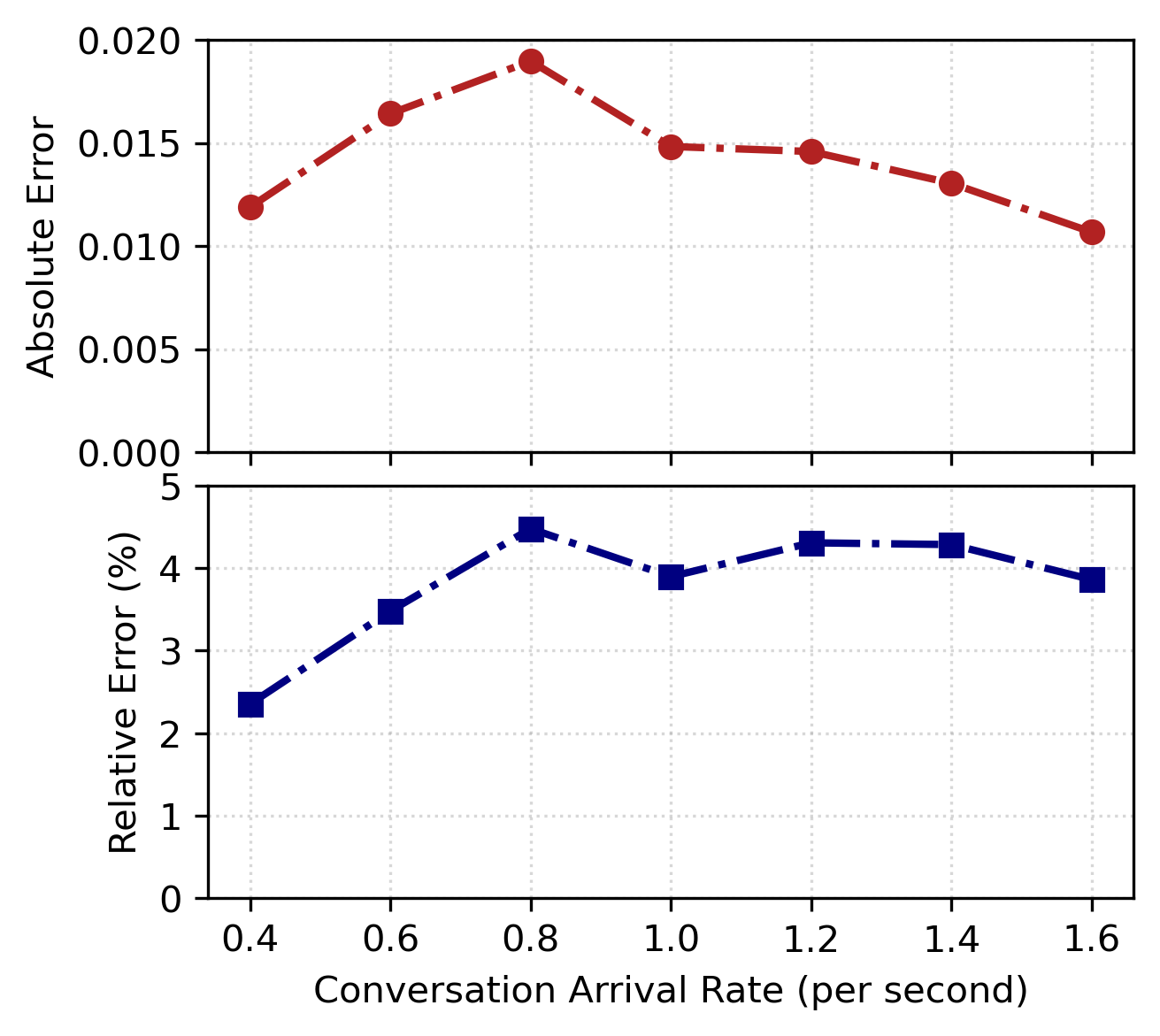}
        \caption{Absolute and relative errors.}
        \label{fig: E3_fixed_mit160_estimator_errors}
    \end{subfigure}
    \caption{Qwen3-32B hit ratio estimation with fixed target mean IAT $\mu_F=160$ s and varying conversation-arrival rate $\lambda_0$.}
    \label{fig: Qwen32 fixed IAT analysis}
\end{figure}

\subsection{WAQ-LRU Evaluation under Heterogeneous Workflows}
\label{ssec: WAQ-LRU estimator real practice}

We next empirically evaluate the performance of our WAQ-LRU policy.
% whether the quota allocation selected by our mean-field estimator improves both the HBM hit ratio and serving latency in a heterogeneous-workflow system. 
We consider three workflow mixtures and present the mixture of WildChat Legal/Business/Finance and Math/Science conversations in this section.
We defer the results of the other two mixtures to \cref{app: additional experiments}.

\subsubsection{Workload Description and Experimental Setup.}

We use the same experimental setup as in \cref{ssec: LRU estimate Workload Description and Experimental Setup}. Recall that the capacity is $N=4,608$.
We implement WAQ-LRU by statically assigning blocks to workflows and maintaining an independent LRU cache within each quota at the block level.

The workload combines two subsets of WildChat. Workflow 1 contains Legal/Business/Finance conversations and contributes 2,000 measurement conversations, whereas Workflow 2 contains Math/Science conversations and contributes 3,000. 
The corresponding warm-up set contains 400 and 600 conversations, respectively. 
The practical MCM estimates are given by
\begin{align*}
\text{Workflow 1:}\quad
&\widehat a_1^{\mathrm{cum}}=19.4350,\quad
\widehat a_1^{\mathrm{reu,hash}}=11.1660,\quad
\widehat b_1=8.2690,\\
\text{Workflow 2:}\quad
&\widehat a_2^{\mathrm{cum}}=21.2623,\quad
\widehat a_2^{\mathrm{reu,hash}}=11.6067,\quad
\widehat b_2=9.6557.
\end{align*}
\Cref{fig: WAQ C Dataset info} visualizes the empirical number of turns per conversation and new KV blocks per turn for the two workflows.

\begin{figure}[ht]
    \centering
    \begin{subfigure}[t]{0.48\linewidth}
        \centering
        \includegraphics[width=\linewidth]{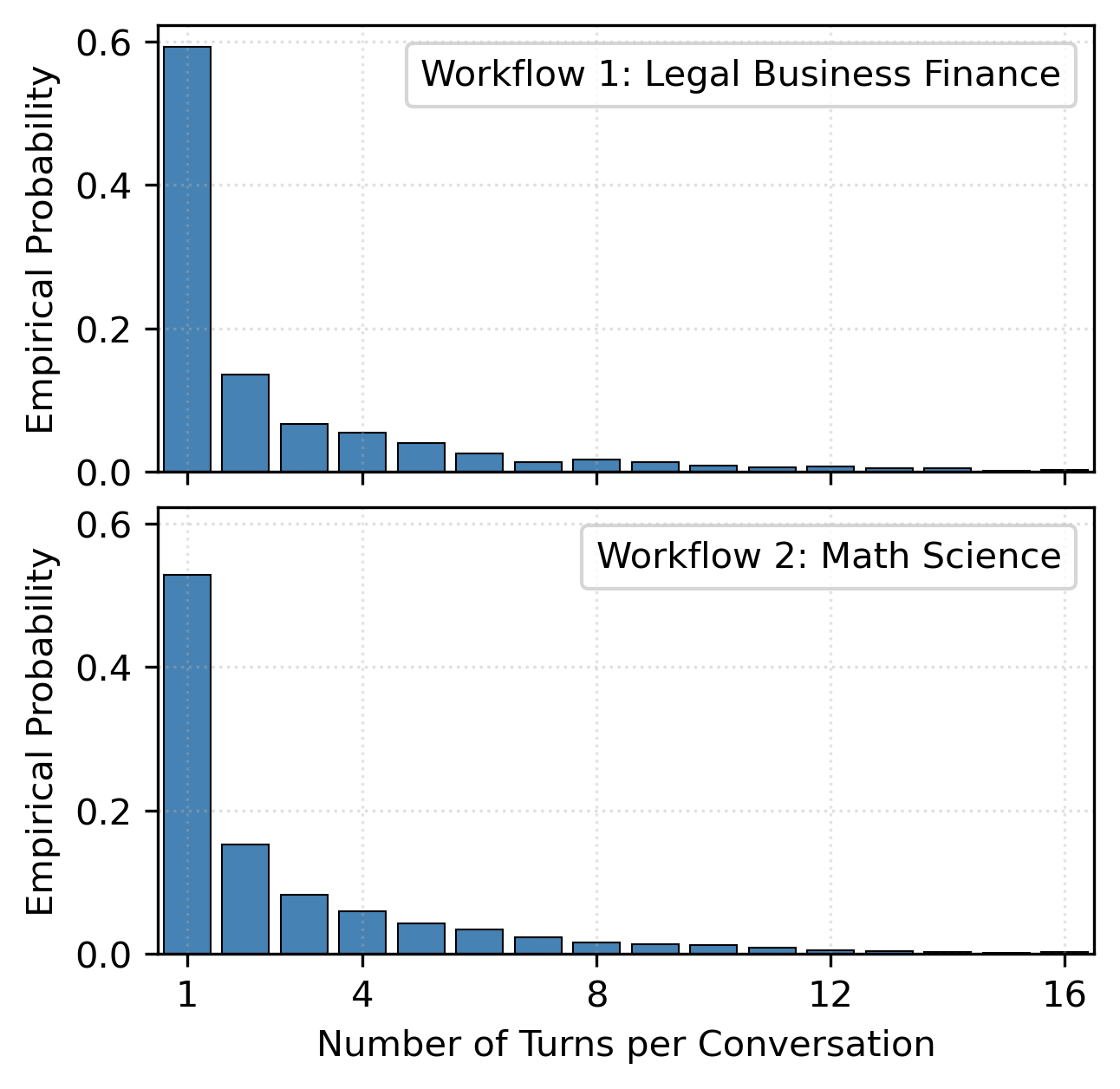}
        \caption{Number of turns per conversation.}
        \label{fig: C turn count distributions}
    \end{subfigure}
    \hfill
    \begin{subfigure}[t]{0.48\linewidth}
        \centering
        \includegraphics[width=\linewidth]{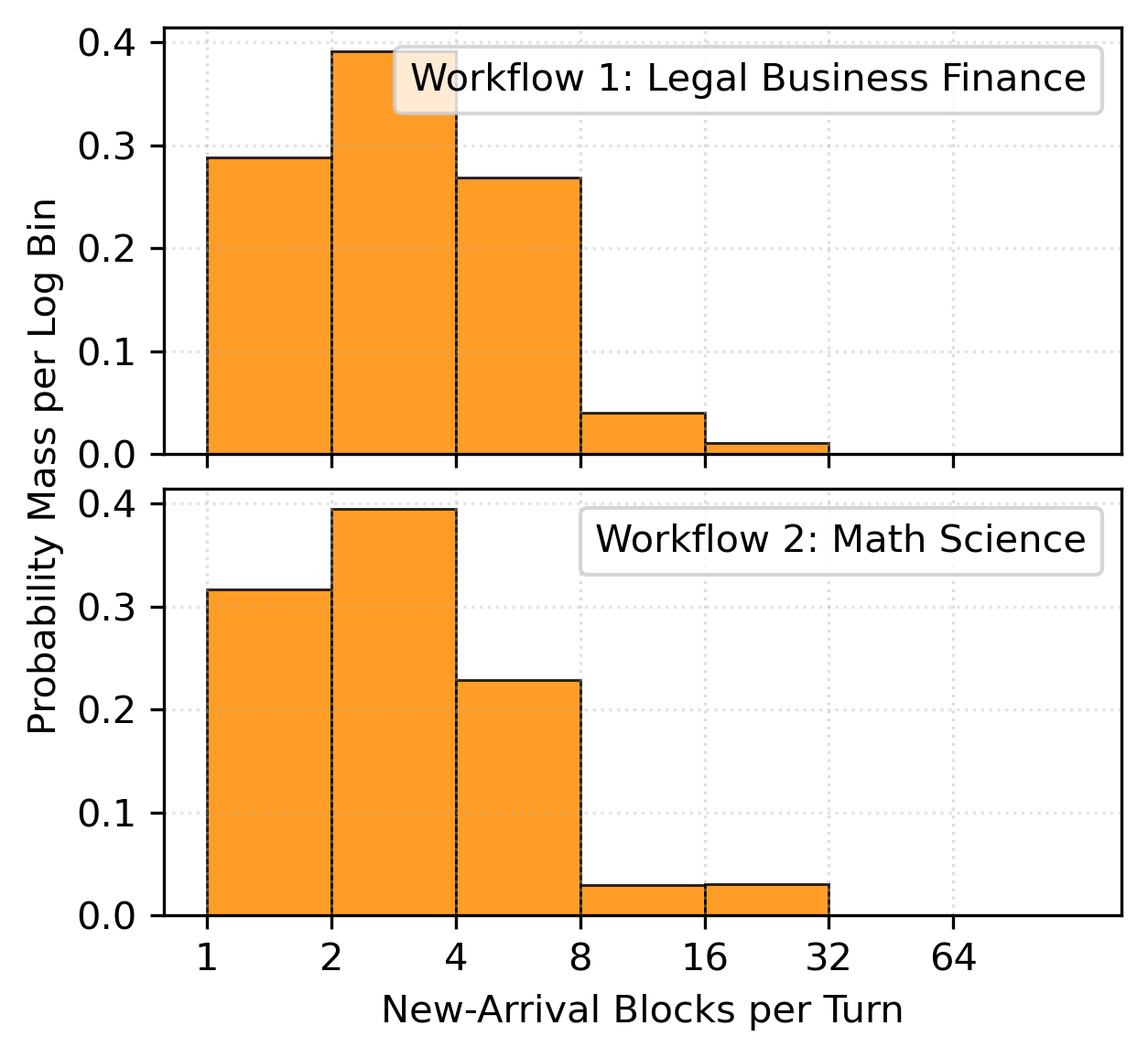}
        \caption{New KV blocks per turn.}
        \label{fig: C new KV block histograms}
    \end{subfigure}
    \caption{Empirical characteristics of the two WildChat workflows in the measurement set, which contains 5,000 conversations and 13,644 turns. (a) Empirical probability mass functions of the number of turns, displayed through turn 16 for readability. Workflow 1 has mean $2.5785$, variance $9.9088$, and maximum $31$; Workflow 2 has mean $2.8290$, variance $11.8858$, and maximum $45$. 
    % The omitted probability masses beyond turn 16 are $0.90\%$ and $1.13\%$, respectively. 
    (b) Marginal empirical probability mass functions of the number of new KV blocks per turn, shown in logarithmic bins.}
    \label{fig: WAQ C Dataset info}
\end{figure}

Let $r$ denote the total conversation arrival rate. We set the rates of the two independent Poisson conversation arrival processes to $\widehat\lambda_1(r)=0.4r$ and $\widehat\lambda_2(r)=0.6r$. The two workflows have exponential target IATs with respective means $\widehat\mu_1(r)=\frac{320}{r}\ \mathrm{s}$ and $\widehat\mu_2(r)=\frac{1280}{r}\ \mathrm{s}$.
Thus, increasing $r$ simultaneously increases the conversation arrival rates and decreases the mean IATs. 
This scaling varies the load of the system but leaves the mean-field optimization problem unchanged.
To be specific, after the change of variable $u=rt$, both practical displacement equations are independent of $r$.
Hence, the estimated TTLs scale as $1/r$, while the values of the exponential CDFs at those TTLs remain constant. 
Consequently, the estimated limiting hit ratios and the optimal quota vector are invariant across the sweep.
By maximizing the estimator in \cref{equ: practical WAQ hit ratio estimator} we obtain the unique solution
\begin{equation*}
    (N_1^*,N_2^*)=(3,575,1,032),
    \qquad \text{and}\qquad
    \widehat{\vec\alpha}^{\,*}=(0.7760,0.2240).
\end{equation*}
The estimated limiting hit ratios are $17.73\%$ under LRU and $21.41\%$ under the optimized WAQ-LRU, corresponding to a predicted relative improvement of $20.71\%$.

We sweep $r$ from $2.4$ to $10.4$ conversations per second in increments of $0.4$. Each run starts with a fresh serving stack and an empty cache, followed by 1,000 warm-up conversations and 5,000 measurement conversations. For each value of $r$, the LRU and WAQ-LRU runs use the same conversations, planned arrival traces, and random seeds. The measurement window extends from the first to the last measurement-conversation submission.
Right censoring is more pronounced here than in \cref{ssec: LRU estimator real practice}, because the higher conversation arrival rates shorten the measurement window and Workflow 2 has a relatively long mean IAT. 
For instance, at $r=2.4$, only $9,149$ of the $13,644$ scheduled measurement turns are executed under LRU before the window closes. 
This truncation changes the realized workflow and turn-state composition, contributing to the gap between stationary predictions and measured hit ratios.
Therefore, in these experiments, we mainly focus on whether the quota selected by our mean-field model improves the empirical hit ratio and TTFT over LRU.
The accuracy of our estimator is evaluated separately in \cref{ssec: LRU estimator real practice}.

As in \cref{ssec: LRU estimator real practice}, each decoder generates one token, while subsequent turns replay the corresponding recorded response. 
To ensure that TTFT primarily reflects the HBM cache hit ratio, we disable reusable CPU KV caching in our experiments. This prevents CPU cache hits from reducing TTFT.
Therefore, the reported hit ratios measure HBM cache reuse, and the relationship between the cache policy, HBM hit ratio, and TTFT is not confounded by CPU cache reuse.

\subsubsection{Experimental Results.}

\Cref{fig: WAQ C experimental results} compares LRU and WAQ-LRU across the arrival rate sweep. 
First, \cref{fig: C measured hit ratios} shows that both policies maintain relatively stable hit ratios when the total conversation arrival rate is at most $8.4$ per second, which we refer to as the stable regime.
This insensitivity to load is consistent with our mean-field prediction, because the conversation arrival rates and prompt IATs are scaled proportionally.
Throughout this stable regime, WAQ-LRU consistently outperforms LRU, improving the empirical hit ratio by approximately $13\%$ in relative terms. 
Under high load, where the conversation arrival rate is over $8.4$ per second, the excessive serving pressure causes a substantial request backlog.
Consequently, the hit ratios of both policies begin to decrease.
Nevertheless, the theoretically optimized WAQ-LRU quota remains effective.
Its hit ratio declines considerably more slowly and maintains a clear advantage over LRU even when the LRU hit ratio drops rapidly.

We further show in \cref{fig: C mean TTFT,fig: C TTFT improvement rate} how the hit ratio advantage translates into a TTFT improvement. We define
\begin{equation*}
    \text{TTFT improvement rate}
    :=\frac{\overline{\mathrm{TTFT}}^{\LRU}-\overline{\mathrm{TTFT}}^{\mathrm{WAQ}}}
    {\overline{\mathrm{TTFT}}^{\LRU}}\times 100\%,
\end{equation*}
where the bars denote the empirical means under the corresponding policies.
At low load (total arrival rate below $7$), the system has sufficient compute capacity, and therefore the reduction in prefill work produces a modest latency benefit.
As the load increases and requests begin to queue, the saved prefill computation prevents additional work from accumulating in the queue, which substantially amplifies the latency benefit. 
Specifically, when the total arrival rate is at least $6.8$, WAQ-LRU reduces the mean TTFT by at least $10\%$.
In particular, at $r=8.4$, the mean TTFT decreases from $6,343$ ms under LRU to $2,915$ ms under WAQ-LRU, with an improvement rate of $54.1\%$.

These results further empirically show that a hit ratio improvement of only approximately $12\%$--$14\%$ in relative terms can produce a much larger TTFT improvement under relatively high loads.
Moreover, the persistent advantage of WAQ-LRU in the high-load regime suggests that the quota selected by our mean-field estimator can remain robust even when the serving system departs from the stationary regime.

\begin{figure}[ht]
    \centering
    \begin{subfigure}[t]{0.32\linewidth}
        \centering
        \includegraphics[width=\linewidth]{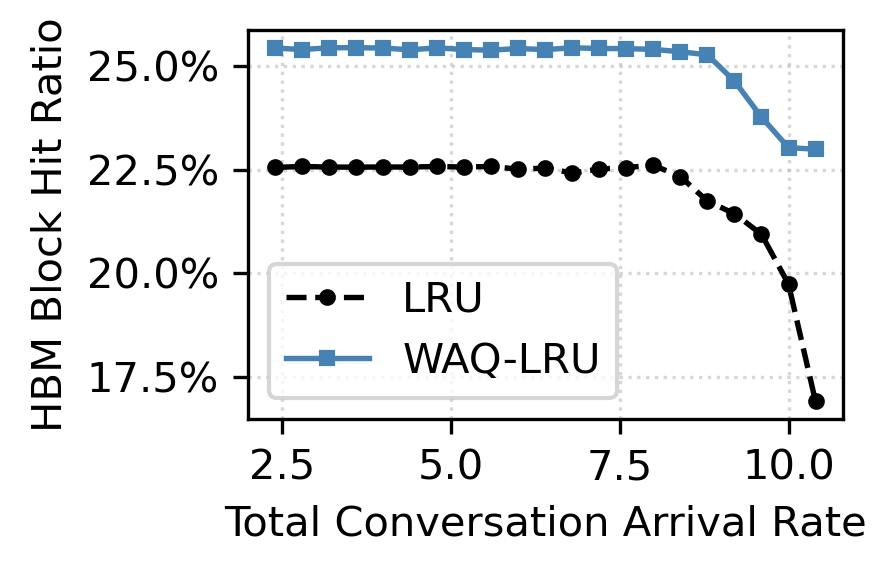}
        \caption{Measured HBM hit ratios.}
        \label{fig: C measured hit ratios}
    \end{subfigure}
    \hfill
    \begin{subfigure}[t]{0.32\linewidth}
        \centering
        \includegraphics[width=\linewidth]{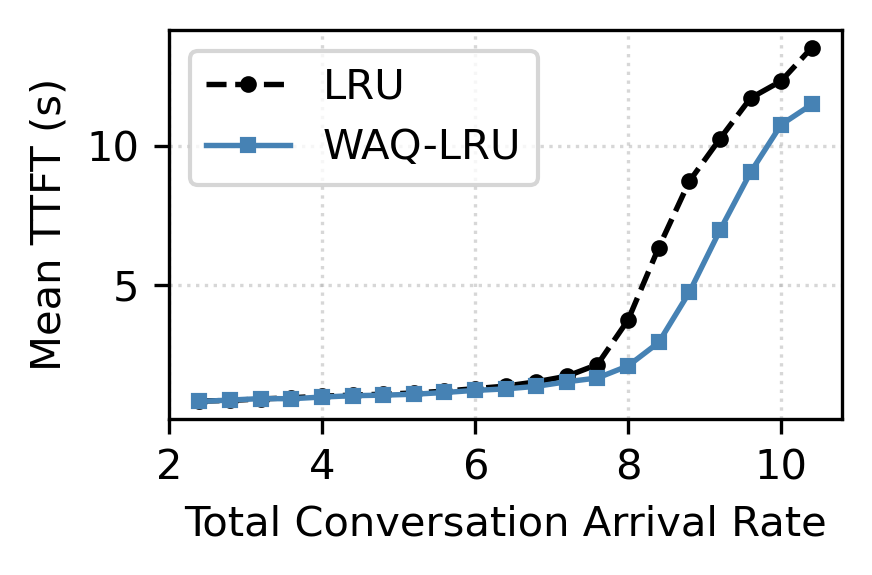}
        \caption{Mean TTFT on a linear scale.}
        \label{fig: C mean TTFT}
    \end{subfigure}
    \hfill
    \begin{subfigure}[t]{0.32\linewidth}
        \centering
        \includegraphics[width=\linewidth]{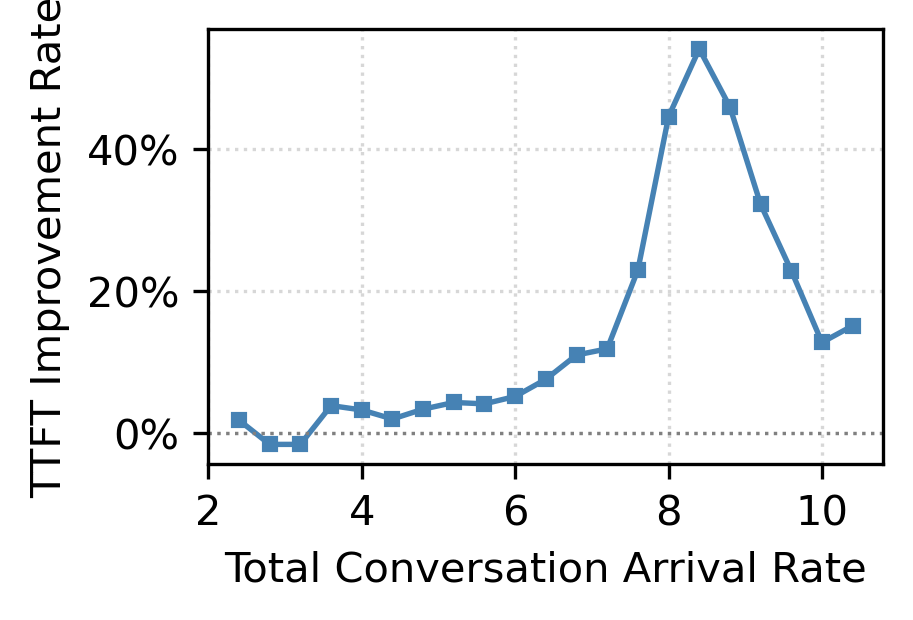}
        \caption{Mean TTFT improvement rate of WAQ-LRU.}
        \label{fig: C TTFT improvement rate}
    \end{subfigure}
    \caption{Comparison of LRU and WAQ-LRU for the WildChat Legal/Business/Finance and Math/Science workflow mixture. The total conversation arrival rate is increased while the workflow specific mean IATs are decreased proportionally.}
    \label{fig: WAQ C experimental results}
\end{figure}

% \section{Conclusion}
%
% We studied the characterization of the hit ratio in LLM serving systems with multi-turn conversations. We proposed a parsimonious multi-turn conversation model under the widely adopted least-recently-used policy, which captures several key components governing the real system dynamics. Through a mean-field asymptotic framework, we proved the convergence of the hit ratios to a certain mean-field hit ratio limit as the HBM capacity $N\to \infty$. Based on our characterization of the mean-field limit, we designed a hit ratio estimator for practical LLM serving systems, and validated its accuracy via LLM serving experiments with the Qwen3-8B, Qwen3-4B, and Qwen3-32B models implemented on Ascend 910B2 NPUs. Our work provides a theoretical foundation for analyzing multi-turn LLM services, and provides practical system-design guidelines on memory provisioning that achieve a target hit ratio. 

\section{Conclusion}
\label{sec: conclusion}

We studied the characterization and optimization of the hit ratio in LLM serving systems with multi-turn conversations. We formulated the MCM, which captures conversation arrival rates, IATs between turns, the total number of turns, and the growing KV-block lengths. Through a mean-field asymptotic framework, we proved that the hit ratio of the MCM under the widely used LRU policy converges to a deterministic limit as the conversation arrival rate and HBM memory capacity scale proportionally. This limit is characterized through a scalar equation for the mean-field TTL, giving an analytical description of how workload distributions and memory capacity determine the hit ratio.

The above descriptive characterization provides a basis for prescriptive memory allocation.
For heterogeneous workflows, the mean-field analysis reveals that LRU induces an implicit allocation of cache capacity across workflows that may not lead to the maximal hit ratio. We proposed WAQ-LRU to make this allocation an explicit control decision while retaining LRU within each workflow. Optimizing the quota vector maximizes the limiting hit ratio within the WAQ-LRU family and yields a limiting hit ratio no less than that of LRU.

We extended the analysis to account for unhashable tail blocks and block-level eviction, and constructed practical hit ratio estimators for LRU and WAQ-LRU. Experiments with Qwen3-4B, Qwen3-8B, and Qwen3-32B on Ascend 910B2 NPUs, using ShareGPT and WildChat conversations, validate the LRU estimator. Further experiments with Qwen3-32B show that quotas selected using the WAQ-LRU estimator improve the hit ratio and TTFT across three heterogeneous workflow mixtures. These experiments also show that moderate improvements in cache reuse can produce substantially larger latency reductions as serving load increases.

Future work may include extending our framework to KV cache management in multi-turn LLM serving with multiple storage layers, including HBM, DRAM, and SSD. Another direction is to model KV cache reuse across different conversations that share prompt prefixes. In this setting, a block may remain useful after the conversation that generated it terminates, and the value of retaining it depends on requests from multiple conversations. Furthermore, jointly studying these two extensions could guide cache allocation and eviction across storage layers while accounting for reuse both within and across conversations.

\paragraph{\textbf{Code and Data Availability.}}
The code and experimental artifacts for this paper are publicly available at
\url{https://github.com/Heyuan-Yao/mcm-waq-lru-artifact},
with the complete request-level results provided in the associated v1.0.0 release.

\paragraph{\textbf{Acknowledgment.}}
We would like to express special appreciation to Izzy Grosof for fruitful personal discussions on the paper. 
% on the design of the experiments. We would also like to thank Yuhan Chen, Yuxi Chen, Zhongrui Chen, Jin Guang, Yige Hong, Zhouzi Li, Yining Song, Ziyuan Wang, and Yaosheng Xu, for helpful discussions and comments.

% Literature references are maintained in the references.bib file.
\bibliographystyle{informs2014}
\bibliography{references}

\clearpage
\appendix

\section{Proofs of Results in \cref{sec: main results}}
\label{sec: Proofs}

In this section, we provide the proofs of our single-workflow MCM-LRU mean-field results. To be specific, the normalized background displacement process convergence (\cref{thm: Displacement Process Convergence}), the tagged-$j$ TTL convergence (\cref{thm: Characteristic Time Convergence- Full}), and the hit ratio convergence (\cref{thm: Hit Ratio Convergence Full}), are demonstrated in \cref{ssec: Proof of the full version of the sample path convergence,ssec: Proof of the Full Version of CT convergence,ssec: Proof of the Full Version of hit ratio conv MCM-LRU short}, respectively.
The mean-field practical single- and mixed-workflow extensions will be stated and proved in \cref{app: direct block-level practical MCM}.

We first provide equivalent expressions for the expected cumulative requested workload $a^{\mathrm{cum}}$ and the expected reusable history workload $a^{\mathrm{reu}}$, which will be used throughout the proofs. By the Tonelli theorem and the independence of $M$ from the increment sequence,
\begin{align}
    {a^{\mathrm{cum}}}
    &{=\sum\limits_{j\geq1}s_j\,r_j,}
    \label{equ: cumulative requested workload identity}\\
    {a^{\mathrm{reu}}}
    &{=\sum\limits_{j\geq1}s_{j+1}\,r_j.}
    \label{equ: number of blocks eligible for reuse}
\end{align}

\subsection{Proof of \cref{thm: Displacement Process Convergence}}
\label{ssec: Proof of the full version of the sample path convergence}

% Nonetheless, we provide a proof sketch.

% {\color{blue}\noindent\textit{Proof Sketch:} We first represent the stationary background conversations using a marked Poisson random measure, expressing the displacement process as an integral of their cumulative displacement contributions. We then derive the mean displacement curve $d(t)$ by separately evaluating the expected contributions from pre-existing history sessions and newly generated KV blocks. The Poisson random measure representation further yields a pointwise second-moment bound of order $O(N^{-1})$ for the normalized displacement process. Finally, using the nondecreasing property of $D^{(N)}(t)$ and $d(t)$, the Lipschitz continuity of $d(t)$, and a discretization argument over a finite time horizon with granularity $O(N^{\frac{1}{3}})$, we obtain the uniform convergence. \hfill $\square$}

We first prove \cref{thm: Displacement Process Convergence}, the sample-path convergence of the background displacement process.
Our proof contains three steps. 
In \cref{sssec: Representation of the Displacement Process}, we first construct the marked Poisson random measure (PRM), establish Palm invariance of the background process, and derive a single-workflow variance bound.
In \cref{sssec: Displacement Process Decomposition}, we decompose it into the pre-existing history-session contribution and the newly generated KV-block contribution. 
We next derive their respective mean curves in \cref{lem: Initial-prefix contribution,lem: Mean New Blocks Contribution}, which sum to the mean displacement curve $d(t)$ in \cref{lem: Mean Displacement Curve}. 
Finally, in \cref{sssec: normalized displacement sample path convergence}, we prove \cref{thm: Displacement Process Convergence} by combining \cref{lem: single workflow PRM variance bound,lem: Mean Displacement Curve}.

\subsubsection{Poisson Random Measure Representation of the Displacement Process.}
\label{sssec: Representation of the Displacement Process}

We first construct a marked Poisson random measure for the mixed-workflow MCM. 
We will next establish Palm invariance and the required single-workflow variance bound in \cref{lem: Palm invariance of background displacement,lem: single workflow PRM variance bound}.

For each workflow $c\in\mathcal C$, define the trajectory mark space $\mathcal M_c:=\mathbb N_+\times\mathbb N_+^{\mathbb N_+}\times\mathbb R_+^{\mathbb N_+}$,
and let
$\zeta_c:=\left(M_c,\boldsymbol Z_c,\boldsymbol W_c\right)
=\left(M_c,\{Z_{c,\ell}\}_{\ell\geq1},\{W_{c,\ell}\}_{\ell\geq1}\right)\in\mathcal M_c.$
We note that only the coordinates up to the terminal turn $M_c$ are used. Let $\pro_c:=\mathcal L\!\left(M_c,\{Z_{c,\ell}\}_{\ell\geq1},\{W_{c,\ell}\}_{\ell\geq1}\right)$
denote the probability distribution of a workflow-$c$ trajectory mark.
A marked point $z=(c,a,\zeta_c)$ records the workflow label, its first-prompt arrival time $a\in\mathbb R$, and its entire trajectory.

We then define the mixed-workflow marked space as the disjoint union $\mathcal Z^{\mathrm{mix}}
:=\bigsqcup_{c\in\mathcal C}(\{c\}\times\mathbb R\times\mathcal M_c).$
Because the $J$ workflow arrival processes are independent Poisson processes, their superposition forms the Poisson random measure $\mathcal N_N^{\mathrm{mix}}:=\sum\limits_i\delta_{(C_i,A_i,\zeta_i)}$ on $\mathcal Z^{\mathrm{mix}}$, with intensity measure
\begin{equation*}
N\nu^{\mathrm{mix}}(d\gamma,da,d\zeta)
:=N\sum\limits_{c\in\mathcal C}\lambda_c\,\delta_c(d\gamma)\,da\,\pro_c(d\zeta).
\end{equation*}

Recall that we assume the session of the tagged conversation is put on the MRU end at time $0$, and impose a virtual condition that the tagged conversation is not accessed again.
To track the total number of blocks moved to the MRU end during the time interval $[0,T]$, we first study the contribution of each conversation in stationarity.
For a revealed trajectory $z=(c,a,\zeta_c)\in\mathcal Z^{\mathrm{mix}}$, 
we let $\Phi_t(z)$ denote the cumulative block mass contributed by this conversation to the displacement during $(0,t]$.
Note that a pre-existing history block is counted when it first crosses the tagged session, whereas a subsequently generated block is counted when it is appended. 
As a result, the system-wide and the workflow-specific Poisson-integral functionals are respectively given by
\begin{equation*}
\mathscr D^{(N)}(t):=\int_{\mathcal Z^{\mathrm{mix}}}\Phi_t(z)\,\mathcal N_N^{\mathrm{mix}}(dz),
\qquad
\mathscr D_c^{(N)}(t):=\int_{\mathcal Z^{\mathrm{mix}}}\one\{C(z)=c\}\Phi_t(z)\,\mathcal N_N^{\mathrm{mix}}(dz).
\end{equation*}

We note that $\mathscr D^{(N)}(t)$ is the stationary process without our virtual condition. 
We then show in \cref{lem: Palm invariance of background displacement} that under our virtual condition, the system with remaining conversations still forms a stationary process, following the same law as $\mathscr D^{(N)}(t)$.

\begin{lemma}[Palm Invariance of the Background Conversation Process]
\label{lem: Palm invariance of background displacement}
Fix a workflow $c\in\mathcal C$ and a turn state $j\geq 0$. 
Under the Palm distribution associated with a workflow-$c$, turn-$j$ access at time zero, 
we remove the entire tagged conversation and denote the remaining marked process by $\mathcal N_{N}^{\mathrm{bg}}$. Then $\mathcal N_N^{\mathrm{bg}}$ is still the marked PRM:
\begin{equation}
\label{equ: Palm background invariance}
\mathcal N_N^{\mathrm{bg}}\stackrel{\mathcal L}{=}\mathcal N_N^{\mathrm{mix}}.
\end{equation}
Moreover, $\mathcal N_N^{\mathrm{bg}}$ is independent of the tagged conversation mark. 
As a result,
\begin{equation}
\label{equ: Palm displacement representation}
D^{(N)}(t)\stackrel{\mathcal L}{=}\mathscr D^{(N)}(t),
\qquad
D_c^{(N)}(t)\stackrel{\mathcal L}{=}\mathscr D_c^{(N)}(t),
\end{equation}
and the stationary law of the background displacement process does not depend on the tagged workflow or turn state.
\end{lemma}

\begin{proof}
    \Cref{equ: Palm background invariance} is a direct consequence of Slivnyak's theorem (see, e.g., Theorem 9.4 in \cite{last2017lectures}). 
    As a result, the laws of server-side and the workflow-specific displacement processes can be characterized by the PRM representation in \cref{equ: Palm displacement representation}. 
    The independence of the tagged conversation session block length and the background displacement holds directly because the mark distributions are independent.
% Apply Slivnyak's theorem to the workflow-$r$ marked Poisson process. After removing the tagged conversation, the remaining workflow-$r$ atoms retain their original Poisson law, while the other workflow processes are unchanged by independence. The tagged turn-$j$ epoch differs from its conversation-arrival epoch by a random shift determined only by the tagged mark. Because the remaining marked process is stationary and independent of that mark, translating this epoch to time zero does not alter its law. This proves \cref{equ: Palm background invariance}, the asserted independence, and hence \cref{equ: Palm displacement representation}.
\end{proof}

To control the variation of the process, we next introduce the notation used after specializing this PRM representation to the single-workflow model in \cref{sec: main results}. 
Set $J=1$, write $\lambda_1=\lambda_0$, $\mathcal M:=\mathcal M_1$, and $\pro_\zeta:=\mathcal L\!\left(M,\{Z_\ell\}_{\ell\geq1},\{W_\ell\}_{\ell\geq1}\right)$.
Moreover, we suppress the workflow coordinate and write $\mathcal Z:=\mathbb R\times\mathcal M,$ $\mathcal N_N:=\sum\limits_i\delta_{(A_i,\zeta_i)},$ and $N\nu(da,d\zeta):=N\lambda_0\,da\,\pro_\zeta(d\zeta).$

For $z=(a,\zeta)$ and a fixed horizon $T<\infty$, define $L(\zeta):=\sum\limits_{\ell=1}^{M-1}W_\ell$ as the time span between the first and the last prompt arrivals of the conversation with mark $\zeta$. 
Then we can bound the contribution $\Phi_t(z)$ by $G_T(a,\zeta):=B(\zeta)\one\{-L(\zeta)<a\leq T\}.$ To be specific,
the path $\Phi_\cdot(z)$ is nonnegative and nondecreasing. 
Each block crosses the tagged session at most once, so the total contribution of $z=(a,\zeta)$ is at most $B(\zeta)$. A conversation can contribute during $[0,T]$ only if $-L(\zeta)<a\leq T$. Hence the construction satisfies
\begin{equation}
\label{equ: PRM contribution bound}
0\leq\Phi_s(z)\leq\Phi_t(z)\leq G_T(z),
\qquad 0\leq s\leq t\leq T.
\end{equation}
We therefore establish the following property of our displacement process that is needed in the remaining proof.

\begin{lemma}[Variance Bound for the Single-Workflow Poisson Integral]
\label{lem: single workflow PRM variance bound}
Fix a constant time $T>0$, for every $t\in[0,T]$, the displacement process satisfies
\begin{equation*}
    \operatorname{Var}\!\left(\frac{D^{(N)}(t)}{N}\right) \leq\frac{V_T}{N},
\end{equation*}
where
\begin{equation*}
        V_T:=\int_{\mathcal Z}G_T(z)^2\,\nu(dz)=\lambda_0T\E[B^2]+\lambda_0\mu_W\E[(M-1)B^2]<\infty.
    \end{equation*}
\end{lemma}
\begin{proof}
For a fixed trajectory mark $\zeta$, the set of arrival times for which $G_T(a,\zeta)$ is nonzero has Lebesgue measure $T+L(\zeta)$. The independence of the IAT sequence from $(M,\{Z_\ell\}_{\ell\geq1})$ therefore gives
\begin{align*}
\int_{\mathcal Z}G_T(z)^2\,\nu(dz)
=\lambda_0\E[B^2(T+L)] =\lambda_0T\E[B^2]+\lambda_0\mu_W\E[(M-1)B^2].
\end{align*}
Because $M\leq B$ almost surely and the exponential-moment condition in \cref{eq:assumption B} implies $\E[B^3]<\infty$, $\int_{\mathcal Z}G_T(z)^2\,\nu(dz)$ is finite. 
By \cref{lem: Palm invariance of background displacement}, $D^{(N)}(t)\stackrel{\mathcal L}{=}\mathscr D^{(N)}(t)$.
We then use Campbell's equation to derive the variance of $\frac{D^{(N)}(t)}{N}$, and bound it by \cref{equ: PRM contribution bound}:
\begin{equation*}
    \operatorname{Var}\!\left(\frac{D^{(N)}(t)}{N}\right)
=\frac{1}{N}\int_{\mathcal Z}\Phi_t(z)^2\,\nu(dz)
\leq\frac{1}{N}\int_{\mathcal Z}G_T(z)^2\,\nu(dz)
=\frac{V_T}{N}. \qedhere
\end{equation*}
\end{proof}

% We use the notation
% \begin{equation}
% \label{equ: Poisson Random Measure Representation}
% \mathscr D_{\mathrm{sw}}^{(N)}(t):=\int_{\mathcal Z}\Phi_t(z)\,\mathcal N_N(dz).
% \end{equation}
% The path $\Phi_\cdot(z)$ is nonnegative and nondecreasing. Each block crosses the tagged session at most once, so the total contribution of $z=(a,\zeta)$ is at most $B(\zeta)$. A conversation can contribute during $[0,T]$ only if $-L(\zeta)<a\leq T$. Hence the construction satisfies
% \begin{equation}
% \label{equ: PRM contribution bound}
% 0\leq\Phi_s(z)\leq\Phi_t(z)\leq G_T(z),
% \qquad 0\leq s\leq t\leq T.
% \end{equation}
% All marked spaces, probability laws, Poisson random measures, and contribution functionals used below have now been specified. We next establish the two probabilistic properties of this construction that are needed in the remaining proof.

\subsubsection{Displacement Process Decomposition and the Mean Displacement Curve.}
\label{sssec: Displacement Process Decomposition}

To derive the mean curve of our displacement process,
we now decompose the background displacement process $D^{(N)}(t)$ into two subprocesses:
\begin{equation*}
    D^{(N)}(t)=D_{\init}^{(N)}(t)+D_{\newb}^{(N)}(t).
\end{equation*}
Here, $D_{\init}^{(N)}(t)$ denotes the displacement contributed by pre-existing history blocks of background conversations. 
These blocks are already active at time $0$ and are moved ahead of the tagged session when their subsequent access happens, during the time interval $(0,t]$. 
Additionally, $D_{\newb}^{(N)}(t)$ denotes the displacement contributed by newly generated KV blocks during $(0,t]$.

We first characterize the mean curve of the normalized pre-existing history contribution $\frac{D_{\init}^{(N)}(t)}{N}$.

\begin{lemma}[Mean Normalized Initial-Prefix Contribution]
    \label{lem: Initial-prefix contribution}
    During the time period $(0,t]$, the total pre-existing history session mass contributed to the displacement process, $D_{\init}^{(N)}(t)$, has the following normalized mean
    \begin{equation*}
        \frac{1}{N}\E[D_{\init}^{(N)}(t)]=\lambda_0 a^{\mathrm{reu}}\int_0^t\overline F(s)\,ds,
    \end{equation*}
    where $a^{\mathrm{reu}}:=\sum\limits_{j\geq1}s_{j+1}r_j$ denotes the expected amount of reusable history session, defined in \cref{equ: number of blocks eligible for reuse}.
\end{lemma}

\begin{proof}
    Fix a turn state $\ell\geq1$. A conversation that reaches turn $\ell+1$ has a waiting interval of length $W_\ell$ after completing turn $\ell$. 

    Let $s<0$ denote the time at which this conversation completes turn $\ell$, and let its inter-turn waiting time be $W_\ell=w$. Then its waiting interval is $[s,s+w)$. For its pre-existing history session to contribute to the displacement during $(0,t]$, the conversation's next prompt must arrive no later than time $t$. Equivalently, we require
    \begin{equation*}
        s<0<s+w\leq t.
    \end{equation*}
    For a fixed $w$, the set of completion times $s$ satisfying this condition is
    \begin{equation*}
        -w<s\leq\min\{0,t-w\},
    \end{equation*}
    the admissible length of which is $\min\{t,w\}$. Therefore, when we integrate over the stationary time shifts of the marked conversation trajectories, the contribution associated with an inter-turn interval of length $w$ is weighted by $\min\{t,w\}$. 

    To determine how frequently such waiting intervals occur on the stationary time axis, we recall that conversations arrive according to a Poisson process with rate $N\lambda_0$. Among these conversations, a fraction $s_{\ell+1}=\pro(M\geq \ell+1)$ reaches turn $\ell+1$. In addition, the push-forward time $\sum\limits_{j=1}^{\ell} W_j$ for each conversation is i.i.d. Therefore, by the mapping theorem of the Poisson process (see, e.g., Theorem 5.1 in \cite{last2017lectures}), we have that, the arrival process of conversations that reach turn $\ell+1$, viewed at their turn-$\ell$ completion epochs, is a Poisson process with rate $N\lambda_0s_{\ell+1}$.

    Therefore, by Campbell's theorem, we obtain the expected number of turn-$\ell$ active conversations at time $0$ and accessed by their turn-$(\ell+1)$ prompts during $(0,t]$:
    \begin{align*}
        \E \bigg[\int_{-W_\ell}^{\min\{0,t-W_\ell\}} N\lambda_0 s_{\ell+1} ds \bigg ] &=  N\lambda_0s_{\ell+1}\E[\min\{t,W_\ell\}]\\
        & =N\lambda_0s_{\ell+1}\int_0^t\pro(W_\ell>u)\,du \\&=N\lambda_0s_{\ell+1}\int_0^t\overline F(u)\,du.
    \end{align*}

    At the terminating prompt of each such interval, its pre-existing history session of size $R_\ell$ is moved in front of the tagged session. Because $W_\ell$ is independent of the conversation length and block increments, the expected turn-$\ell$ history contribution is therefore
    \begin{equation*}
        N\lambda_0s_{\ell+1}r_\ell\int_0^t\overline F(u)\,du.
    \end{equation*}
    Summing over all  $\ell\geq1$ and recalling that $a^{\mathrm{reu}}=\sum\limits_{\ell\geq1}s_{\ell+1}r_\ell$, we obtain that
    \begin{equation*}
        \E[D_{\init}^{(N)}(t)] = N\lambda_0 a^{\mathrm{reu}}\int_0^t\overline F(u)\,du.
    \end{equation*}
\end{proof}

Similarly, the normalized new blocks contribution $\frac{D_{\newb}^{(N)}(t)}{N}$ admits the following expectation, where we recall from \cref{equL The history session size R_j} that $B: = R_M$ denotes the total lifetime KV block numbers of a conversation.

\begin{lemma}[Mean Normalized New Blocks Contribution]
\label{lem: Mean New Blocks Contribution}
    For every $t\geq 0$, we let $D_{\newb}^{(N)}(t)$ denote the new blocks generated by new prompts in the time duration $(0,t]$. Then the mean of this normalized contribution is
    \begin{equation}
        \label{equ: Mean New Blocks Contribution}
        \frac{1}{N}\E[D_{\newb}^{(N)}(t)] = \lambda_0 b t.
    \end{equation}
\end{lemma}

\begin{proof}
    A conversation generates its turn-$j$ increment only if it reaches turn $j$, an event with probability $s_j$. 
    By the mapping theorem for marked Poisson processes, turn-$j$ accesses therefore form a Poisson process with rate $N\lambda_0s_j$. 
    Hence, the aggregate new-block growth rate is
    \begin{equation}
    \label{equ: new block growth rate}
        N\lambda_0\sum\limits_{j\geq1}s_j\E[Z_j].
    \end{equation}

    Because $B=\sum\limits_{j=1}^{M}Z_j$ and the increment sequence is independent of $M$,
    \begin{equation*}
        b
        =\sum\limits_{j\geq1}\E[Z_j\one\{M\geq j\}]
        =\sum\limits_{j\geq1}s_j\E[Z_j].
    \end{equation*}
    Thus, the rate in \cref{equ: new block growth rate} equals $N\lambda_0 b$. 
    Integrating over a time interval of length $t$, we obtain \cref{equ: Mean New Blocks Contribution}.
\end{proof}

Now we combine \cref{lem: Initial-prefix contribution,lem: Mean New Blocks Contribution} to obtain the following mean displacement curve.

\begin{lemma}[Mean Displacement Curve]
\label{lem: Mean Displacement Curve}
    For every $t\geq0$, the mean normalized background displacement is
    \begin{equation}
    \label{equ: Mean Displacement Curve}
        \frac{1}{N}\E[D^{(N)}(t)]=d(t):= \lambda_0 b t + \lambda_0 a^{\mathrm{reu}}\int_0^t\overline F(s)\,ds\text{.}
    \end{equation}
    Additionally, $d$ is strictly increasing and Lipschitz-continuous, with
    \begin{equation*}
        \lambda_0 b(t-s) \leq d(t)-d(s) \leq \lambda_0(b+a^{\mathrm{reu}})(t-s), \qquad\text{for all } 0\leq s\leq t.
    \end{equation*}
\end{lemma}

\begin{proof}
    We obtain \cref{equ: Mean Displacement Curve} by directly combining the mean normalized initial-prefix contribution and the normalized new blocks contribution derived from \cref{lem: Initial-prefix contribution,lem: Mean New Blocks Contribution}, respectively. 

    Because $0\leq\overline F(u)\leq1$, and for every $0\leq s\leq t$,
    \begin{align*}
        d(t)-d(s) =\lambda_0 b(t-s)+ \lambda_0 a^{\mathrm{reu}}\int_s^t\overline F(u)\,du,
    \end{align*}
    we have that
    \begin{equation*}
        \lambda_0 b(t-s) \leq d(t)-d(s) \leq \lambda_0(b+a^{\mathrm{reu}})(t-s).
    \end{equation*}
    Thus $d$ is strictly increasing and Lipschitz-continuous.   
\end{proof}

We are now ready to prove \cref{thm: Displacement Process Convergence}.

\subsubsection{Proof of \cref{thm: Displacement Process Convergence}.}
\label{sssec: normalized displacement sample path convergence}

\begin{proof}
    We prove convergence in $\mathbb D$ by establishing finite-$N$ concentration bounds. Fix $T<\infty$. Specifically, let
    \begin{equation*}
        V_T:=\lambda_0T\E[B^2]+\lambda_0\mu_W\E[(M-1)B^2],\qquad
        L:=\sup_{0\leq t\leq T}d'(t)=\lambda_0b+\lambda_0a^{\mathrm{reu}},
    \end{equation*}
    and define $C_T:=\frac{16V_T}{9}(2+4LT)$. We first show that, for every $t\in[0,T]$ and every $\varepsilon>0$,
    \begin{equation}
    \label{equ: displacement pointwise concentration}
        \pro\left(\left|\frac{D^{(N)}(t)}{N}-d(t)\right|>\varepsilon\right)
        \leq\frac{V_T}{N\varepsilon^2},
    \end{equation}
    and that, for every $0<\varepsilon\leq1$,
    \begin{equation}
    \label{equ: displacement sample path concentration}
        \pro\left(\sup_{0\leq t\leq T}\left|\frac{D^{(N)}(t)}{N}-d(t)\right|>\varepsilon\right)
        \leq\frac{C_T}{N\varepsilon^3}.
    \end{equation}

    For every $t\in[0,T]$, \cref{lem: single workflow PRM variance bound} gives $\operatorname{Var}(D^{(N)}(t)/N)\leq V_T/N$. In addition, due to \cref{lem: Mean Displacement Curve}, $\frac{1}{N}\E[D^{(N)}(t)] = d(t)$, we have that
    \begin{equation*}
        \E\left[\left|\frac{D^{(N)}(t)}{N}-d(t)\right|^2\right]\leq\frac{V_T}{N}.
    \end{equation*}
    Therefore, applying Chebyshev's inequality, we first derive the following pointwise second-moment bound for any $\varepsilon>0$:
    \begin{equation*}
        \pro\left(\left|\frac{D^{(N)}(t)}{N}-d(t)\right|>\varepsilon\right)\leq\frac{1}{\varepsilon^2}\E\left[\left|\frac{D^{(N)}(t)}{N}-d(t)\right|^2\right]\leq\frac{V_T}{N\varepsilon^2}.
    \end{equation*}
    This proves \cref{equ: displacement pointwise concentration}.

    We next demonstrate the upper bound for the sample path $\mathbb L^\infty$-difference. Recall from \cref{lem: Mean Displacement Curve} that for any $t>0$,
    \begin{equation*}
        d'(t)=\lambda_0 b+\lambda_0 a^{\mathrm{reu}} \overline F(t) \leq L := d'(0)=\lambda_0 b+\lambda_0 a^{\mathrm{reu}}.
    \end{equation*}
    Hence, the mean displacement curve $d$ is $L$-Lipschitz.

    Fix $0<\varepsilon\leq1$. Choose an equal partition
% Formatting only: fit this long display within the text width.
\begingroup\small
    \begin{equation*}
        \Pi = \{t_0,\cdots,t_K\},\quad \text{where}\quad 0=t_0<t_1<\cdots<t_K=T, \text{ and } t_{k+1}-t_k = t_{\ell+1}-t_\ell \text{ for all } 0\leq k < \ell \leq K-1,
    \end{equation*}
\endgroup
    satisfying that $t_{1}-t_0 \leq\frac{\varepsilon}{4 L}$. It is enough to choose $K$ such that $K+1\leq 2+\frac{4 L T}{\varepsilon}$, i.e., $K=  \lfloor1+\frac{4 L T}{\varepsilon}\rfloor$.

    Fix $t\in[t_k,t_{k+1}]$. Because $D^{(N)}(t)$ and  $d(t)$ are both nondecreasing,
    \begin{equation*}
        D^{(N)}(t_k)\leq D^{(N)}(t)\leq D^{(N)}(t_{k+1}) \quad \text{and} \quad d(t_k)\leq d(t)\leq d(t_{k+1}).
    \end{equation*}
    Therefore, 
    \begin{align*}
        \frac{D^{(N)}(t)}{N}-d(t)
        &\leq\frac{D^{(N)}(t_{k+1})}{N}-d(t_k) \\
        &=\left(\frac{D^{(N)}(t_{k+1})}{N}-d(t_{k+1})\right)+\left(d(t_{k+1})-d(t_k)\right) \\
        &\leq\left|\frac{D^{(N)}(t_{k+1})}{N}-d(t_{k+1})\right|+\frac{\varepsilon}{4}.
    \end{align*}
    and similarly,
    \[
    d(t)-\frac{D^{(N)}(t)}{N} \leq \left|\frac{D^{(N)}(t_{k})}{N}-d(t_{k})\right|+\frac{\varepsilon}{4}.
    \]

    To establish 
    \begin{equation*}
        \sup_{0\leq t\leq T}\left|\frac{D^{(N)}(t)}{N}-d(t)\right|\leq\varepsilon,
    \end{equation*}
    it is enough to let, for any $1\leq k\leq K$ 
    \begin{equation*}
        \max_{0\leq k\leq K}\left|\frac{D^{(N)}(t_k)}{N}-d(t_k)\right|\leq\frac{3\varepsilon}{4}.
    \end{equation*}
    Therefore, we can bound the $\mathbb L^\infty$-distance via a union bound, such that
    \begin{align*}
        \pro\left(\sup_{0\leq t\leq T}\left|\frac{D^{(N)}(t)}{N}-d(t)\right|>\varepsilon\right)& \leq\sum\limits_{k=0}^{K}\pro\left(\left|\frac{D^{(N)}(t_k)}{N}-d(t_k)\right|>\frac{3\varepsilon}{4}\right)\\
        & \leq (K+1) \frac{16V_T}{9N\varepsilon^2}\\
        & \leq (2+\frac{4 L T}{\varepsilon})\frac{16V_T}{9N\varepsilon^2}.
    \end{align*}

    When $0<\varepsilon\leq1$, the last term can be further bounded by  
    \begin{equation*}
        (2+\frac{4 L T}{\varepsilon})\frac{16V_T}{9N\varepsilon^2} \leq \frac{2+4 L T}{\varepsilon}\frac{16V_T}{9N\varepsilon^2} = \frac{C_T}{N\varepsilon^3}.
    \end{equation*}
    This proves \cref{equ: displacement sample path concentration}.

    For every fixed $\varepsilon>0$, the right-hand side of \cref{equ: displacement sample path concentration} converges to zero as $N\to\infty$. Since $T$ is arbitrary, this proves $D^{(N)}(\cdot)/N\overset{\mathbb P}{\to}d\text{ in }\mathbb D$, as stated in \cref{equ: displacement sample path convergence}.
%
%
%    As additional consequences of the same bounds, we record the corresponding convergence rates. For the pointwise convergence rate, set $\varepsilon = \delta/\sqrt N$ for some fixed $\delta>0$ in \cref{equ: displacement pointwise concentration}. Then
%    \begin{equation*}
%        \pro\left(\sqrt N\left|\frac{D^{(N)}(t)}{N}-d(t)\right|>\delta \right)\leq\frac{V_T}{\delta^2}.
%    \end{equation*}
%    Taking $\limsup_{N\to\infty}$ and then $\delta\to\infty$, we have that
%    \begin{equation*}
%        \left|\frac{D^{(N)}(t)}{N}-d(t)\right|=O_{\mathbb P}(N^{-1/2}).
%    \end{equation*}
%
%    For the sample path convergence rate, set $\varepsilon=\delta N^{-1/3}$ in \cref{equ: displacement sample path concentration}. For every fixed $\delta>0$ and all sufficiently large $N$, we have that $\varepsilon\leq 1$, and
%    \begin{equation*}
%        \pro\left(N^{1/3}\sup_{0\leq t\leq T}\left|\frac{D^{(N)}(t)}{N}-d(t)\right|>\delta\right)\leq\frac{C_T}{\delta^3}.
%    \end{equation*}
%    Taking $\limsup_{N\to\infty}$, we conclude that
%    \begin{equation*}
%        \sup_{0\leq t\leq T}\left|\frac{D^{(N)}(t)}{N}-d(t)\right|=O_{\mathbb P}(N^{-1/3}).\qedhere
%    \end{equation*}
\end{proof}

\subsection{Proof of \cref{thm: Characteristic Time Convergence- Full}}
\label{ssec: Proof of the Full Version of CT convergence}

% We also provide a proof sketch here.

% {\color{blue}\noindent\textit{Proof Sketch:} Different tagged turn-$j$ sessions have different eviction thresholds, because eviction occurs when $D^{(N)}(t)>N-R_j$. However, for any fixed $j$ with $s_j>0$, the session size satisfies $R_j=O_{\mathbb P}(1)$, and hence $R_j/N$ is asymptotically negligible. To control the event $\{|\tau_j^{(N)}-t^{\mathrm{MF}}|>\varepsilon\}$, it is sufficient to control the normalized displacement process at the two deterministic times $t^{\mathrm{MF}}-\varepsilon$ and $t^{\mathrm{MF}}+\varepsilon$, and the deviation of $R_j/N$ from $0$. Because $d(t^{\mathrm{MF}}-\varepsilon)<1<d(t^{\mathrm{MF}}+\varepsilon)$, the pointwise concentration bounds for $|D^{(N)}(t)/N-d(t)|$ imply that, with high probability, the tagged session is still resident at time $t^{\mathrm{MF}}-\varepsilon$ and already evicted by time $t^{\mathrm{MF}}+\varepsilon$. By applying a union bound to these deviations, we obtain the convergence of $\tau_j^{(N)}$ to $t^{\mathrm{MF}}$. The $O_{\mathbb P}(N^{-1/2})$ convergence rate follows from the pointwise concentration rate of the displacement process in \cref{thm: Displacement Process Convergence}. \hfill $\square$}

We next prove the convergence of the eviction time of a tagged turn-$j$ conversation session. First, we bound the tail of the turn-$j$ conversation session length $R_j$ for each $j\geq 1$, via the following lemma.

\begin{lemma}[Tail Bound of Turn-$j$ Session Length]
\label{lem: Tail Bound of Turn-j session}
    Fix a turn state $j$ such that $s_j>0$. For any fixed $\theta\in(0,\theta_0)$, we have the following tail bound
    \begin{equation*}
        \pro(R_j>x)\leq \frac{\E[e^{\theta B}]}{s_j}e^{-\theta x}, \qquad x\geq0.
    \end{equation*}
    Consequently, $R_j=O_{\mathbb P}(1).$
\end{lemma}

\begin{proof}

    The block increment sequence is independent of $M$. Because $R_j$ depends only on the block increments up to turn $j$, $R_j$ is independent of $M$. Hence,
    \begin{equation*}
        \E[e^{\theta R_j}]=\E[e^{\theta R_j}\mid M\geq j].
    \end{equation*}
    Conditional on the event $\{M\geq j\}$, the turn-$j$ history session is contained in the total lifetime block mass, and therefore $R_j\leq B$. Hence,
    \begin{align*}
        \E[e^{\theta R_j}]
        =\frac{1}{s_j}\E\left[e^{\theta R_j}\one\{M\geq j\}\right]
        \leq\frac{1}{s_j}\E\left[e^{\theta B}\one\{M\geq j\}\right] \leq\frac{\E[e^{\theta B}]}{s_j}.
    \end{align*}
    Applying Markov's inequality, we obtain the following bound for every $x>0$:
    \begin{equation*}
        \pro(R_j>x)\leq e^{-\theta x}\E[e^{\theta R_j}]\leq\frac{\E[e^{\theta B}]}{s_j}e^{-\theta x}.
    \end{equation*}
     Because the distribution of $R_j$ is independent of $N$, the family $\{R_j\}_{N\geq1}$ is tight, and hence $R_j=O_{\mathbb P}(1)$.
\end{proof}

%We are now ready to prove \cref{thm: Characteristic Time Convergence- Full}.

\begin{proof}[Proof of \cref{thm: Characteristic Time Convergence- Full}]

    We prove the stated TTL convergence by establishing a stronger finite-$N$ deviation bound. Fix a finite horizon $T>t^{\mathrm{MF}}$ and $\theta\in(0,\theta_0)$, and define
    \begin{equation*}
        c_d:=\lambda_0b,\qquad
        V_T:=\lambda_0T\E[B^2]+\lambda_0\mu_W\E[(M-1)B^2].
    \end{equation*}
    Specifically, for every $0<\varepsilon<\min\{t^{\mathrm{MF}},T-t^{\mathrm{MF}}\}$, we establish
    \begin{equation}
        \label{equ: TTL finite N bound}
        \pro\left(\left|\tau_j^{(N)}-t^{\mathrm{MF}}\right|>\varepsilon\right)
        \leq\frac{18V_T}{Nc_d^2\varepsilon^2}
        +\frac{\E[e^{\theta B}]}{s_j}\exp\left(-\frac{\theta c_d}{3}N\varepsilon\right).
    \end{equation}

    First, because for any $t>0$, $d'(t) = c_d + \lambda_0 a^{\mathrm{reu}} \overline F(t) \geq c_d$, thus for every $0<\varepsilon<t^{\mathrm{MF}}$,
    \begin{equation}
    \label{equ: lower separation around TC}
        d(t^{\mathrm{MF}}-\varepsilon)\leq1-c_d\varepsilon,
    \end{equation}
    and for every $\varepsilon>0$,
    \begin{equation}
    \label{equ: upper separation around TC}
        d(t^{\mathrm{MF}}+\varepsilon)\geq1+c_d\varepsilon.
    \end{equation}

    Fix $0<\varepsilon<\min\{t^{\mathrm{MF}},T-t^{\mathrm{MF}}\}$ and define
    \begin{equation*}
        t_-:=t^{\mathrm{MF}}-\varepsilon, \qquad t_+:=t^{\mathrm{MF}}+\varepsilon, \qquad \eta:=\frac{c_d\varepsilon}{3}.
    \end{equation*}
    Then we rewrite \cref{equ: lower separation around TC,equ: upper separation around TC} as 
    \begin{equation*}
    d(t_-)\leq1-3\eta, \quad \text{and }\quad d(t_+)\geq1+3\eta, \quad \text{respectively.}
    \end{equation*}

    We next introduce the event
    \begin{equation*}
        \mathcal E_N(\varepsilon) := \left\{\left|\frac{D^{(N)}(t_-)}{N}-d(t_-)\right |\leq\eta\right\}\cap\left\{\left|\frac{D^{(N)}(t_+)}{N}-d(t_+)\right|\leq\eta\right\}\cap\left\{\frac{R_j}{N}\leq\eta\right\},
    \end{equation*}
    and demonstrate that on $\mathcal E_N(\varepsilon)$, we have that 
    \begin{equation*}
        t^{\mathrm{MF}}-\varepsilon< \tau_j^{(N)}\leq t^{\mathrm{MF}}+\varepsilon.
    \end{equation*}

    On the one hand, when event $\mathcal E_N(\varepsilon)$ happens,
    \begin{equation*}
        \frac{D^{(N)}(t_-)}{N}\leq d(t_-)+\eta\leq1-2\eta, \quad \text{and} \quad  \frac{N-R_j}{N}=1-\frac{R_j}{N}\geq1-\eta,
    \end{equation*}
    which implies that, at the time $t_-$, the displacement is below the threshold, i.e., $D^{(N)}(t_-)<N-R_j$, and thus $\tau_j^{(N)}>t^{\mathrm{MF}}-\varepsilon.$

    On the other hand, at time $t_+$, given the event $\mathcal E_N(\varepsilon)$ happening, 
    \begin{equation*}
        \frac{D^{(N)}(t_+)}{N}\geq d(t_+)-\eta\geq1+2\eta>1.
    \end{equation*}
    Because $R_j > 0$, $N-R_j< N<D^{(N)}(t_+)$, and hence $\tau_j^{(N)}\leq t^{\mathrm{MF}}+\varepsilon.$ Thus we conclude the inclusion $\mathcal E_N(\varepsilon)\subseteq\left\{\left|\tau_j^{(N)}-t^{\mathrm{MF}}\right|\leq\varepsilon\right\}$.

    We now apply the union bound, which gives
    \begin{equation}
        \label{equ: union bound for tau and TC diff}
        \pro\left(\left|\tau_j^{(N)}-t^{\mathrm{MF}}\right|>\varepsilon\right)
        \leq\pro\left(\left|\frac{D^{(N)}(t_-)}{N}-d(t_-)\right|>\eta\right)+\pro\left(\left|\frac{D^{(N)}(t_+)}{N}-d(t_+)\right|>\eta\right)+\pro(R_j>N\eta).
    \end{equation}

    We use the pointwise displacement tail bound in \cref{equ: displacement pointwise concentration} to bound the first two terms in the R.H.S. of \cref{equ: union bound for tau and TC diff}, such that
    \begin{equation*}
        \pro\left(\left|\frac{D^{(N)}(t_\pm)}{N}-d(t_\pm)\right|>\eta\right)\leq\frac{V_T}{N\eta^2} \stackrel{\eta=c_d\varepsilon/3}{=}\frac{9V_T}{N c_d^2\varepsilon^2}.
    \end{equation*}
    In addition, the last term on the R.H.S. of \cref{equ: union bound for tau and TC diff} can be bounded using \cref{lem: Tail Bound of Turn-j session}:
    \begin{equation*}
        \pro(R_j>N\eta)\leq\frac{\E[e^{\theta B}]}{s_j}e^{-\theta N\eta}=\frac{\E[e^{\theta B}]}{s_j}\exp\left(-\frac{\theta c_d}{3}N\varepsilon\right).
    \end{equation*}
    Combining these bounds, we prove \cref{equ: TTL finite N bound}. Moreover, for every fixed $\varepsilon>0$, all terms on the right-hand side of \cref{equ: union bound for tau and TC diff} converge to zero as $N\rightarrow\infty$. Hence, $\tau_j^{(N)}\xrightarrow{\mathbb P}t^{\mathrm{MF}}.$

    As an additional consequence, the same deviation bound yields a convergence rate. When we let $x>0$ be fixed and set $\varepsilon=\frac{x}{\sqrt N}$, $\varepsilon<\min\{t^{\mathrm{MF}},T-t^{\mathrm{MF}}\}$ can always be satisfied for sufficiently large $N$, and \cref{equ: TTL finite N bound} gives
    \begin{equation*}
        \pro\left(\sqrt N\left|\tau_j^{(N)}-t^{\mathrm{MF}}\right|>x\right)\leq\frac{18V_T}{c_d^2x^2}+\frac{\E[e^{\theta B}]}{s_j}\exp\left(-\frac{\theta c_d}{3}x\sqrt N\right).
    \end{equation*}
    Therefore, we conclude that $\tau_j^{(N)}-t^{\mathrm{MF}}=O_{\mathbb P}(N^{-1/2})$ and the convergence in probability.
\end{proof}

\subsection{Proof of \cref{thm: Hit Ratio Convergence Full}}
\label{ssec: Proof of the Full Version of hit ratio conv MCM-LRU short}

% \noindent\textit{Proof Sketch:} The proof mainly relies on the concentration of the TTL established in \cref{thm: Characteristic Time Convergence- Full} and the Lipschitz continuity of the interarrival CDF $F$. For each fixed turn state $j$, \cref{thm: Characteristic Time Convergence- Full} suggests that its hit probability is asymptotically characterized by $F(t^{\mathrm{MF}})$. However, the TTL convergence is not uniform over all turn states. We therefore truncate the reusable-workload series at some level $\kappa$, and compare $\sum\limits_{j=1}^{\kappa}s_{j+1}r_jF(t^{\mathrm{MF}})$ with $\E\big[\sum\limits_{j=1}^{\min\{M-1,\kappa\}}R_jH_{N,j+1}\big].$ For every fixed $\kappa$, the difference between these two quantities vanishes as $N\rightarrow\infty$ by the TTL convergence and the Lipschitz continuity of $F$. The remaining contribution from turn states $j>\kappa$ is controlled by the tail reusable workload $\sum\limits_{j>\kappa}s_{j+1}r_j$, which vanishes as $\kappa\rightarrow\infty$. Finally, we conclude \cref{thm: Hit Ratio Convergence Full} by first taking $N\rightarrow\infty$ and then $\kappa\rightarrow\infty$. \hfill $\square$

We finally establish the mean-field convergence of the hit ratio under the whole-session LRU policy.
For every $j\geq2$, define $H_{N,j}$ as the indicator that the history session of size $R_{j-1}$ remains resident in memory when the turn-$j$ prompt arrives. Under whole-session LRU, the entire history session is either reused or missed; hence,
\begin{equation}
    \label{equ: finite N whole session hit ratio}
    h^{\LRU,(N)}
    =\frac{\E\!\Big[\sum\limits_{j=1}^{M-1}R_jH_{N,j+1}\Big]}
    {a^{\mathrm{cum}}}.
\end{equation}

To analyze the limiting behavior of the hit ratio, we first define some auxiliary notations. 

Recall that $a^{\mathrm{reu}}= \E [\sum\limits_{j=1}^{M-1} R_j]= \sum\limits_{j\geq1}s_{j+1}r_j$ denotes the expected amount of reusable history session workload. We further define $a^{\mathrm{reu},(\kappa)}$ as a truncation of $a^{\mathrm{reu}}$ at level $\kappa\geq 1$, i.e.,
\begin{equation*}
    a^{\mathrm{reu},(\kappa)}:= \sum\limits_{j=1}^{\kappa}s_{j+1}r_j.
\end{equation*}
The tail workload is then defined such that
\begin{equation*}
    \mathcal T_\kappa:=\sum\limits_{j>\kappa}s_{j+1}r_j=a^{\mathrm{reu}}-a^{\mathrm{reu},(\kappa)}.
\end{equation*}
In addition, we let 
\begin{equation*}
    C_\kappa:=\sum\limits_{j=1}^{\kappa}s_{j+1}\sqrt{\E[R_j^2]}
\end{equation*}
which we will use to control the finite-$\kappa$ approximation error.

For a finite horizon $T>t^{\mathrm{MF}}$ and a fixed $\theta\in(0,\theta_0)$, define
\begin{equation}
\label{equ: finite N TTL probability shorthand}
    p_{N,j}(\varepsilon)
    :=\frac{18V_T}{Nc_d^2\varepsilon^2}
    +\frac{\E[e^{\theta B}]}{s_j}\exp\left(-\frac{\theta c_d}{3}N\varepsilon\right),
\end{equation}
where $c_d:=\lambda_0b$ and
$V_T:=\lambda_0T\E[B^2]+\lambda_0\mu_W\E[(M-1)B^2]$.

We prove \cref{thm: Hit Ratio Convergence Full} through the following lemmas. We first establish an exact representation of the finite-$N$ LRU hit ratio in terms of the TTL. Recall that $W_j$, the prompt IAT between turns $j$ and $j+1$, has CDF $F$. By the definition of the TTL, we have that
\begin{equation*}
    H_{N,j+1}=\one\left\{W_j<\tau_j^{(N)}\right\}.
\end{equation*}

\begin{lemma}[Exact Representation of the LRU Hit Ratio]
\label{lem: Exact Representation of LRU Hit Probability}
For every $N$, the hit ratio under the MCM-LRU model admits the representation
\begin{equation}
\label{equ: exact LRU hit ratio characteristic representation}
    h^{\LRU,(N)} = \frac{1}{a^{\mathrm{cum}}}\sum\limits_{j\geq1}s_{j+1}\E\left[R_j F(\tau_j^{(N)}) \right].
\end{equation}
\end{lemma}

\begin{proof}
    The tagged-conversation IAT $W_j$ is independent of the block increment sequence and hence of $R_j$. Moreover, the displacement process is determined entirely by the background conversations, which are independent of the tagged conversation. Therefore, $W_j$ is independent of $(R_j,\tau_j^{(N)})$. Conditioning on $(R_j,\tau_j^{(N)})$ and using the continuity of $F$ from \cref{eq: assumption Lipschitz} gives
    \begin{align*}
        \E\left[H_{N,j+1}\mid R_j,\tau_j^{(N)}\right]
        =\pro\left(W_j<\tau_j^{(N)}\mid R_j,\tau_j^{(N)}\right) = F \big (\tau_j^{(N)}\big) %\\&=1-e^{-r\tau_j^{(N)}}.
    \end{align*}

    Using the monotone convergence theorem, we have that
    \begin{equation*}
        \E\Big[\sum\limits_{j=1}^{M-1}R_jH_{N,j+1}\Big]=\E\Big[\sum\limits_{j\geq1}R_jH_{N,j+1}\one\{M\geq j+1\}\Big] = \sum\limits_{j\geq1}\E\left[R_jH_{N,j+1}\one\{M\geq j+1\}\right].
    \end{equation*}
    Moreover, $M$ is independent of the block increment sequence and the IAT sequence, and \cref{lem: Palm invariance of background displacement} gives independence from the background process. Hence, $M$ is independent of $(R_j,H_{N,j+1})$, so
    \begin{equation*}
        \E\Big[R_jH_{N,j+1}\one\{M\geq j+1\}\Big] =s_{j+1}\E\left[R_jH_{N,j+1}\right] =s_{j+1}\E\left[R_jF(\tau_j^{(N)})\right].
    \end{equation*}
    By summing over $j\geq 1$ and dividing by $a^{\mathrm{cum}}$, we prove \cref{equ: exact LRU hit ratio characteristic representation}.
\end{proof}

Recall from \cref{thm: Characteristic Time Convergence- Full} that the convergence of $\tau_j^{(N)}$ is not uniform in $j$. We therefore first use the finite-turn proxy
\begin{equation*}
    \sum\limits_{j=1}^{\kappa}s_{j+1}\E\left[R_jF(\tau_j^{(N)})\right].
\end{equation*}
The difference between the finite-size and mean-field hit ratios can then be decomposed into two parts: the tail contributions beyond turn $\kappa$, which are controlled by $\mathcal T_\kappa$, and the truncated approximation error $| \sum\limits_{j=1}^{\kappa}s_{j+1}\E\left[R_jF(\tau_j^{(N)})\right]- a^{\mathrm{reu},(\kappa)}F(t^{\mathrm{MF}}) |$, which can be bounded via the following lemma.

\begin{lemma}[Finite-Turn Hit Ratio Error]
    \label{lem: Finite Turn Hit Ratio Error}
        Fix a finite $\kappa$ such that $s_{\kappa+1}>0$. For every $0<\varepsilon<\min\{t^{\mathrm{MF}},T-t^{\mathrm{MF}}\}$,
    \begin{equation}
    \label{equ: finite turn hit ratio error}
        \left| \sum\limits_{j=1}^{\kappa}s_{j+1}\E\left[R_jF(\tau_j^{(N)})\right]- a^{\mathrm{reu},(\kappa)}F(t^{\mathrm{MF}})\right|\leq L_F \varepsilon a^{\mathrm{reu},(\kappa)}+C_\kappa\sqrt{p_{N,\kappa}(\varepsilon)}.
    \end{equation}
\end{lemma}

\begin{proof}
    By the triangle inequality, the L.H.S. of \cref{equ: finite turn hit ratio error} can be bounded such that
    \begin{align*}
        \left|\sum\limits_{j=1}^{\kappa}s_{j+1}\E\left[R_jF\left(\tau_j^{(N)}\right)\right]-\sum\limits_{j=1}^{\kappa}s_{j+1}r_jF(t^{\mathrm{MF}})\right| \leq\sum\limits_{j=1}^{\kappa}s_{j+1}\E\left[R_j\left|F\left(\tau_j^{(N)}\right)-F(t^{\mathrm{MF}})\right|\right],
    \end{align*}

    By \cref{eq: assumption Lipschitz}, $F(\cdot)$ is $L_F$-Lipschitz, i.e.,
    \begin{equation*}
        |F(t)-F(s)|\leq L_F|t-s|, \qquad \text{for all } 0\leq t<s<\infty.
    \end{equation*}
    For every $j\leq \kappa$, we define the concentration event $\mathcal G_{N,j}(\varepsilon):=\left\{\left|\tau_j^{(N)}-t^{\mathrm{MF}}\right|\leq\varepsilon\right\}$. On $\mathcal G_{N,j}(\varepsilon)$, the Lipschitz condition implies that
    \begin{equation*}
        \left|F\left(\tau_j^{(N)}\right)-F(t^{\mathrm{MF}})\right|\leq L_F\varepsilon.
    \end{equation*}
    Because $0\leq F(t)\leq1$ for all $t\geq0$, on the complementary event $\mathcal G_{N,j}(\varepsilon)^c$ we still have \[\left|F\left(\tau_j^{(N)}\right)-F(t^{\mathrm{MF}})\right|\leq 1.\]
    Therefore, $\E\left[R_j\left|F\left(\tau_j^{(N)}\right)-F(t^{\mathrm{MF}})\right|\right]
        \leq L_F\varepsilon\E[R_j]+\E\left[R_j\one\left\{\mathcal G_{N,j}(\varepsilon)^c\right\}\right].$
    By the Cauchy-Schwarz inequality,
    \begin{equation*}
        \E\left[R_j\one\left\{\mathcal G_{N,j}(\varepsilon)^c\right\}\right]
        \leq\sqrt{\E[R_j^2]}\sqrt{\pro\left(\left|\tau_j^{(N)}-t^{\mathrm{MF}}\right|>\varepsilon\right)}.
    \end{equation*}

    According to the finite-$N$ TTL deviation bound in \cref{equ: TTL finite N bound},
    \begin{equation*}
        \pro\left(\left|\tau_j^{(N)}-t^{\mathrm{MF}}\right|>\varepsilon\right)\leq p_{N,j}(\varepsilon).
    \end{equation*}
    Moreover, because $s_j$ is nonincreasing in $j$, we have $s_j\geq s_\kappa$ for every $j\leq \kappa$. Hence, {\color{blue}\cref{equ: finite N TTL probability shorthand}} implies $p_{N,j}(\varepsilon)\leq p_{N,\kappa}(\varepsilon),$
    and therefore $\pro\left(\left|\tau_j^{(N)}-t^{\mathrm{MF}}\right|>\varepsilon\right)\leq p_{N,\kappa}(\varepsilon)$ holds for all $j\leq \kappa$.
    It follows that for each $1\leq j\leq \kappa$,
    \begin{equation*}
        \E\left[R_j\left|F\left(\tau_j^{(N)}\right)-F(t^{\mathrm{MF}})\right|\right]
        \leq L_F\varepsilon r_j+\sqrt{\E[R_j^2]}\sqrt{p_{N,\kappa}(\varepsilon)}.
    \end{equation*}
    Multiplying by $s_{j+1}$ and summing over $j=1,\ldots,\kappa$, we obtain that
    \begin{equation*}
        \sum\limits_{j=1}^{\kappa}s_{j+1}\E\left[R_j\left|F\left(\tau_j^{(N)}\right)-F(t^{\mathrm{MF}})\right|\right]
        \leq L_F\varepsilon a^{\mathrm{reu},(\kappa)}+C_\kappa\sqrt{p_{N,\kappa}(\varepsilon)},
    \end{equation*}
    and we prove \cref{equ: finite turn hit ratio error}.
\end{proof}

We next control the error of the tail workload $\mathcal T_\kappa = \sum\limits_{j>\kappa}s_{j+1}r_j$, via the lemma below.

\begin{lemma}[Bound for the Tail Workload]
\label{lem: Tail Bound for Reusable Workload}
    For every integer $\kappa\geq1$ and any fixed $\theta\in(0,\theta_0)$, the tail workload $\mathcal T_\kappa$ has the following upper bound:
    \begin{equation*}
        \mathcal T_\kappa\leq \E[B^2e^{\theta B}]e^{-\theta(\kappa+2)}.
    \end{equation*}
    Consequently, $\mathcal T_\kappa\rightarrow 0$ as $\kappa\to \infty$.

\end{lemma}

\begin{proof}
    Using the independence between turn number $M$ and the block increments $R_j$'s, we can rewrite $\mathcal T_\kappa$, such that
    \begin{align*}
        \mathcal T_\kappa = \sum\limits_{j>\kappa}s_{j+1}r_j =\sum\limits_{j>\kappa}\E\left[R_j\one\{M\geq j+1\}\right] =\E\Big[\sum\limits_{j=\kappa+1}^{M-1}R_j\Big].
    \end{align*}
    Conditional on the event $\{M\geq \kappa+2\}$, each partial history session satisfies $R_j\leq B$. Hence,
    \begin{equation*}
        \sum\limits_{j=\kappa+1}^{M-1}R_j \leq MB\one\{M\geq \kappa+2\}.
    \end{equation*}
    Because each turn generates at least one additional block, $M\leq B$ almost surely. Therefore,
    \begin{equation*}
        \sum\limits_{j=\kappa+1}^{M-1}R_j \leq B^2\one\{B\geq \kappa+2\}.
    \end{equation*}
    For any fixed $\theta\in(0,\theta_0)$, because $\one\{B\geq \kappa+2\} \leq e^{\theta (B-\kappa-2)}$,
    \begin{equation*}
        B^2\one\{B\geq \kappa+2\} \leq B^2e^{\theta B}e^{-\theta(\kappa+2)} \quad \Rightarrow \quad  \mathcal T_\kappa\leq\E[B^2e^{\theta B}]e^{-\theta(\kappa+2)},
    \end{equation*}
    and we conclude \cref{lem: Tail Bound for Reusable Workload}.
\end{proof}

We are now ready to prove \cref{thm: Hit Ratio Convergence Full}, by applying \cref{lem: Exact Representation of LRU Hit Probability,lem: Finite Turn Hit Ratio Error,lem: Tail Bound for Reusable Workload}.

%\hitconvergencetheorem*

\begin{proof}[Proof of \cref{thm: Hit Ratio Convergence Full}]

    We prove the stated hit ratio convergence by first establishing a stronger finite-$N$ error bound. Fix $T>t^{\mathrm{MF}}$, a finite truncation level $\kappa$ such that $s_{\kappa+1}>0$, and $0<\varepsilon<\min\{t^{\mathrm{MF}},T-t^{\mathrm{MF}}\}$. Specifically, we show that
    \begin{equation}
    \label{equ: finite N hit ratio error bound}
        \left|h^{\LRU,(N)}-h^{\LRU,\infty}\right|
        \leq \frac{L_F\varepsilon a^{\mathrm{reu},(\kappa)}}{a^{\mathrm{cum}}}
        +\frac{C_\kappa}{a^{\mathrm{cum}}}\sqrt{p_{N,\kappa}(\varepsilon)}
        +\frac{2\E[B^2e^{\theta B}]e^{-\theta(\kappa+2)}}{a^{\mathrm{cum}}}.
    \end{equation}
    
    Recall \cref{lem: Exact Representation of LRU Hit Probability} that the finite-size hit ratio has the following representation
    \begin{equation*}
    h^{\LRU,(N)} = \frac{1}{a^{\mathrm{cum}}} \sum\limits_{j\geq1} s_{j+1} \E\left[R_j F\left(\tau_j^{(N)}\right)\right],
    \end{equation*}
    and the limiting hit ratio,
    \begin{equation*}
    h^{\LRU,\infty} = \frac{1}{a^{\mathrm{cum}}} \sum\limits_{j\geq1} s_{j+1}r_j F(t^{\mathrm{MF}}).
    \end{equation*}
    We fix a truncation $\kappa$ and use the triangle inequality to bound the difference between $h^{\LRU,(N)}$ and $h^{\LRU,\infty}$, such that
    \begin{align*}
        \left|h^{\LRU,(N)}-h^{\LRU,\infty}\right|
        &\leq \frac{1}{a^{\mathrm{cum}}}\left|\sum\limits_{j=1}^{\kappa}s_{j+1}\E\left[R_j F\left(\tau_j^{(N)}\right)\right]-
        \sum\limits_{j=1}^{\kappa}s_{j+1}r_j F(t^{\mathrm{MF}})\right| \\
        &\quad+\frac{1}{a^{\mathrm{cum}}}\sum\limits_{j>\kappa}s_{j+1}\E\left[R_j F\left(\tau_j^{(N)}\right)\right]+\frac{1}{a^{\mathrm{cum}}} \sum\limits_{j>\kappa}s_{j+1}r_j F(t^{\mathrm{MF}}).
    \end{align*}
    Because $0\leq F(t)\leq1$,
    \begin{equation*}
        \sum\limits_{j>\kappa}s_{j+1}\E\left[R_j F\left(\tau_j^{(N)}\right)\right]\leq \sum\limits_{j>\kappa}s_{j+1}\E\left[R_j\right] \leq \mathcal T_\kappa,\quad\text{and} \quad \sum\limits_{j>\kappa}s_{j+1}r_j F(t^{\mathrm{MF}})\leq\mathcal T_\kappa.
    \end{equation*}
    Applying \cref{lem: Finite Turn Hit Ratio Error} to the finite-turn error, we obtain that
    \begin{equation*}
        \left|h^{\LRU,(N)}-h^{\LRU,\infty}\right| \leq
        \frac{L_F \varepsilon a^{\mathrm{reu},(\kappa)}}{a^{\mathrm{cum}}} + \frac{C_\kappa}{a^{\mathrm{cum}}}\sqrt{p_{N,\kappa}(\varepsilon)} + \frac{2\mathcal T_\kappa}{a^{\mathrm{cum}}}.
    \end{equation*}
    By \cref{lem: Tail Bound for Reusable Workload}, we have that $\mathcal T_\kappa\leq\E[B^2e^{\theta B}]e^{-\theta(\kappa+2)},$ which proves the upper bound in \cref{equ: finite N hit ratio error bound}.

    To prove the convergence result, we fix $\kappa$ and choose $\varepsilon=N^{-1/4}$. For sufficiently large $N$, $\varepsilon<\min\{t^{\mathrm{MF}},T-t^{\mathrm{MF}}\}$. The first term satisfies
    \begin{equation*}
        \frac{L_F \varepsilon a^{\mathrm{reu},(\kappa)}}{a^{\mathrm{cum}}}=O(N^{-1/4}).
    \end{equation*}
    Moreover,
    \begin{align*}
        p_{N,\kappa}(N^{-1/4}) &= \frac{18V_T}{c_d^2N^{1/2}} + \frac{\E[e^{\theta B}]}{s_\kappa} \exp\left(-\frac{\theta c_d}{3}N^{3/4}\right) = O(N^{-\frac{1}{2}}),
    \end{align*}
    and therefore $\sqrt{p_{N,\kappa}(N^{-1/4})}\rightarrow0.$

    Hence, for every fixed $\kappa$,
    \begin{equation*}
        \limsup_{N\rightarrow\infty} \left|h^{\LRU,(N)}-h^{\LRU,\infty}\right| \leq \frac{2\mathcal T_\kappa}{a^{\mathrm{cum}}}.
    \end{equation*}
    Finally, by \cref{lem: Tail Bound for Reusable Workload}, $\mathcal T_\kappa\rightarrow 0$ as $\kappa\rightarrow\infty$. Taking $\kappa\rightarrow\infty$, we conclude \cref{thm: Hit Ratio Convergence Full}.
\end{proof}

\section{Proofs of Results in \cref{sec: mixed workflow optimization}}
\label{app: section5 proofs}

In this section, we prove the results in \cref{sec: mixed workflow optimization}.
For each workflow $c$, the same Tonelli argument as in \crefrange{equ: cumulative requested workload identity}{equ: number of blocks eligible for reuse} expresses the expected cumulative requested workload and the expected reusable history workload, respectively, as
\begin{equation}
    a_c^{\mathrm{cum}}=\sum\limits_{\ell\geq1}s_{c,\ell}r_{c,\ell},
    \qquad
    a_c^{\mathrm{reu}}=\sum\limits_{\ell\geq1}s_{c,\ell+1}r_{c,\ell}.
    \label{equ: mixed workflow workload identities}
\end{equation}

We begin by proving \cref{thm: mixed workflow LRU mean field}, the mean-field result of the mixed-workflow MCM under LRU.

% \mixedlrutheorem*

\begin{proof}[Proof of \cref{thm: mixed workflow LRU mean field}]
    For a tagged conversation, decompose the displacement process into the sum of the $J$ workflow-specific contributions. 
    Then we can directly apply the sample-path argument in \cref{thm: Displacement Process Convergence} to each independent workflow stream to obtain $D_c^{(N)}(\cdot)/N\overset{\mathbb P}{\to}d_c\text{ in }\mathbb D$. 
    Moreover, because $J$ is fixed and each $d_c$ is continuous, we can sum these convergences and obtain $D^{(N)}(\cdot)/N\overset{\mathbb P}{\to}d\text{ in }\mathbb D$.
    The curve $d$ is continuous and strictly increasing at its unique level-one crossing. 
    Therefore, we use same argument as in \cref{thm: Characteristic Time Convergence- Full} and obtain the TTL convergence in \cref{item: mixed workflow TTL convergence} of \cref{thm: mixed workflow LRU mean field}.
    
    For a workflow-$c$ conversation, let $H_{c,\ell+1}^{(N)}$ denote the indicator that its turn-$\ell$ history session remains resident when turn $\ell+1$ arrives. The whole-session LRU therefore admits the following hit ratio representations
    \begin{align}
        {h_c^{\LRU,(N)}}
        &{=\frac{\E\!\left[\sum\limits_{\ell=1}^{M_c-1}R_{c,\ell}H_{c,\ell+1}^{(N)}\right]}
        {a_c^{\mathrm{cum}}},}
        \label{equ: mixed finite N workflow hit ratio}\\
        {h^{\LRU,(N)}}
        &{=\frac{\sum\limits_{c=1}^J\lambda_c
        \E\!\left[\sum\limits_{\ell=1}^{M_c-1}R_{c,\ell}H_{c,\ell+1}^{(N)}\right]}
        {\sum\limits_{c=1}^J\lambda_ca_c^{\mathrm{cum}}}.}
        \label{equ: mixed finite N system hit ratio}
    \end{align}
    For a workflow-$c$ conversation, we recall in \cref{lem: Palm invariance of background displacement} that the next-turn IAT has CDF $F_c$ and is independent of the background displacement process. 
    Hence, the probability that any fixed reusable history survives until its next access converges to $F_c(t^{\mathrm{MF}}_{\LRU})$. 
    We then apply the truncation and uniform-integrability argument of \cref{thm: Hit Ratio Convergence Full} and weight it by the reusable history size, and conclude \cref{equ: type specific LRU hit ratio mean field}. 
    Finally, we aggregate the hit and requested block rates across workflows and prove \cref{equ: server wide LRU hit ratio mean field}.
\end{proof}

Based on \cref{thm: mixed workflow LRU mean field}, we are ready to prove \cref{thm: implicit workflow allocation under LRU}, which shows that in our mean-field system LRU implicitly allocates the capacity for each workflow.

\begin{proof}[Proof of \cref{thm: implicit workflow allocation under LRU}]
    We first establish a workflow-specific sample-path convergence that will be used as an intermediate step. For each $c\in\mathcal C$, the workflow-$c$ conversation stream is a single-workflow MCM with arrival rate $N\lambda_c$, block statistics $(b_c,a_c^{\mathrm{reu}})$, and interarrival distribution $F_c$. 
    Hence, by \cref{thm: Displacement Process Convergence}, for every fixed $T<\infty$, we have the following sample path convergence
    \begin{equation}
        \label{equ: workflow displacement convergence LRU}
        {\sup_{0\leq t\leq T}
        \left|\frac{D_c^{(N)}(t)}{N}-d_c(t)\right|
        \xrightarrow{\mathbb P}0.}
    \end{equation}
    This proves $D_c^{(N)}(\cdot)/N\overset{\mathbb P}{\to}d_c\text{ in }\mathbb D$.

    We next choose any fixed $T>t^{\mathrm{MF}}_{\LRU}$ and evaluate \cref{equ: workflow displacement convergence LRU} at $t^{\mathrm{MF}}_{\LRU}$. Then
    \begin{equation*}
        \frac{D_c^{(N)}(t^{\mathrm{MF}}_{\LRU})}{N}\xrightarrow{\mathbb P}d_c(t^{\mathrm{MF}}_{\LRU})=\alpha_c^{\LRU}.
    \end{equation*}
    For any fixed $j$ with $s_{c,j}>0$, by \cref{thm: mixed workflow LRU mean field} we have that $\tau_{c,j}^{(N)}\xrightarrow{\mathbb P}t^{\mathrm{MF}}_{\LRU}$. 
    Fix $H>t^{\mathrm{MF}}_{\LRU}$. On $\{\tau_{c,j}^{(N)}\leq H\}$,
    \begin{align*}
    \Big|\frac{D_c^{(N)}(\tau_{c,j}^{(N)})}{N}-d_c(t^{\mathrm{MF}}_{\LRU})\Big|
    \leq\sup_{0\leq t\leq H}\Big|\frac{D_c^{(N)}(t)}{N}-d_c(t)\Big|+\left|d_c(\tau_{c,j}^{(N)})-d_c(t^{\mathrm{MF}}_{\LRU})\right|.
    \end{align*}
    The first term converges to zero in probability by \cref{equ: workflow displacement convergence LRU}, and the second does so by continuity of $d_c$ and TTL convergence. 
    Therefore, we prove \cref{equ: random-time implicit quota convergence}.
    Moreover, 
    \[
    \sum\limits_{c\in\mathcal C}\alpha_c^{\LRU}=\sum\limits_{c\in\mathcal C}d_c(t^{\mathrm{MF}}_{\LRU})=d(t^{\mathrm{MF}}_{\LRU})=1.
    \]
\end{proof}

We then prove \cref{thm: mixed workflow WAQ LRU mean field}, the mean-field limiting result of the mixed-workflow MCM under our WAQ-LRU.

%\waqmeanfieldtheorem*

\begin{proof}[Proof of \cref{thm: mixed workflow WAQ LRU mean field}]
    We note that, under WAQ-LRU, the workflow-$c$ list evolves as an isolated single-workflow LRU cache with capacity $\lfloor\alpha_cN\rfloor$ and arrival rate $N\lambda_c$.
    Normalize its displacement by the total reference capacity $N$. By \cref{thm: Displacement Process Convergence}, $D_c^{(N)}(\cdot)/N\overset{\mathbb P}{\to}d_c\text{ in }\mathbb D$.
    Moreover, the eviction occurs when the normalized displacement process $\frac{D^{(N)}_c(t)}{N}$ reaches $\alpha_c$. 
    Therefore, by \cref{thm: Characteristic Time Convergence- Full}, the finite-$N$ TTL $\tau_{c,j}^{(N)}(\alpha_c)$ converges to the unique limiting TTL $t_{c}^{\mathrm{MF}}(\alpha_c)=d_c^{-1}(\alpha_c)$, for $c\in \mathcal C$ and $j \geq 1$.
    We next apply the single-workflow hit ratio convergence shown in \cref{thm: Hit Ratio Convergence Full}, and conclude \cref{equ: type specific WAQ hit ratio}.
    By weighting the resulting hit-block rates by the workflow arrival rates, we prove \cref{equ: server wide WAQ hit ratio}.

    Finally, $\alpha_c^{\LRU}=d_c(t^{\mathrm{MF}}_{\LRU})$ implies
    $t_{c}^{\mathrm{MF}}(\alpha_c^{\LRU})=d_c^{-1}(\alpha_c^{\LRU})=t^{\mathrm{MF}}_{\LRU}$.
    Substituting these TTLs into \cref{equ: server wide WAQ hit ratio}, we conclude 
    $h^{\mathrm{WAQ},\infty}(\vec\alpha^{\LRU})=h^{\LRU,\infty}$.

\end{proof}

We now move to provide some properties of the optimal solution $\vec \alpha^*$ to the quota-allocation problem in \cref{equ: WAQ quota optimization}.
For notational convenience, define
\[
g_c(\alpha_c):=\lambda_ca_c^{\mathrm{reu}}F_c(t_{c}^{\mathrm{MF}}(\alpha_c)).
\]
If $F_c$ is differentiable with density $f_c$, implicit differentiation gives
\begin{equation*}
    \frac{dt_{c}^{\mathrm{MF}}(\alpha_c)}{d\alpha_c}=\frac{1}{\lambda_c(b_c+a_c^{\mathrm{reu}}\overline F_c(t_{c}^{\mathrm{MF}}(\alpha_c)))}.
\end{equation*}
As a result,
\begin{equation}
\label{equ: WAQ marginal value}
    g_c'(\alpha_c)=m_c(t_{c}^{\mathrm{MF}}(\alpha_c)),\qquad \text{where} \qquad
    m_c(t):=\frac{a_c^{\mathrm{reu}}f_c(t)}{b_c+a_c^{\mathrm{reu}}\overline F_c(t)}.
\end{equation}

\begin{proposition}[Marginal-Value Characterization of an Optimal Quota]
\label{prop: optimal WAQ quota marginal condition}
    Suppose that $m_c(t)$ is nonincreasing in $t$ for every workflow $c$. Then the objective in \cref{equ: WAQ quota optimization} is concave. 
    A quota vector $\vec\alpha^*$ is optimal if and only if there exists some $\nu\geq0$, such that
    \begin{equation*}
        m_c(t_{c}^{\mathrm{MF}}(\alpha_c^*))=\nu
    \end{equation*}
    for every $c$ with $\alpha_c^*>0$, where $m_c(0)\leq\nu$ when $\alpha_c^*=0$. Additionally, if each $m_c$ is strictly decreasing, the optimizer is unique.
    \end{proposition}
    \begin{proof}
        Note that the inverse $t_{c}^{\mathrm{MF}}(\alpha_c)=d_c^{-1}(\alpha_c)$ is increasing. Thus, \cref{equ: WAQ marginal value} implies that every $g_c$ is concave when $m_c$ is nonincreasing. 
        Moreover, we note that the stated equalities and inequalities are exactly the KKT conditions. 
        Finally, when $m_c(\cdot)$ is strictly decreasing, $g_c$ is strictly concave, which guarantees the uniqueness of the optimizer $\vec\alpha^*$.
\end{proof}

\begin{proposition}[Existence and Uniqueness of the Optimal Quota]
\label{prop: WAQ optimizer existence and uniqueness}
The mean-field quota-allocation problem in \cref{equ: WAQ quota optimization} admits at least one optimal solution. If, in addition, every workflow has a positive reusable workload $a_c^{\mathrm{reu}}>0$ and an exponential IAT distribution
\begin{equation*}
    F_c(t)=1-e^{-\mu_ct},\qquad \mu_c>0,
\end{equation*}
then its objective is strictly concave on $\Delta^J$, and the optimal quota vector $\vec\alpha^*$ is unique.
\end{proposition}
\begin{proof}
    First, the existence of the optimal solution follows from the continuity of the objective function, and the compactness of $\Delta^J$. 
    Next, for the IAT CDF $F_c(t)=1-e^{-\mu_ct}$,
    \[
    m_c(t)=\frac{a_c^{\mathrm{reu}}\mu_ce^{-\mu_ct}}{b_c+a_c^{\mathrm{reu}}e^{-\mu_ct}},\qquad \Rightarrow \qquad
    m_c'(t)=-\frac{a_c^{\mathrm{reu}}b_c\mu_c^2e^{-\mu_ct}}{(b_c+a_c^{\mathrm{reu}}e^{-\mu_ct})^2}<0.
    \]
    Because $m_c(\cdot)$ is strictly decreasing for all $c\in \mathcal C$, by \cref{prop: optimal WAQ quota marginal condition}, the optimal solution is unique.
\end{proof}
% The implicit LRU allocation can be optimal only if $m_1(t^{\mathrm{MF}}_{\LRU})=\cdots=m_J(t^{\mathrm{MF}}_{\LRU})$ for all active workflows, together with the KKT inequalities for inactive workflows.

\section{Closed-Form Mean-Field Result to the SCM under LRU}
\label{ssec: Mean-Field Results of the SCM under LRU}

Based on the mean-field convergence results in \cref{sec: main results}, we now focus our analysis on the SCM-LRU system. We first introduce the hit intensity coefficient $\beta$, which summarizes the relative time scale between prompt revisits and background KV-block accumulation.

\begin{definition}[Hit Intensity Coefficient]
    \label{def: hit strength coefficient}
    Given the unscaled conversation arrival rate $\lambda_0$, the prompt interarrival rate $\mu$, and the expected cumulative KV-block workload $b$ of a conversation, we define the hit intensity coefficient
    \begin{equation}
    \label{equ: hit strength coefficient beta}
        \beta:=\frac{ \mu }{\lambda_0 b}.
    \end{equation}
    In the SCM setting, where $b=z/q$, $\beta$ reduces to
    \begin{equation}
    \label{equ: hit strength coefficient beta SCM}
        \beta=\frac{ \mu q}{\lambda_0z}.
    \end{equation}
\end{definition}

We note that the hit intensity coefficient $\beta$ reflects the relative speed between the return of a tagged conversation and the accumulation speed of background KV-block workload. To be specific, the numerator $ \mu $ in \cref{equ: hit strength coefficient beta} represents the prompt interarrival rate of an individual conversation. In addition, we will show in \cref{sssec: Displacement Process Decomposition} that the denominator $\lambda_0 b$ represents the average rate at which background KV blocks are generated. A larger $ \mu $ implies that the next turn of the tagged conversation is likely to arrive sooner, increasing the chance that its history session is still resident in memory. In contrast, a larger background block-arrival rate $\lambda_0 b$ causes the tagged session to be displaced more rapidly from the MRU side toward the LRU side. Therefore, a larger $\beta$ corresponds to a more favorable regime for KV cache reuse.

We now present the closed-form limiting TTL and hit ratio result in the SCM-LRU model.
We use $W_0(\cdot)$ to denote the principal branch of the Lambert $W$ function, which is defined as the inverse of $x\mapsto xe^x$ restricted to $[-1,\infty)$.

\begin{corollary}[Closed-Form Limiting TTL and Hit Ratio under the SCM]
\label{cor: SCM Closed Form}
    Consider the SCM with continuation probability $p\in(0,1)$, termination probability $q:=1-p$, common increment mean $\E[Z_j]=z>0$ for every $j\geq1$, and hit intensity coefficient $\beta$ defined in \cref{equ: hit strength coefficient beta SCM}. The mean-field TTL admits the closed-form expression
    \begin{equation}
    \label{equ: SCM TC closed form}
        t^{\mathrm{MF,SCM}}=\frac{1}{ \mu }\left[\beta-\frac{p}{q}+W_0\left(\frac{p}{q}\exp\left(\frac{p}{q}-\beta\right)\right)\right],
    \end{equation}
    where $W_0(\cdot)$ denotes the principal branch of the Lambert $W$ function.

    Consequently, the limiting hit ratio under whole-session LRU reduces to 
    \begin{equation*}
        h_{\mathrm{SCM}}^{\LRU,\infty}=p\left(1-e^{- \mu t^{\mathrm{MF,SCM}}}\right)=p-qW_0\left(\frac{p}{q}\exp\left(\frac{p}{q}-\beta\right)\right).
    \end{equation*}
\end{corollary}

\begin{remark}
    \label{rmk: T and h as marginal functions}
    We list the marginal behavior of the mean-field TTL $t^{\mathrm{MF,SCM}}$ and the limiting hit ratio $h_{\mathrm{SCM}}^{\LRU,\infty}$ as functions of the main system and workload parameters. In \cref{prop: SCM Monotonicity and Unimodality}, we will prove the two nontrivial properties, the monotonicity of $t^{\mathrm{MF,SCM}}$ w.r.t. $ \mu $ and the unimodality of $h_{\mathrm{SCM}}^{\LRU,\infty}$ w.r.t. $p$. The remaining monotonicity results are left to the readers for verification.
    \begin{enumerate}
        \item As the conversation arrival rate $\lambda_0$ increases, the hit intensity coefficient $\beta$ decreases. Moreover, both $t^{\mathrm{MF,SCM}}$ and $h_{\mathrm{SCM}}^{\LRU,\infty}$ decrease.

        \item\label{item: SCM prompt rate monotonicity} As the prompt interarrival rate $ \mu $ increases, the hit intensity coefficient $\beta$, the TTL $t^{\mathrm{MF,SCM}}$, and the limiting hit ratio $h_{\mathrm{SCM}}^{\LRU,\infty}$ all increase.

        \item As the mean turn-wise increment $z$ increases, the hit intensity coefficient $\beta$, the TTL $t^{\mathrm{MF,SCM}}$, and the limiting hit ratio $h_{\mathrm{SCM}}^{\LRU,\infty}$ all decrease.

        \item\label{item: SCM continuation probability behavior} As the continuation probability $p$ increases, the hit intensity coefficient $\beta$ and the TTL $t^{\mathrm{MF,SCM}}$ decrease. However, $h_{\mathrm{SCM}}^{\LRU,\infty}$ is unimodal as a function of $p\in(0,1)$. To be specific, there exists a unique $p^*\in(0,1)$ such that $h_{\mathrm{SCM}}^{\LRU,\infty}$ is strictly increasing when $0<p<p^*$ and strictly decreasing when $p^*<p<1$. Intuitively, this unimodality reflects a trade-off between reuse potential and cache pressure. On the one hand, a larger continuation probability makes a conversation more likely to generate future turns, thereby increasing the amount of history that can potentially be reused. On the other hand, it increases the amount of reusable history generated by background conversations, which accelerates the mean-field displacement process and reduces the opportunity for a tagged session to remain resident until its next access.
    \end{enumerate}
\end{remark}

% \subsection{\textcolor{blue}{Proof of \cref{cor: SCM Closed Form} and its Consequences}}
% \label{ssec: Proof of the closed form of SCM}

We first prove \cref{cor: SCM Closed Form}, the closed form of the limiting mean-field TTL and hit ratio.

\begin{proof}
    Under the SCM, $M$ is geometrically distributed with continuation probability $p$ and termination probability $q$. Hence, $\E[M]=\frac{1}{q},$ and $ \E[M-1]=\frac{p}{q}.$ Using the independence between $M$ and the increment sequence, we have
    \begin{equation*}
        b=z\E[M]=\frac{z}{q}.
    \end{equation*}
    Moreover, as derived in \cref{ssec: The  Hit Ratio}, $a^{\mathrm{cum}}=\frac{z}{q^2},$ and $a^{\mathrm{reu}}=\frac{pz}{q^2}.$
    We also note that the prompt IATs are i.i.d. $Exp( \mu )$-distributed. Their mean is $1/ \mu <\infty$, and $F(t)=1-e^{- \mu t}$ satisfies \cref{eq: assumption Lipschitz} with $L_F= \mu $.

    Substituting $b$, $a^{\mathrm{reu}}$ and $F(\cdot)$ into the mean displacement curve in \cref{def: Mean displacement Curve}, we have that
    \begin{equation*}
        d_{\mathrm{SCM}}(t) = \frac{\lambda_0z}{q}t + \frac{\lambda_0pz}{ \mu q^2}(1-e^{- \mu t}).
    \end{equation*}
    By the definition of the mean-field TTL, $d_{\mathrm{SCM}}(t^{\mathrm{MF,SCM}})=1$. Let $x:= \mu t^{\mathrm{MF,SCM}}$. After multiplying by $\frac{ \mu q}{\lambda_0z}$, the TTL equation becomes
    \begin{equation*}
        x+\frac{p}{q}(1-e^{-x})=\beta\quad \Rightarrow \quad \left(x-\beta+\frac{p}{q}\right)e^x=\frac{p}{q}.
    \end{equation*}
    By letting $y:=x-\beta+\frac{p}{q}$, we rewrite the above equation to be
    \begin{equation*}
        ye^y=\frac{p}{q}\exp\left(\frac{p}{q}-\beta\right) \quad \Rightarrow \quad y=W_0\left(\frac{p}{q}\exp\left(\frac{p}{q}-\beta\right)\right).
    \end{equation*}
    We then substitute $x=y+\beta-\frac{p}{q}$ and claim \cref{equ: SCM TC closed form}.

    Finally, by \cref{thm: Hit Ratio Convergence Full} and $a^{\mathrm{reu}}/a^{\mathrm{cum}}=p$,
    \begin{equation*}
        h_{\mathrm{SCM}}^{\LRU,\infty}=p(1-e^{- \mu t^{\mathrm{MF,SCM}}})=p(1-e^{-x}).
    \end{equation*}
    From the TTL equation,
    \begin{equation*}
        e^{-x}=\frac{q}{p}W_0\left(\frac{p}{q}\exp\left(\frac{p}{q}-\beta\right)\right),
    \end{equation*}
    we conclude that
    \begin{equation*}
        h_{\mathrm{SCM}}^{\LRU,\infty}=p-qW_0\left(\frac{p}{q}\exp\left(\frac{p}{q}-\beta\right)\right).\qedhere
    \end{equation*}
\end{proof}

We now prove \cref{item: SCM prompt rate monotonicity,item: SCM continuation probability behavior} of \cref{rmk: T and h as marginal functions}.

\begin{proposition}
\label{prop: SCM Monotonicity and Unimodality}
Consider the SCM under LRU with fixed $\lambda_0>0$ and $z>0$. We have the following two results:
\begin{enumerate}
    \item For every fixed continuation probability $p\in(0,1)$, the mean-field TTL $t^{\mathrm{MF,SCM}}$ is strictly increasing in $ \mu $, the inter-turn arrival rate of a conversation.

    \item For every fixed $ \mu >0$, the limiting LRU hit ratio $h_{\mathrm{SCM}}^{\LRU,\infty}$ is strictly unimodal in the continuation probability $p\in(0,1)$. In particular, there exists a unique $p^*\in(0,1)$ such that
    \begin{equation*}
    \frac{\partial h_{\mathrm{SCM}}^{\LRU,\infty}}{\partial p}>0 \quad\text{for }p<p^*, \qquad \frac{\partial h_{\mathrm{SCM}}^{\LRU,\infty}}{\partial p}<0 \quad\text{for }p>p^*.
    \end{equation*}
    Consequently, $p^*$ is the unique global maximizer of $h_{\mathrm{SCM}}^{\LRU,\infty}$ over $p\in(0,1)$.
\end{enumerate}
\end{proposition}

\begin{proof}
    We first prove the monotonicity of $t^{\mathrm{MF,SCM}}$ in $ \mu $. By \cref{cor: SCM Closed Form}, $t^{\mathrm{MF,SCM}}$ is the unique solution to
    \begin{equation*}
        \varphi(T, \mu ):=d_{\mathrm{SCM}}(T)-1=\frac{\lambda_0z}{q}T+\frac{\lambda_0pz}{ \mu q^2}(1-e^{- \mu T})-1=0,
    \end{equation*}
    where $q=1-p$. For fixed $p\in(0,1)$, the derivatives of $\varphi$ with respect to $T$ and $ \mu $ are respectively
    \begin{equation*}
        \frac{\partial \varphi}{\partial T}=\frac{\lambda_0z}{q}+\frac{\lambda_0pz}{q^2}e^{- \mu T}>0, \quad \text{and}\quad \frac{\partial \varphi}{\partial  \mu }=\frac{\lambda_0pz}{q^2 \mu ^2}\left[(1+ \mu T)e^{- \mu T}-1\right]<0,
    \end{equation*}
    where the second inequality follows from $(1+x)e^{-x}<1$ for every $x>0$. The implicit-function theorem therefore gives
    \begin{equation*}
        \frac{\partial t^{\mathrm{MF,SCM}}}{\partial  \mu }=-\frac{\varphi_{ \mu }}{\varphi_T}>0.
    \end{equation*}
    Hence, $t^{\mathrm{MF,SCM}}$ is strictly increasing as a function of $ \mu $.

    Next, we prove the unimodality of the mean-field hit ratio as a function of $p$. For simplicity, we introduce the one-to-one reparameterization
    \begin{equation*}
        a:=\frac{p}{q}\in(0,\infty), \qquad \rho:=\frac{ \mu }{\lambda_0z},
    \end{equation*}
    and let $x(a):= \mu t^{\mathrm{MF,SCM}}$. The TTL equation in \cref{cor: SCM Closed Form} becomes
    \begin{equation}
    \label{equ: SCM dimensionless characteristic equation}
        x+a(1-e^{-x})=\frac{\rho}{1+a}.
    \end{equation}
    Moreover, the hit ratio can be represented as
    \begin{equation}
    \label{equ: SCM hit ratio a}
        h_{\mathrm{SCM}}^{\LRU,\infty}(a)=\frac{a}{1+a}(1-e^{-x(a)}).
    \end{equation}

    For notational convenience, we then let $u:=e^{-x}$, $s:=1-e^{-x}$, and $g(x):=x/(1-e^{-x})$. Differentiating both sides of \cref{equ: SCM dimensionless characteristic equation} with respect to $a$, we have
    \begin{equation*}
        x'(a)=-\frac{s+\rho/(1+a)^2}{1+au}<0.
    \end{equation*}
    We similarly differentiate \cref{equ: SCM hit ratio a}, use \cref{equ: SCM dimensionless characteristic equation} to substitute $\rho=(x+as)(1+a)$, and obtain
    \begin{equation}
    \label{equ: SCM hit ratio derivative factorization}
        \frac{d h_{\mathrm{SCM}}^{\LRU,\infty}}{da}=\frac{s}{(1+a)^2}\left[1-R(a)\right],
    \end{equation}
    where
    \begin{equation*}
        R(a):=a\frac{u\left(1+2a+g(x)\right)}{1+au}.
    \end{equation*}

    Because $\frac{s}{(1+a)^2}>0$, it suffices to show that $R(a)$ is strictly increasing from $0$ to $\infty$ as $a$ increases from $0$ to $\infty$. Write $R(a)=A(a)H(a,x(a))$, where
    \begin{equation*}
        A(a):=a>0, \quad \text{and}\quad H(a,x):=\frac{u\left(1+2a+g(x)\right)}{1+au}>0.
    \end{equation*}
    First, $A(a)$ is strictly increasing because
    \begin{equation*}
        A'(a)=1>0.
    \end{equation*}

    We next show that $H(a,x(a))$ is also strictly increasing in $a$. Define
    \begin{equation*}
        v:=ug(x)=\frac{x}{e^x-1}.
    \end{equation*}
    For every $x>0$, we have $0<u<1$ and $0<v<1$. We then differentiate $H$ with respect to $a$ while fixing $x$, and have that
    \begin{equation*}
        \frac{\partial H}{\partial a}=\frac{u(2-u-v)}{(1+au)^2}>0,
    \end{equation*}
    where the inequality follows from $u<1$ and $v<1$.

    To determine the dependence of $H$ on $x$, we rewrite
    \begin{equation*}
        H(a,x)=2+\frac{u+v-2}{1+au}.
    \end{equation*}
    Let the numerator $N(x):=u+v-2$ and the denominator $D(x):=1+au$. Because $0<u<1$ and $0<v<1$, we have $N(x)<0$. Moreover, both $u=e^{-x}$ and $v=x/(e^x-1)$ are strictly decreasing in $x$, and hence $N'(x)<0$ and $D'(x)<0$. Therefore,
    \begin{equation*}
        \frac{\partial H}{\partial x}=\frac{N'(x)D(x)-N(x)D'(x)}{D(x)^2}<0.
    \end{equation*}
    Because $x'(a)<0$, we have that
    \begin{equation*}
        \frac{d}{da}H(a,x(a))=\frac{\partial H}{\partial a}+\frac{\partial H}{\partial x}x'(a)>0.
    \end{equation*}
    Consequently,
    \begin{equation*}
        R'(a)=A'(a)H(a,x(a))+A(a)\frac{d}{da}H(a,x(a))>0,
    \end{equation*}
    and thus $R(a)$ is strictly increasing on $(0,\infty)$.

    It remains to characterize the boundary behavior of $R(a)$. As $a\downarrow0$, $A(a)\rightarrow0$, while $H(a,x(a))$ remains finite. Hence,
    \begin{equation*}
        R(a)\rightarrow0.
    \end{equation*}
    As $a\rightarrow\infty$, \cref{equ: SCM dimensionless characteristic equation} implies
    \begin{equation*}
        0<x(a)\leq\frac{\rho}{1+a}\rightarrow0.
    \end{equation*}
    Therefore, $u=e^{-x}\rightarrow1$ and $g(x)=x/(1-e^{-x})\rightarrow1$, which gives that
    \begin{equation*}
        H(a,x(a))\rightarrow2.
    \end{equation*}
    Meanwhile, $A(a)=a\rightarrow\infty$, and hence $R(a)\rightarrow\infty$.

    Therefore, $R(a)$ is continuous and strictly increasing from $0$ to $\infty$, so there exists a unique $a^*\in(0,\infty)$ such that $R(a^*)=1$. By \cref{equ: SCM hit ratio derivative factorization},
    \begin{equation*}
        \frac{d h_{\mathrm{SCM}}^{\LRU,\infty}}{da}>0 \quad \text{for }a<a^*, \qquad \frac{d h_{\mathrm{SCM}}^{\LRU,\infty}}{da}<0 \quad \text{for }a>a^*.
    \end{equation*}
    Hence, $a^*$ is the unique global maximizer of $h_{\mathrm{SCM}}^{\LRU,\infty}$ over $a\in(0,\infty)$. Finally, because the mapping $p\mapsto a=p/(1-p)$ is strictly increasing, $p^*:=\frac{a^*}{1+a^*}$ is the unique global maximizer of $h_{\mathrm{SCM}}^{\LRU,\infty}$ over $p\in(0,1)$.
\end{proof}

\section{Proof of \cref{thm: practical mixed workflow mean field}}
\label{app: direct block-level practical MCM}

In this section, we will prove \cref{thm: practical mixed workflow mean field}. 
Our proof still builds on the same argument of the chain of displacement process--TTL--hit ratio convergence.
% We will first show that the normalized displacement process converges to a mean displacement curve in \cref{lem: direct practical block displacement}.

% This proof is included to make clear that the practical displacement curve, the block-level cache boundary, and the prefix-hit rule can be handled within the same argument.

For a workflow-$c$ conversation after turn $j$, let $Q_{c,j}:=R_{c,j}-1$
denote the number of hashable and reusable history blocks. 
Recall that $B_c=\sum_{j=1}^{M_c}Z_{c,j}$ is the total lifetime newly prefilled workload and that
\begin{equation*}
    b_c=\E[B_c],\qquad
    a_c^{\mathrm{reu,hash}}=\E\!\left[\sum_{j=1}^{M_c-1}Q_{c,j}\right].
\end{equation*}
Let $D_{c,\mathrm{blo}}^{(N),\mathrm{prac}}(t)$ denote the number of distinct workflow-$c$ background blocks that become more recent than a tagged block during $(0,t]$. Under LRU, define
\begin{equation*}
    D_{\LRU,\mathrm{blo}}^{(N),\mathrm{prac}}(t)
    :=\sum_{c=1}^J D_{c,\mathrm{blo}}^{(N),\mathrm{prac}}(t).
\end{equation*}

We now demonstrate the sample path convergence results for $D_{c,\mathrm{blo}}^{(N),\mathrm{prac}}(\cdot)$ and $D_{\LRU,\mathrm{blo}}^{(N),\mathrm{prac}}(\cdot)$.

\begin{lemma}[Practical Block-Level Displacement Convergence]
\label{lem: direct practical block displacement}
Consider the mixed-workflow MCM. For every workflow $c$,
\begin{align}
    &\frac{D_{c,\mathrm{blo}}^{(N),\mathrm{prac}}(\cdot)}{N}\overset{\mathbb P}{\to}d_c^{\mathrm{prac}}\quad\text{in }\mathbb D,
    \label{equ: direct practical workflow displacement convergence}\\
    &\frac{D_{\LRU,\mathrm{blo}}^{(N),\mathrm{prac}}(\cdot)}{N}\overset{\mathbb P}{\to}d_{\LRU}^{\mathrm{prac}}\quad\text{in }\mathbb D.
    \label{equ: direct practical LRU displacement convergence}
\end{align}
\end{lemma}

\begin{proof}
Consider the marked Poisson random measure constructed in \cref{sssec: Representation of the Displacement Process}, and retain only workflow-$c$ marks.
As in the proof of \cref{thm: Displacement Process Convergence},
a complete workflow-$c$ conversation contributes its lifetime newly prefilled workload $B_c$, which gives the linear mean displacement $\lambda_c b_ct$. 
In addition, a hashable history block is moved ahead of the tagged block once more if the next prompt arrives within the observation interval.
% By aggregating these blocks over the stationary turn-state population, we have that
% \begin{equation*}
%     \lambda_ca_c^{\mathrm{reu,hash}}
%     \int_0^t\overline F_c(u)\,du.
% \end{equation*}
Consequently, the mean displacement curve for each workflow $c$ is given by
\begin{equation*}
    \frac{1}{N}\E\!\left[D_{c,\mathrm{blo}}^{(N),\mathrm{prac}}(t)\right]
    =d_c^{\mathrm{prac}}(t).
\end{equation*}
Fix $0<T<\infty$. The contribution of one marked conversation on $[0,T]$ is bounded by the same square-integrable envelope used in \cref{equ: PRM contribution bound}, where the original lifetime workload is replaced by $B_c$.
Therefore, we can apply the Poisson-integral variance bound and the partition argument in the proof of \cref{thm: Displacement Process Convergence} without modification. The resulting supremum error converges to zero in probability on each bounded interval. This proves $D_{c,\mathrm{blo}}^{(N),\mathrm{prac}}(\cdot)/N\overset{\mathbb P}{\to}d_c^{\mathrm{prac}}\text{ in }\mathbb D$. 
Finally, by summing over the $J$ independent workflows, we obtain $D_{\LRU,\mathrm{blo}}^{(N),\mathrm{prac}}(\cdot)/N\overset{\mathbb P}{\to}d_{\LRU}^{\mathrm{prac}}\text{ in }\mathbb D$.
\end{proof}

We next bound the TTL for which a reusable tagged conversation session remains in memory. 
Suppose that a workflow-$c$ conversation is tagged immediately after turn $j$. 
Under block-level LRU, define the lower TTL and the upper TTL by
\begin{align*}
    \tau_{c,j,-}^{(N),\mathrm{prac}}
    &:=
    \begin{cases}
        0, & Q_{c,j}>N,\\
        \inf\{t\geq0:D_{\LRU,\mathrm{blo}}^{(N),\mathrm{prac}}(t)>N-Q_{c,j}\}, & Q_{c,j}\leq N,
    \end{cases}
    \\
    \tau_{c,j,+}^{(N),\mathrm{prac}}
    &:=\inf\{t\geq0:D_{\LRU,\mathrm{blo}}^{(N),\mathrm{prac}}(t)>N\},
\end{align*}
respectively.
Specifically, the lower TTL is the first time that the entire reusable session is no longer resident, whereas the upper TTL is the first time that none of the blocks in this session remain resident.

\begin{lemma}[Practical Block-Level TTL Convergence]
\label{lem: direct practical block TTL}
Consider the mixed-workflow MCM. 
For every workflow $c$ and fixed turn state $j$ with $s_{c,j}>0$, the lower and upper TTLs converge to the same limit in probability:
\begin{equation}
    \tau_{c,j,-}^{(N),\mathrm{prac}}
    \xrightarrow{\pro}t_{\LRU}^{\mathrm{MF,prac}},
    \qquad
    \tau_{c,j,+}^{(N),\mathrm{prac}}
    \xrightarrow{\pro}t_{\LRU}^{\mathrm{MF,prac}}.
    \label{equ: direct practical block TTL convergence}
\end{equation}
\end{lemma}

\begin{proof}
Fix $\varepsilon>0$. Because 
$(d_{\LRU}^{\mathrm{prac}})'(t)\geq c_{\LRU}^{\mathrm{prac}}:=\sum_{r=1}^J\lambda_rb_r>0$, we have
\begin{equation*}
    d_{\LRU}^{\mathrm{prac}}(t_{\LRU}^{\mathrm{MF,prac}}-\varepsilon)
    \leq1-c_{\LRU}^{\mathrm{prac}}\varepsilon,
    \qquad
    d_{\LRU}^{\mathrm{prac}}(t_{\LRU}^{\mathrm{MF,prac}}+\varepsilon)
    \geq1+c_{\LRU}^{\mathrm{prac}}\varepsilon.
\end{equation*}
Let $\eta:=c_{\LRU}^{\mathrm{prac}}\varepsilon/3$. On the event
\begin{equation*}
    \sup_{0\leq t\leq t_{\LRU}^{\mathrm{MF,prac}}+\varepsilon}
    \left|\frac{D_{\LRU,\mathrm{blo}}^{(N),\mathrm{prac}}(t)}{N}
    -d_{\LRU}^{\mathrm{prac}}(t)\right|\leq\eta,
    \qquad \frac{Q_{c,j}}{N}\leq\eta,
\end{equation*}
both TTLs in $\tau_{c,j,-}^{(N),\mathrm{prac}}$ and $\tau_{c,j,+}^{(N),\mathrm{prac}}$ lie in
$[t_{\LRU}^{\mathrm{MF,prac}}-\varepsilon,
t_{\LRU}^{\mathrm{MF,prac}}+\varepsilon]$. Therefore,
\begin{align*}
\pro\!\left(
\max_{\sigma\in\{-,+\}}
\left|\tau_{c,j,\sigma}^{(N),\mathrm{prac}}
-t_{\LRU}^{\mathrm{MF,prac}}\right|>\varepsilon
\right)&\leq
\pro\!\left(
\sup_{0\leq t\leq t_{\LRU}^{\mathrm{MF,prac}}+\varepsilon}
\left|\frac{D_{\LRU,\mathrm{blo}}^{(N),\mathrm{prac}}(t)}{N}
-d_{\LRU}^{\mathrm{prac}}(t)\right|>\eta
\right)
\\[-2pt]
&\quad +\pro\!\left(Q_{c,j}>N\eta\right).
\end{align*}
The first term converges to zero by \cref{lem: direct practical block displacement} because its continuous limit $d_{\LRU}^{\mathrm{prac}}$ gives uniform convergence in probability on this interval.
Moreover, $Q_{c,j}\leq B_c$, and \cref{eq:assumption B} implies $\E[B_c]<\infty$.
Hence, the second term converges to zero by Markov's inequality. 
We therefore conclude \cref{equ: direct practical block TTL convergence}.
\end{proof}

We now prove \cref{thm: practical mixed workflow mean field}.

\begin{proof}[Proof of \cref{thm: practical mixed workflow mean field}]
Let $H_{c,j+1}^{(N),\mathrm{blo,prac}}$ denote the number of hashable history blocks reused when the turn-$(j+1)$ prompt of a tagged workflow-$c$ conversation arrives.
Because block-level eviction proceeds from the suffix blocks toward the prefix, any partially resident session retains a reusable prefix of blocks.
Therefore,
\begin{equation}
    Q_{c,j}\one\!\left\{W_{c,j}<\tau_{c,j,-}^{(N),\mathrm{prac}}\right\}
    \leq H_{c,j+1}^{(N),\mathrm{blo,prac}}
    \leq Q_{c,j}\one\!\left\{W_{c,j}<\tau_{c,j,+}^{(N),\mathrm{prac}}\right\}.
    \label{equ: direct practical prefix hit sandwich}
\end{equation}
For every fixed $c$ and $j$, the prompt IAT $W_{c,j}$ is independent of the background displacement process. By \cref{lem: direct practical block TTL}, the continuity of $F_c$, and uniform integrability of $Q_{c,j}$, both expectations in the lower and upper bounds of \cref{equ: direct practical prefix hit sandwich} converge to
\begin{equation*}
    \E[Q_{c,j}\mid M_c\geq j+1]
    F_c(t_{\LRU}^{\mathrm{MF,prac}}).
\end{equation*}
As in the proof of \cref{thm: Hit Ratio Convergence Full}, we first sum over a finite number of turn states.
The remaining tail is bounded by the tail reusable workload, which vanishes because \cref{eq:assumption B} implies $\E[B_c^2]<\infty$ and $M_c\leq B_c$. 
Letting the truncation level tend to infinity, we have that
\begin{equation*}
    h_c^{\LRU,(N)} \to\frac{a_c^{\mathrm{reu,hash}}}{a_c^{\mathrm{cum}}}
    F_c(t_{\LRU}^{\mathrm{MF,prac}})\quad \text{as } N\to\infty.
\end{equation*}
By weighting the workflow-$c$ hit and requested-block rates by $\lambda_c$ and summing over $c$, we finally conclude \cref{equ: practical mixed LRU server limit main}.

For $\vec\alpha$-WAQ-LRU, we apply the same proof to the workflow-$c$ displacement process $D_{c,\mathrm{blo}}^{(N),\mathrm{prac}}$, replace the common threshold $N$ by $\lfloor\alpha_cN\rfloor$, and replace $t_{\LRU}^{\mathrm{MF,prac}}$ by the unique solution $t_c^{\mathrm{MF,prac}}(\alpha_c)$ of
\begin{equation*}
    d_c^{\mathrm{prac}}(t)=\alpha_c.
\end{equation*}
As a result, we prove \cref{equ: practical mixed WAQ workflow limit main,equ: practical mixed WAQ server limit main} for every $\alpha_c>0$, and complete the proof.

\end{proof}

\section{Additional Experimental Results}
\label{app: additional experiments}

In this section, we provide the additional experimental results. 
In \cref{appssec: Additional LRU hit ratio estimator}, we validate the LRU estimator using two additional model and workload combinations.
In \cref{appssec: Additional WAQ-LRU Empirical Performance}, we compare WAQ-LRU with LRU using two additional workflow mixtures from ShareGPT.

\subsection{Additional Experimental Results for LRU Hit Ratio Estimator}
\label{appssec: Additional LRU hit ratio estimator}

In these experiments, we implement two model-workload combinations and empirically demonstrate our estimator's practical accuracy.
We use the same measurement protocol and hit ratio definition as in \cref{ssec: LRU estimator real practice}. 
The prefill-decode disaggregation structure for each experiment will be specified later. 
% We report the model and workload differences and the resulting estimator errors.

\subsubsection{Qwen3-8B Model and ShareGPT Workload.}
\label{app: Qwen3-8B ShareGPT estimator validation}

We first deploy Qwen3-8B on five Ascend 910B2 NPUs, with one NPU serving as the prefiller. We reserve 40 GiB of the prefiller HBM for KV cache. 
Under BF16 precision, one token occupies 144 KiB of KV-cache memory. 
As a result, a 128-token block occupies 18 MiB, and the cache capacity is $N=\lfloor40\,\mathrm{GiB}/(18\,\mathrm{MiB})\rfloor=2,275$ blocks. 
We use the ShareGPT workload and measurement protocol described in \cref{ssec: LRU estimator real practice}.
In this configuration, the decoder generates the same number of tokens as the corresponding recorded response to preserve the decoding workload, while generated text and its response KV cache are discarded.

We first fix $\lambda_0=1.5$ conversations per second and vary the target mean IAT from 90 s to 240 s. 
As shown in \cref{fig: Qwen8 fixed arrival analysis}, the estimator captures the decrease in hit ratio as the mean IAT increases. 
When the target mean IAT is small, because we deliberately increase the decoder workload and increase the E2E,
$W_{ij}^{\mathrm{E2E}}>W_{ij}^{\mathrm{sample}}$ occurs more frequently. 
Hence, the realized IAT distribution deviates more from the target exponential distribution and the estimator tends to overestimate the hit ratio.
Nevertheless, the absolute error remains below $0.02$ and the relative error below $8\%$.

\begin{figure}[ht]
    \centering
    \begin{subfigure}[t]{0.45\linewidth}
        \centering
        \includegraphics[width=\linewidth]{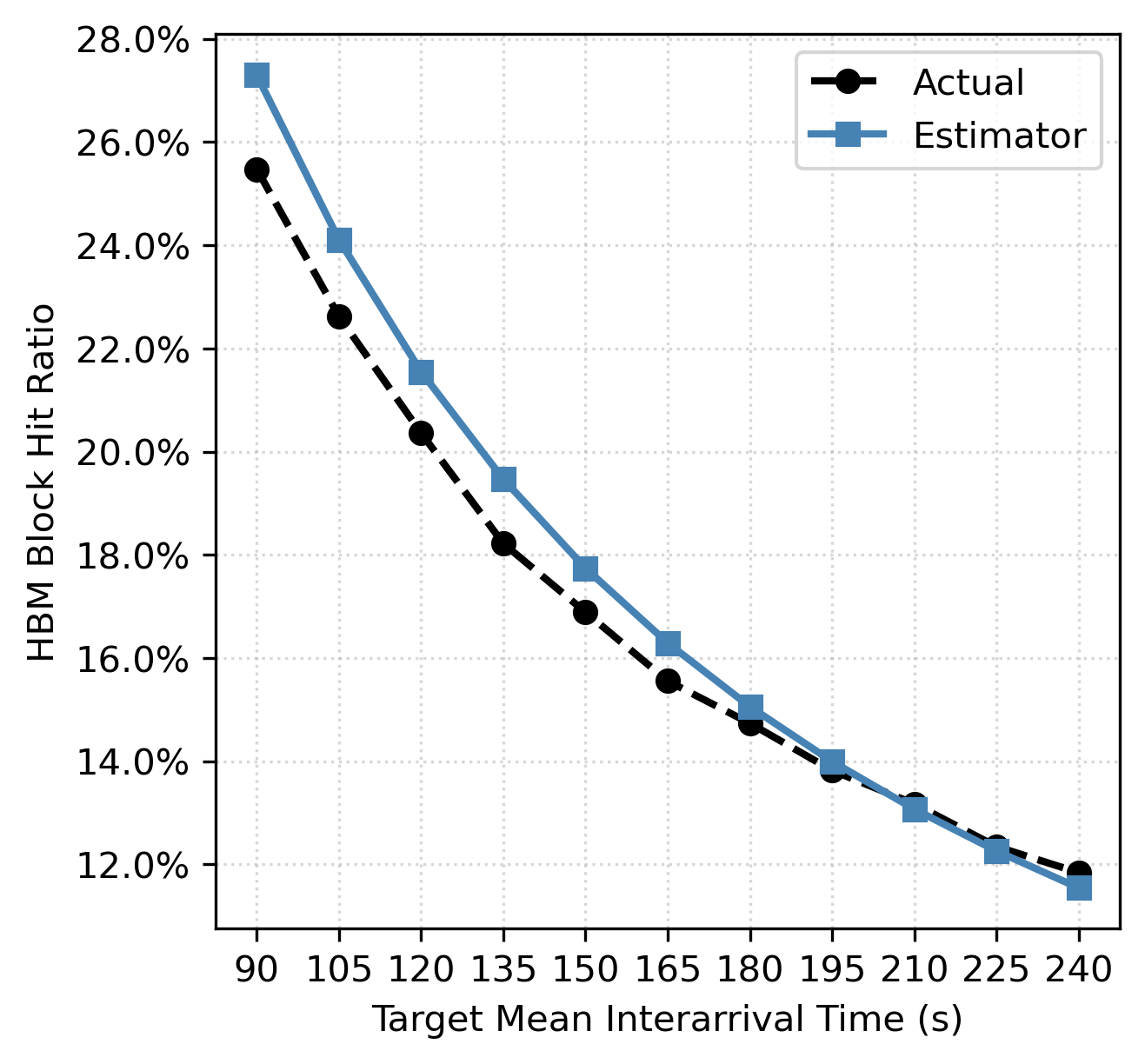}
        \caption{Measured and estimated hit ratios.}
        \label{fig: E1_actual_estimated_hit_ratio}
    \end{subfigure}
    \hfill
    \begin{subfigure}[t]{0.45\linewidth}
        \centering
        \includegraphics[width=\linewidth]{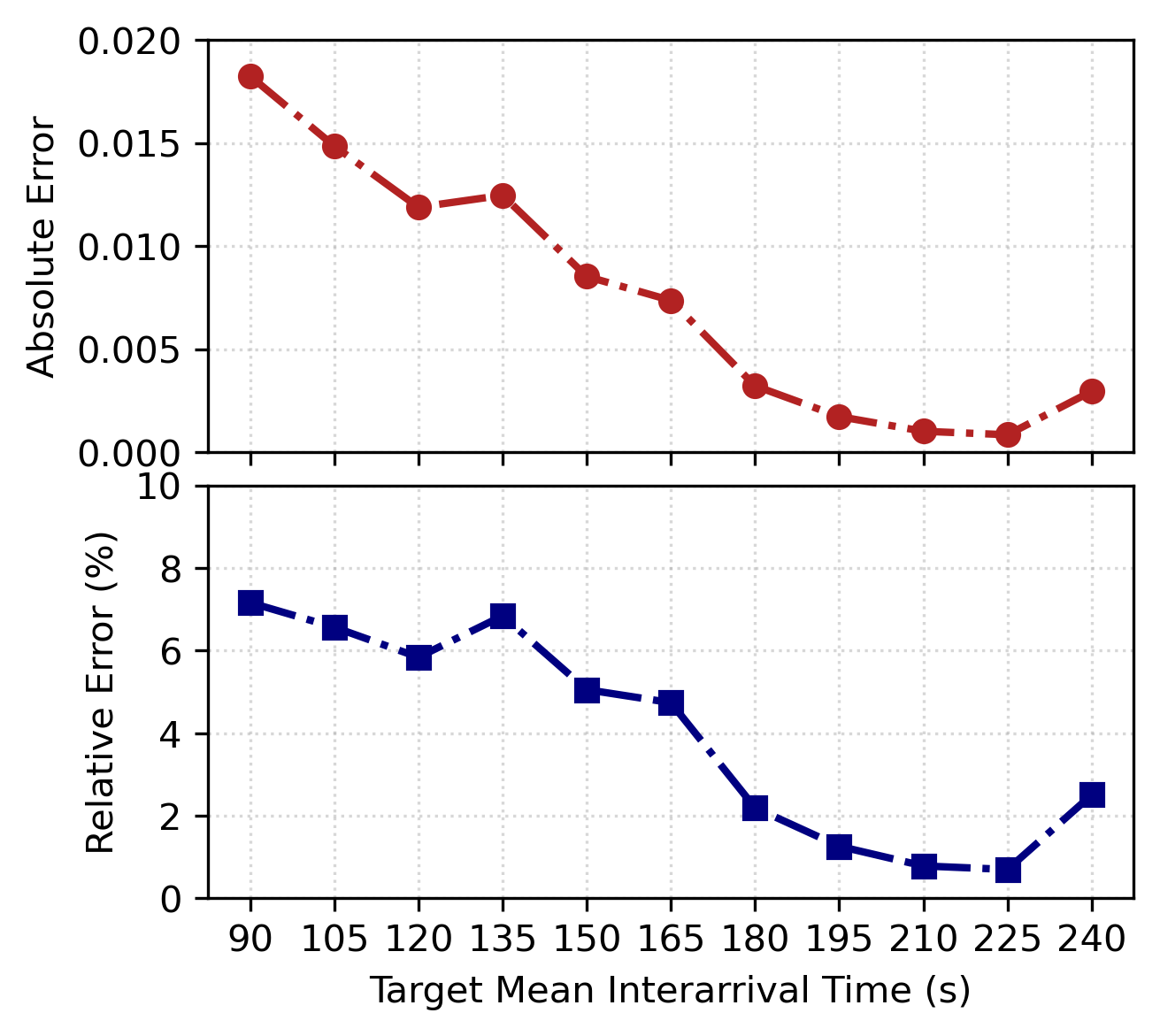}
        \caption{Absolute and relative errors.}
        \label{fig: E1-1-errors}
    \end{subfigure}
    \caption{Qwen3-8B hit ratio estimation with fixed conversation-arrival rate $\lambda_0=1.5$ per second and target mean IATs from 90 s to 240 s.}
    \label{fig: Qwen8 fixed arrival analysis}
\end{figure}

We then fix the target mean IAT at $\mu_F=135$ s and vary $\lambda_0$ from $0.5$ to $2.0$ conversations per second. 
\Cref{fig: Qwen8 fixed IAT analysis} shows that the estimator again captures the reduction in hit ratio as the arrival load increases, with absolute error below $0.02$ and relative error below $10\%$.

\begin{figure}[ht]
    \centering
    \begin{subfigure}[t]{0.45\linewidth}
        \centering
        \includegraphics[width=\linewidth]{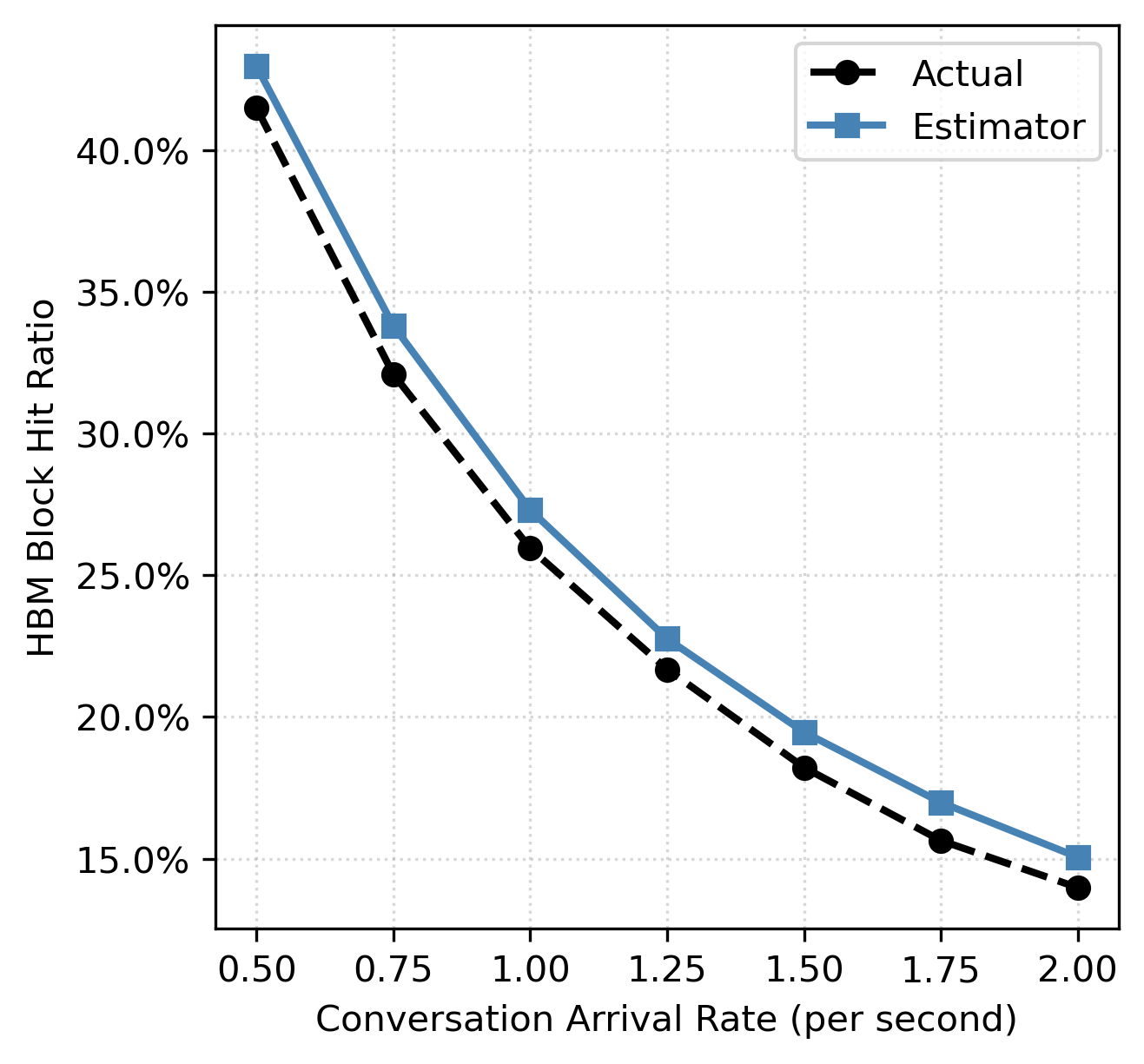}
        \caption{Measured and estimated hit ratios.}
        \label{fig: E1-2_arrival-rate-hit ratio}
    \end{subfigure}
    \hfill
    \begin{subfigure}[t]{0.45\linewidth}
        \centering
        \includegraphics[width=\linewidth]{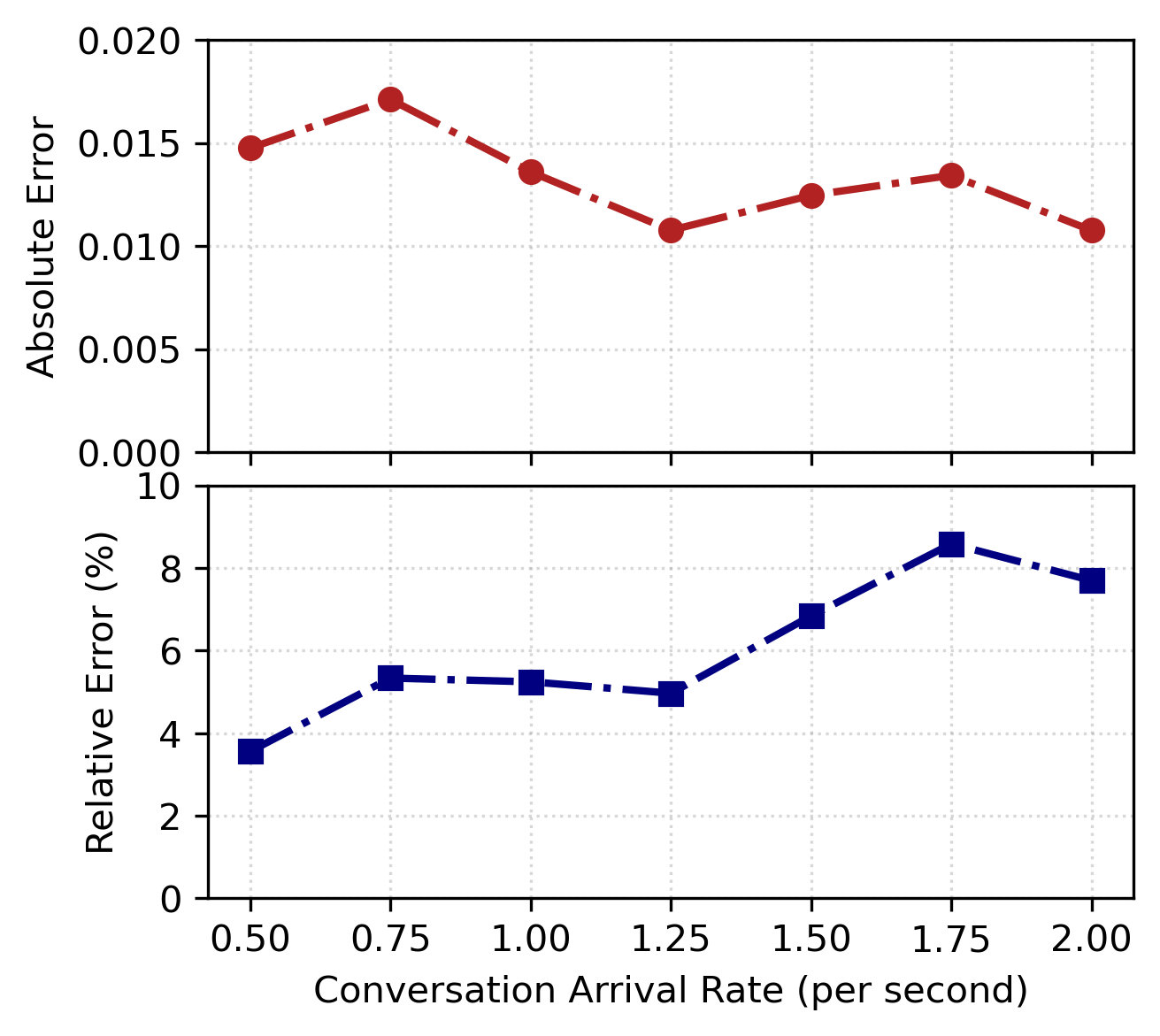}
        \caption{Absolute and relative errors.}
        \label{fig: E1-2-errors}
    \end{subfigure}
    \caption{Qwen3-8B hit ratio estimation with fixed target mean IAT $\mu_F=135$ s and varying conversation-arrival rate $\lambda_0$.}
    \label{fig: Qwen8 fixed IAT analysis}
\end{figure}

\subsubsection{Qwen3-4B Model and WildChat Workload.}
\label{app: Qwen3-4B WildChat estimator validation}

In this configuration, we then deploy Qwen3-4B using the same prefill-decode disaggregated structure as the experiment in \cref{app: Qwen3-8B ShareGPT estimator validation}. 
Because Qwen3-4B has fewer parameters than Qwen3-8B, we allocate $47.87$ GiB of prefiller HBM exclusively to KV cache. 
The KV-cache block size is 128 tokens, and each block occupies $18$ MiB.
Hence, the HBM can store at most $N=\lfloor47.87\,\mathrm{GiB}/(18\,\mathrm{MiB})\rfloor=2,723$ KV blocks.

We sample 6,000 conversations, each containing at least two turns, from the WildChat trace. 
The warm-up set contains 1,000 conversations and 4,467 turns, and the remaining 5,000 conversations form the measurement set.
We use the same arrival-process and measurement-window construction as in the main experiment and force each decode procedure to generate one token, which allows us to test shorter mean-IAT regimes.

The WildChat workload has greater variation in the number of turns and new blocks per turn than the ShareGPT workload.
\Cref{fig: E2 Dataset info} summarizes these two distributions. 
Specifically, we have that
\[\widehat a^{\mathrm{cum}} = 76.8834,\quad\widehat a^{\mathrm{reu,hash}} = 55.4054, \quad \widehat b = 21.4780,\quad \text{and} \quad\widehat\mu_M = 4.3814.\]

\begin{figure}[ht]
    \centering
    \begin{subfigure}[t]{0.48\linewidth}
        \centering
        \includegraphics[width=\linewidth]{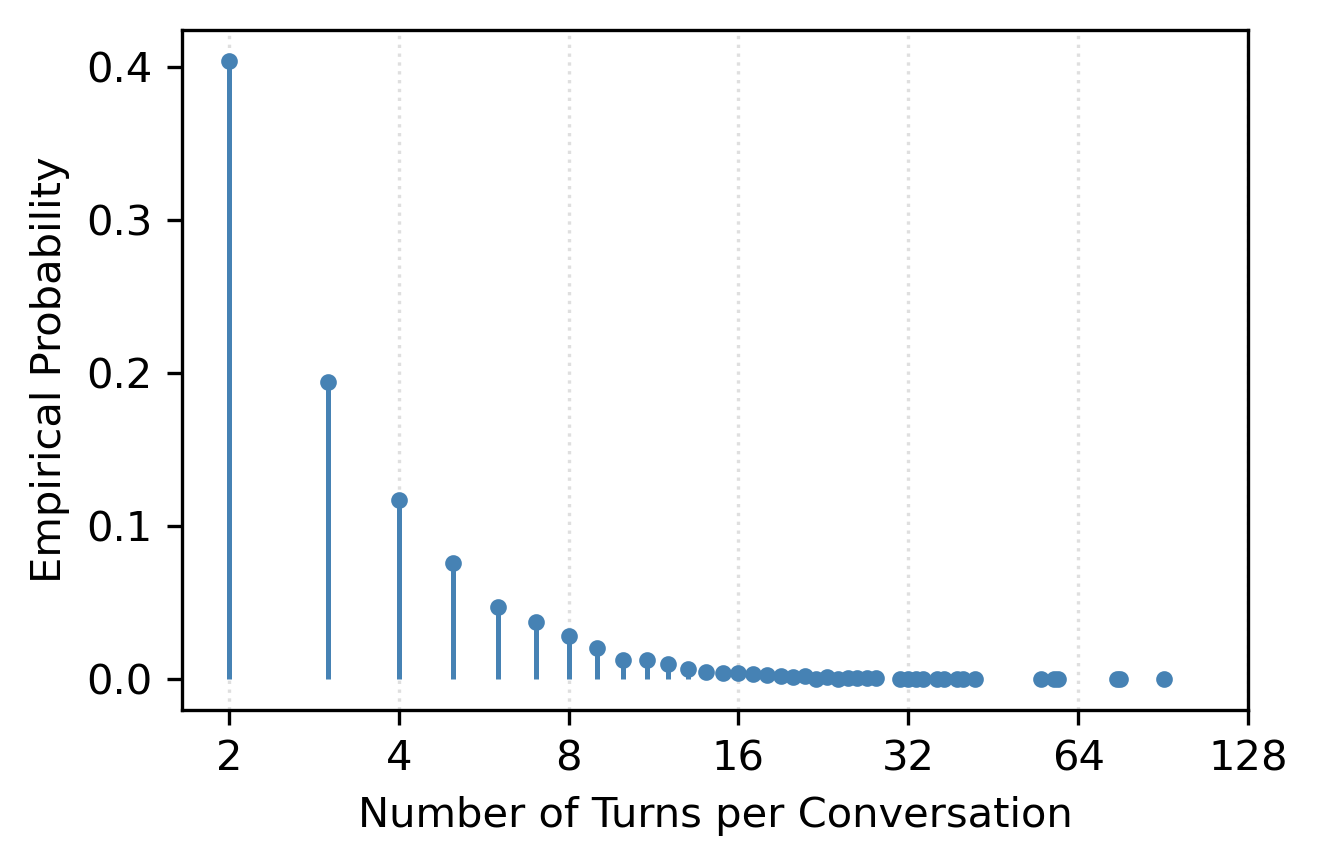}
        \caption{Number of turns per conversation.}
        \label{fig: E2_empirical_conversation_turn_distribution}
    \end{subfigure}
    \hfill
    \begin{subfigure}[t]{0.48\linewidth}
        \centering
        \includegraphics[width=\linewidth]{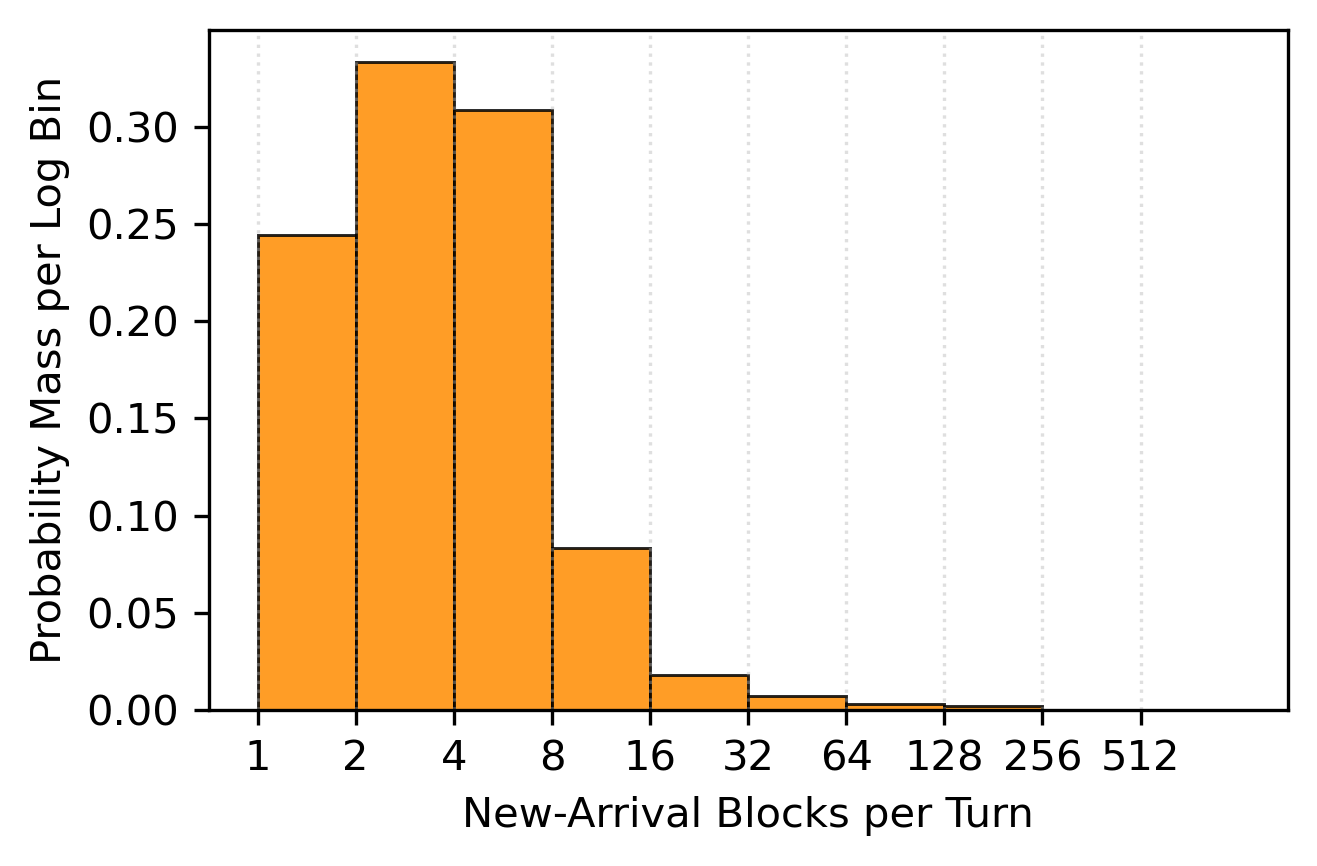}
        \caption{New KV blocks per turn, in logarithmic bins.}
        \label{fig: E2_new_kv_block_histogram}
    \end{subfigure}
    \caption{
    Empirical workload characteristics of the WildChat measurement set, which contains 5,000 multi-turn conversations and 21,907 turns.
    (a) Empirical probability mass function of the number of turns per conversation, with mean $4.3814$, variance $19.6967$, and maximum $91$.
    (b) Distribution of the number of newly computed physical KV blocks per turn, aggregated into logarithmic bins. The distribution in (b) is shown only to characterize the marginal workload and does not impose an i.i.d. assumption across turn indices.
    }
    \label{fig: E2 Dataset info}
\end{figure}

In \cref{fig: E3-1 Error Analysis}, we compare the actual hit ratio and our mean-field estimator, with fixed conversation arrival rate $\lambda_0 = 2$ per second, and varying mean IAT $\mu_F$.
We note that our estimator maintains a relative error below $2\%$.

\begin{figure}[ht]
    \centering
    \begin{subfigure}[t]{0.48\linewidth}
        \centering
        \includegraphics[width=\linewidth]{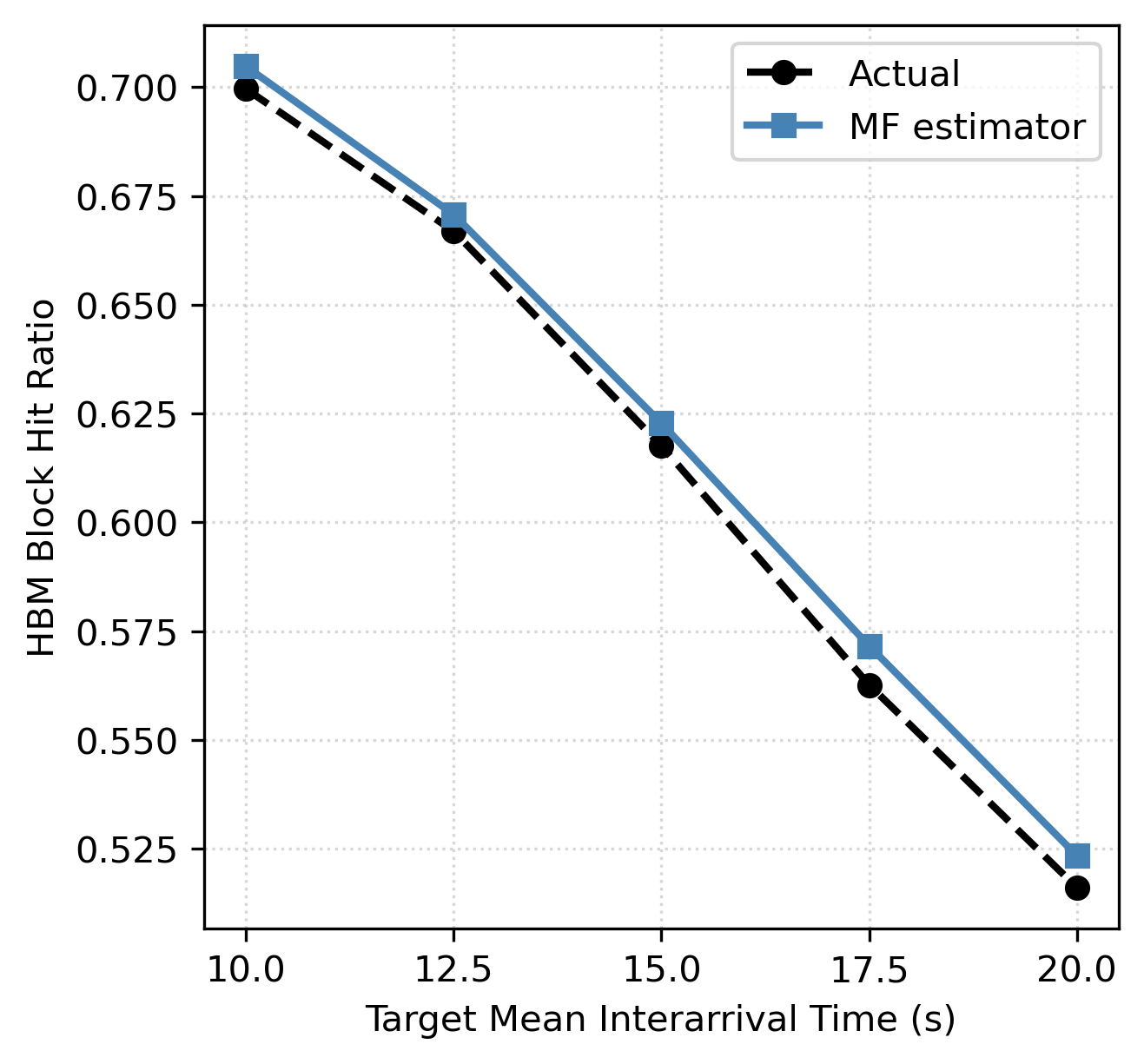}
        \caption{Measured and estimated hit ratios.}
        \label{fig: E2_actual_prior_posterior_hit_ratio}
    \end{subfigure}
    \hfill
    \begin{subfigure}[t]{0.48\linewidth}
        \centering
        \includegraphics[width=\linewidth]{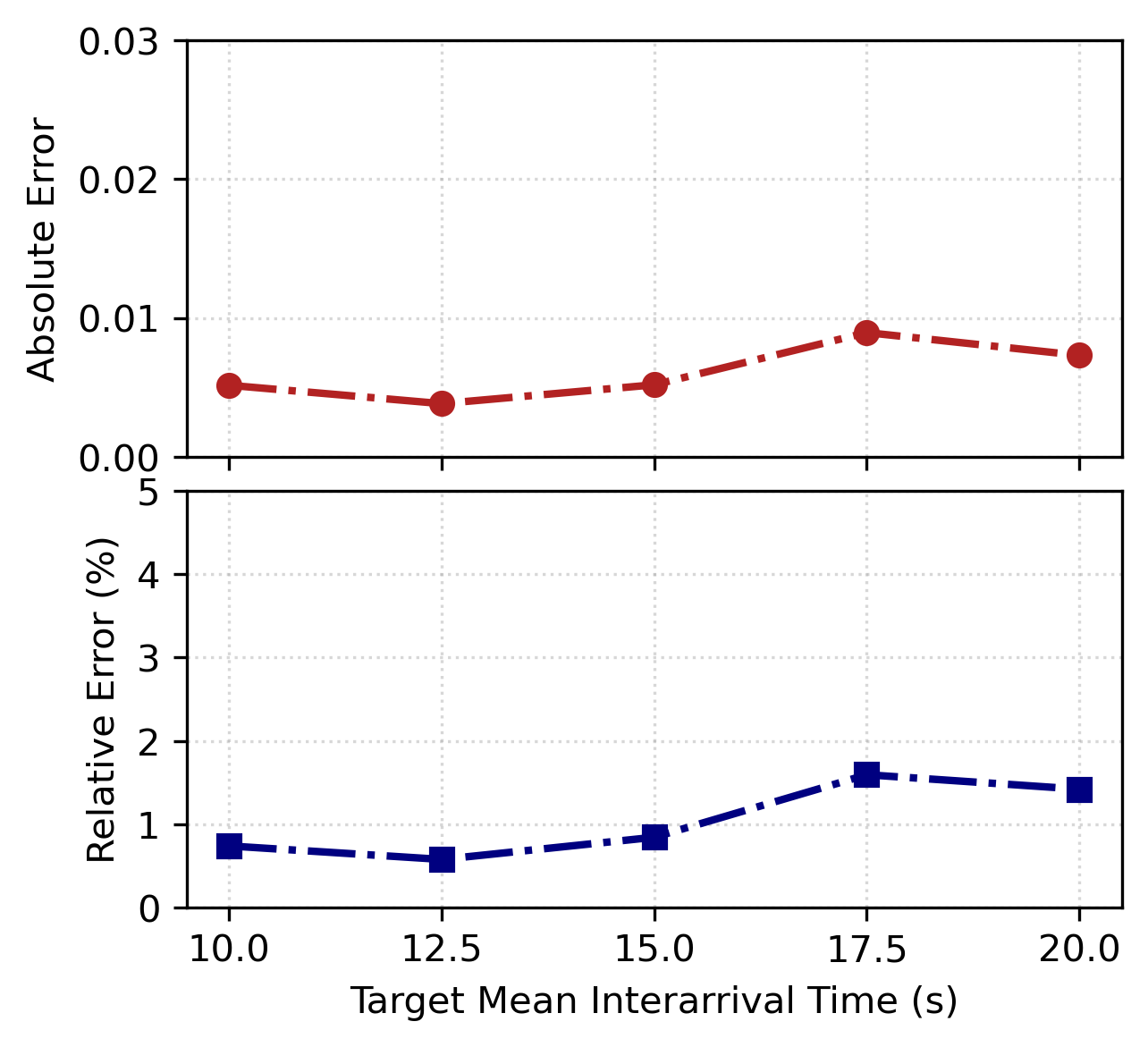}
        \caption{Absolute and relative errors.}
        \label{fig: E2_posterior_estimator_errors}
    \end{subfigure}
    \caption{Qwen3-4B hit ratio estimation with conversation arrival rate $\lambda_0=2$ per second and target mean IATs from $10$ s to $20$ s.}
    \label{fig: E3-1 Error Analysis}
\end{figure}

We then fix the target mean IAT at $12.5$ s and vary the conversation arrival rate from $1$ to $3$ per second. 
As shown in \cref{fig: E2-2 Error Analysis}, the absolute error remains below $0.03$ and the relative error below $5\%$.

\begin{figure}[ht]
    \centering
    \begin{subfigure}[t]{0.48\linewidth}
        \centering
        \includegraphics[width=\linewidth]{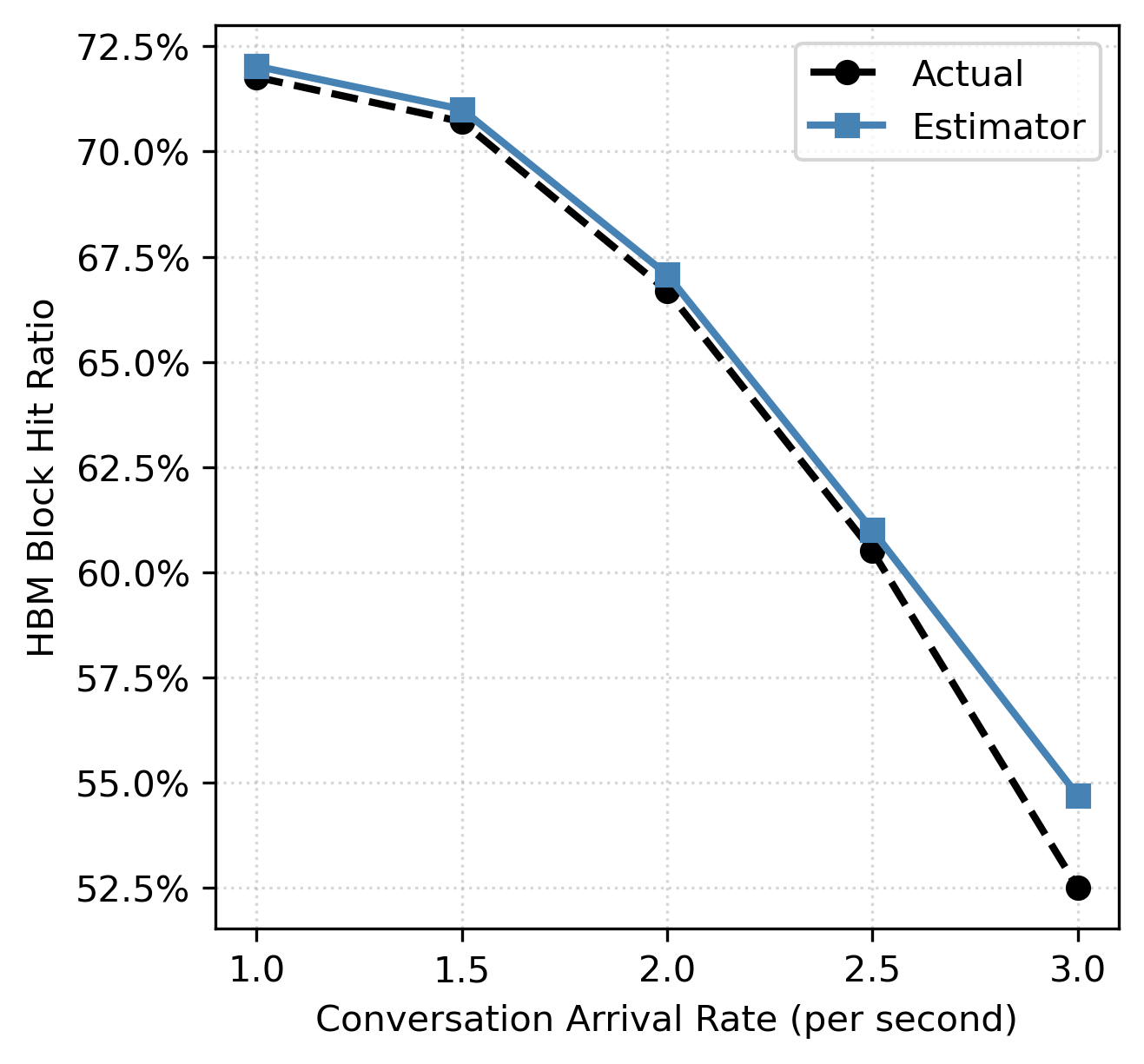}
        \caption{Measured and estimated hit ratios.}
        \label{fig: 2-2_arrival_rate_actual_prior_posterior_hit_ratio}
    \end{subfigure}
    \hfill
    \begin{subfigure}[t]{0.48\linewidth}
        \centering
        \includegraphics[width=\linewidth]{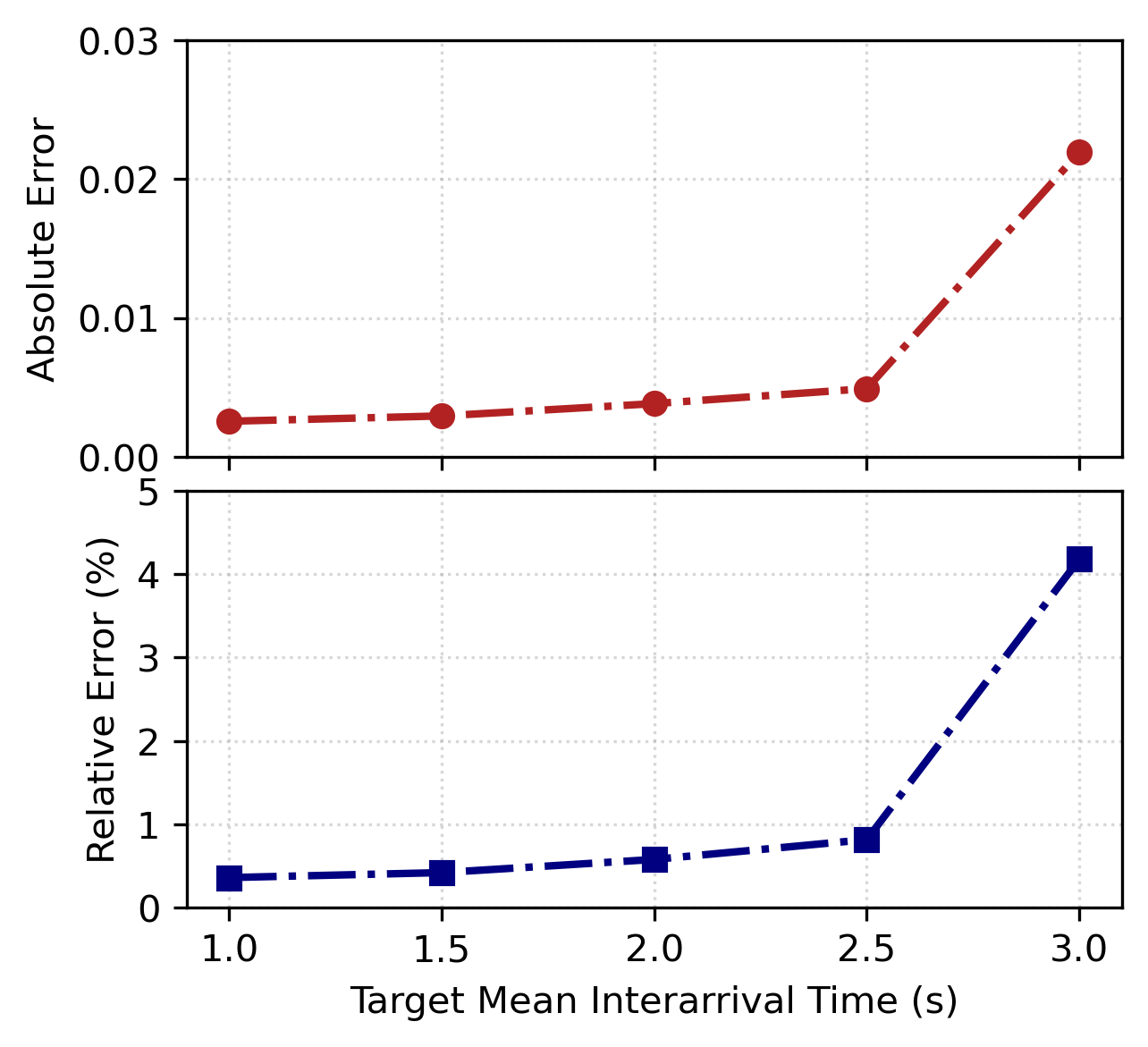}
        \caption{Absolute and relative errors.}
        \label{fig: 2-2_arrival_rate_estimator_errors}
    \end{subfigure}
    \caption{Qwen3-4B hit ratio estimation with conversation arrival rate from $1$ to $3$ per second and target mean IAT $\mu_F=12.5$ s.}
    \label{fig: E2-2 Error Analysis}
\end{figure}

\subsection{Additional Experimental Results for WAQ-LRU Empirical Performance}
\label{appssec: Additional WAQ-LRU Empirical Performance}

In the following groups of experiments, we keep using the same Qwen3-32B deployment, cache capacity, warm-up and measurement procedure, and TTFT definition as in \cref{ssec: WAQ-LRU estimator real practice}. 
The LRU and WAQ-LRU runs use the same conversations and arrival traces.
Likewise, we disable the CPU KV cache reuse to better study HBM hit ratio's influence on the TTFT.

\subsubsection{ShareGPT Creative/Roleplay and Legal/Business/Finance Workflow Mixture.}
\label{app: WAQ group A}

In Group A, Workflow 1 contains 2,000 ShareGPT Creative/Roleplay conversations, and Workflow 2 contains 3,000 ShareGPT Legal/Business/Finance conversations. 
We set their conversation arrival rates to be $\widehat\lambda_1(r)=0.4r$ and $\widehat\lambda_2(r)=0.6r$, 
and their exponential target IATs have means $\widehat\mu_1(r)=210/r$ s and $\widehat\mu_2(r)=1050/r$ s. The practical MCM statistics are
\begin{align*}
\text{Workflow 1:}\quad
&\widehat a_1^{\mathrm{cum}}=26.8065,\quad
\widehat a_1^{\mathrm{reu,hash}}=14.6775,\quad
\widehat b_1={12.1290},\\
\text{Workflow 2:}\quad
&\widehat a_2^{\mathrm{cum}}=22.9880,\quad
\widehat a_2^{\mathrm{reu,hash}}=11.6140,\quad
\widehat b_2={11.3740}.
\end{align*}
By maximizing \cref{equ: practical WAQ hit ratio estimator}, we obtain the following optimal quota for two workflows:
\begin{equation*}
    (N_1^*,N_2^*)=(3,585,1,022),
    \qquad
    \widehat{\vec\alpha}^{\,*}=(0.7782,0.2218).
\end{equation*}

The estimated limiting hit ratios are $20.45\%$ under LRU and $23.78\%$ under WAQ-LRU.
We then sweep $r$ from $2.4$ to $6.4$ conversations per second in increments of $0.4$.

\begin{figure}[ht]
    \centering
    \begin{subfigure}[t]{0.32\linewidth}
        \centering
        \includegraphics[width=\linewidth]{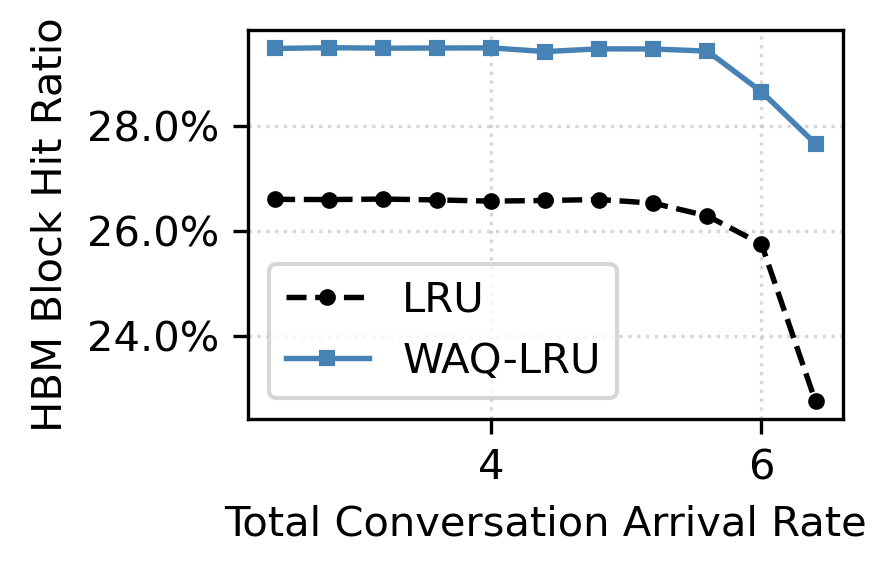}
        \caption{Measured HBM hit ratios.}
        \label{fig: A measured hit ratios}
    \end{subfigure}
    \hfill
    \begin{subfigure}[t]{0.32\linewidth}
        \centering
        \includegraphics[width=\linewidth]{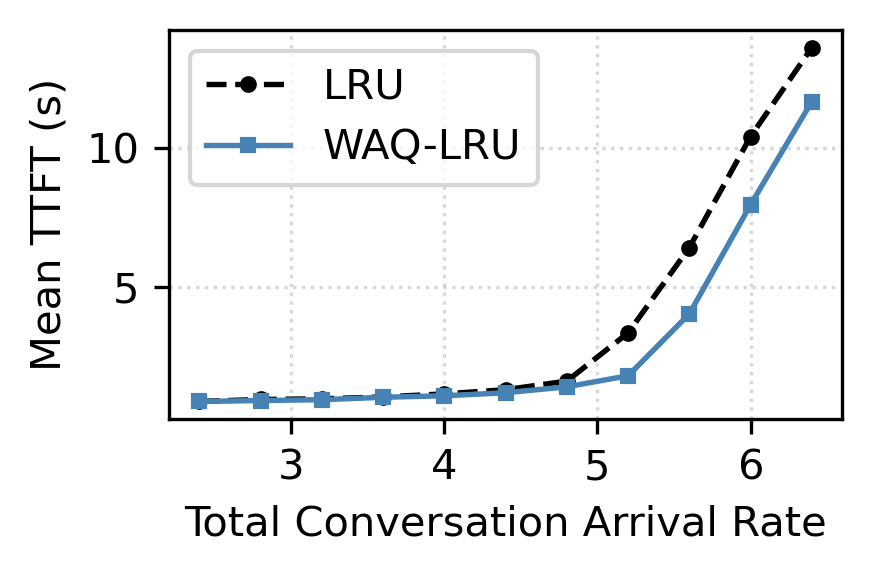}
        \caption{Mean TTFT.}
        \label{fig: A mean TTFT}
    \end{subfigure}
    \hfill
    \begin{subfigure}[t]{0.32\linewidth}
        \centering
        \includegraphics[width=\linewidth]{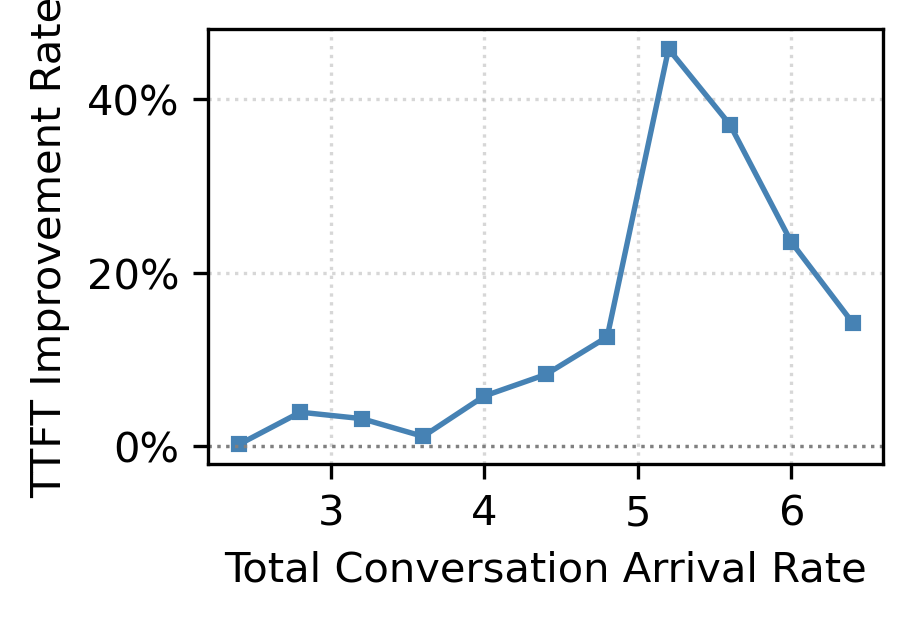}
        \caption{Mean TTFT improvement rate.}
        \label{fig: A TTFT improvement rate}
    \end{subfigure}
    \caption{Comparison of LRU and WAQ-LRU for Group A workflow mixture.}
    \label{fig: WAQ A additional results}
\end{figure}

As shown in \cref{fig: WAQ A additional results}, the hit ratios remain stable through $r=5.2$. 
In this range, the LRU hit ratio is around $26.5\%$, while the WAQ-LRU hit ratio is around $29.5\%$.
The corresponding relative improvement is approximately $10\%$.
The TTFT difference is small in the low-load regime ($r\leq 4.5$) and becomes substantial as the load increases.
Specifically, at $r=5.2$, WAQ-LRU reduces the mean TTFT from $3,368$ ms to $1,827$ ms, corresponding to an improvement rate of $45.7\%$. 
At higher arrival rates, the hit ratios begin to decline, but our WAQ-LRU continues to outperform LRU.

\subsubsection{WildChat Code/Data and Legal/Business/Finance Workflows.}
\label{app: WAQ group B}

In Group B, Workflow 1 contains 3,000 WildChat Code/Data conversations, and Workflow 2 contains 2,000 WildChat Legal/Business/Finance conversations. 
We set their conversation arrival rates to be $\widehat\lambda_1(r)=0.6r$ and $\widehat\lambda_2(r)=0.4r$, and their exponential target IATs have means $\widehat\mu_1(r)=1200/r$ s and $\widehat\mu_2(r)=300/r$ s. 
The practical MCM statistics are as follows:
\begin{align*}
\text{Workflow 1:}\quad
&\widehat a_1^{\mathrm{cum}}=23.0600,\quad
\widehat a_1^{\mathrm{reu,hash}}=13.0770,\quad
\widehat b_1={9.9830},\\
\text{Workflow 2:}\quad
&\widehat a_2^{\mathrm{cum}}=20.4945,\quad
\widehat a_2^{\mathrm{reu,hash}}=12.1200,\quad
\widehat b_2={8.3745}.
\end{align*}
The optimized quota for our WAQ-LRU policy is given by
\begin{equation*}
    (N_1^*,N_2^*)=(1,097,3,510),
    \qquad
    \widehat{\vec\alpha}^{\,*}=(0.2381,0.7619).
\end{equation*}
The estimated limiting hit ratios are $18.11\%$ under LRU and $21.90\%$ under WAQ-LRU. 
We then sweep $r$ from $2.4$ to $9.6$ conversations per second in increments of $0.4$.

\begin{figure}[ht]
    \centering
    \begin{subfigure}[t]{0.32\linewidth}
        \centering
        \includegraphics[width=\linewidth]{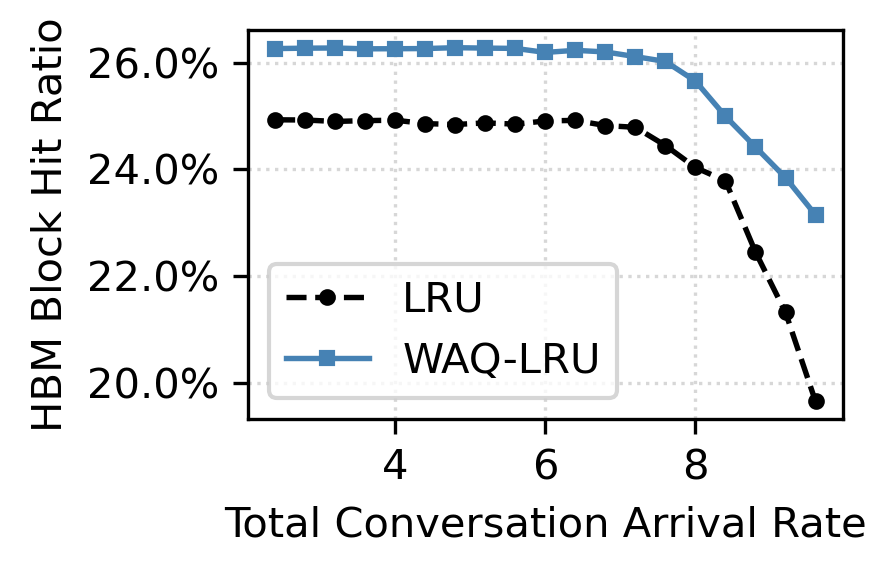}
        \caption{Measured HBM hit ratios.}
        \label{fig: B measured hit ratios}
    \end{subfigure}
    \hfill
    \begin{subfigure}[t]{0.32\linewidth}
        \centering
        \includegraphics[width=\linewidth]{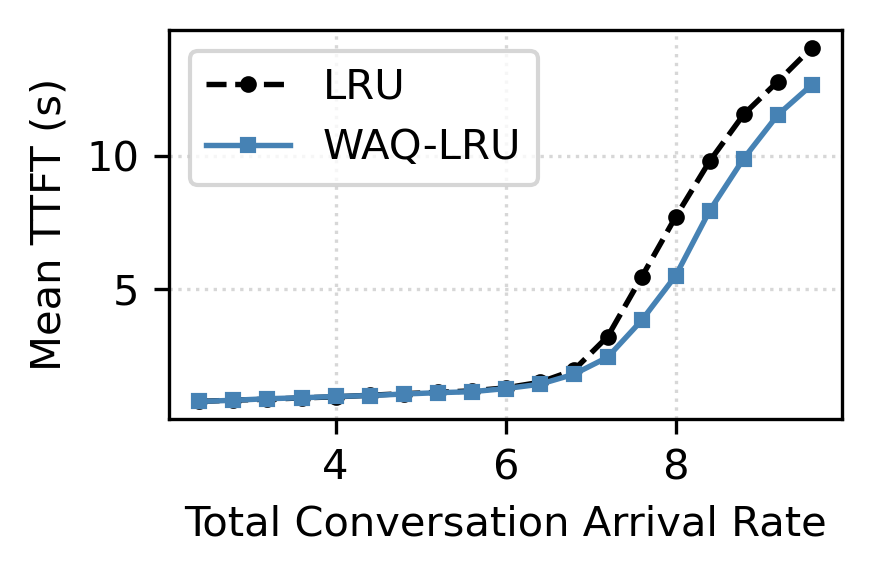}
        \caption{Mean TTFT.}
        \label{fig: B mean TTFT}
    \end{subfigure}
    \hfill
    \begin{subfigure}[t]{0.32\linewidth}
        \centering
        \includegraphics[width=\linewidth]{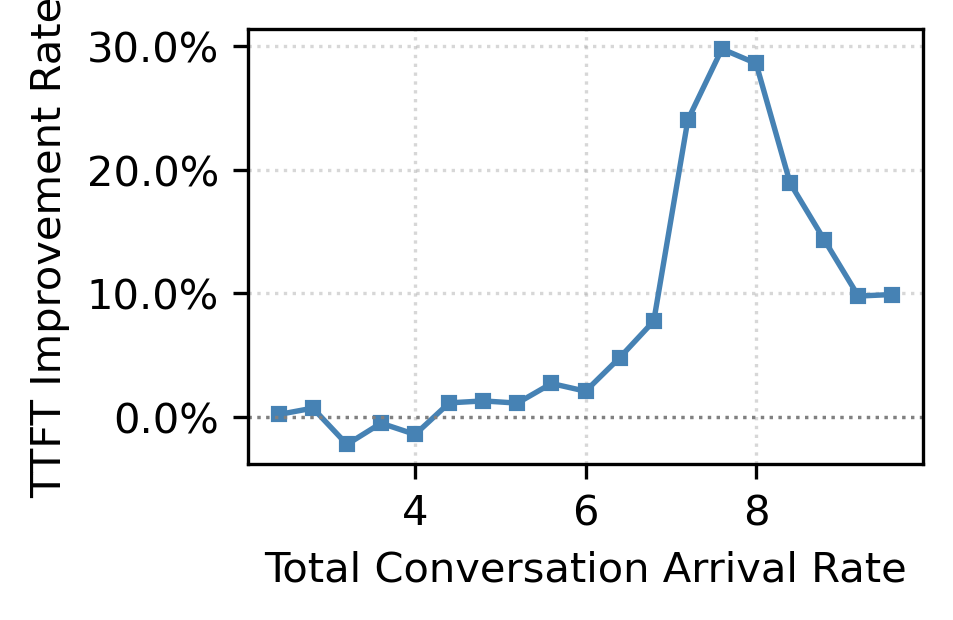}
        \caption{Mean TTFT improvement rate.}
        \label{fig: B TTFT improvement rate}
    \end{subfigure}
    \caption{Comparison of LRU and WAQ-LRU for Group B workflow mixture.}
    \label{fig: WAQ B additional results}
\end{figure}

\Cref{fig: WAQ B additional results} shows a smaller but consistent hit ratio advantage. 
Before entering the high load regime ($r\geq 7.2$), WAQ-LRU improves the hit ratio by approximately $5\%$--$6\%$. 
The TTFTs of the two policies are similar at low load.
Once the load is higher, the improvement becomes more pronounced. 
Specifically, at $r=7.2$, $7.6$, and $8.0$, WAQ-LRU reduces mean TTFT by $24.0\%$, $29.8\%$, and $28.6\%$, respectively.
At still higher loads, both hit ratios decline, while WAQ-LRU continues to maintain a higher hit ratio than LRU.

\end{document}